\documentclass[11pt,a4paper]{article}

\usepackage[T1]{fontenc}
\usepackage[utf8]{inputenc}
\usepackage{lmodern}
\usepackage[a4paper,margin=2.5cm]{geometry}
\usepackage{amsmath,amssymb,amsthm}
\usepackage{mathrsfs}
\usepackage{cite}
\usepackage{authblk}
\usepackage[
  colorlinks=true,
  linkcolor=blue,
  citecolor=blue,
  urlcolor=blue
]{hyperref}

\numberwithin{equation}{section}

\theoremstyle{plain}
\newtheorem{theorem}{Theorem}[section]
\newtheorem{lemma}[theorem]{Lemma}
\newtheorem{proposition}[theorem]{Proposition}

\theoremstyle{definition}

\theoremstyle{remark}
\newtheorem{remark}[theorem]{Remark}

\newcommand{\ii}{\mathrm{i}}
\newcommand{\dd}{\mathrm{d}}
\newcommand{\nn}{\nonumber}

\newcommand{\cyl}{\mathsf{Cyl}}

\title{Localization, Factorization and  Dualities for Elliptic Kernels}

\author[1,2]{Alessio Fontanarossa}
\author[1,2]{Fabrizio Nieri}
\author[1,2]{Antonio Pittelli}

\affil[1]{\emph{Dipartimento di Matematica, Universit\`a di Torino, Via Carlo Alberto 10, 10123 Torino, Italy}}
\affil[2]{\emph{INFN, Sezione di Torino, Via Pietro Giuria 1, 10125 Torino, Italy}}

\date{}

\begin{document}

\setcounter{secnumdepth}{3}

\maketitle

\begin{abstract}
We study the exact partition function
of 4d \(\mathcal N=1\) supersymmetric gauge theories on a torus times a cylinder \(\cyl=I\times S^1\), where \(I\) is a finite interval carrying two boundary components.  Each endpoint supports
an independent Dirichlet or Robin-like boundary polarization, so that the partition
function is a boundary-to-boundary elliptic kernel.  We construct the rigid supersymmetric
geometry, determine the BPS locus, and compute the chiral-multiplet 1-loop
determinants for the four possible boundary polarizations via equivariant localization.
The resulting elementary building blocks are theta functions dressed by
cubic  phases.  We then prove rank-changing Seiberg-type
dualities as identities of Jeffrey--Kirwan residues of these elliptic kernels.
We also discuss factorization into holomorphic-block cap wavefunctions represented by elliptic Gamma functions,
dimensional reductions to three and two dimensions, complete-intersection gauged
linear sigma models, and elliptic kernels for 4d \(\mathcal N=4\) super Yang--Mills and the
Klebanov--Witten theory, useful for holographic applications. 
\end{abstract}

\setcounter{tocdepth}{3}
\tableofcontents

\section{Introduction and Summary}

In supersymmetric and topological settings, exact amplitudes of quantum field theories on manifolds with boundaries naturally behave as kernels.  In two dimensions, this viewpoint is standard: interval amplitudes compute overlaps of boundary states, cap or hemisphere amplitudes assign wavefunctions to boundary conditions, and factorization follows from inserting a complete set of supersymmetric ground states. Exact results for 2d \( \mathcal N = (2,2) \) gauge theories on a disk geometry ($D^2$) provide a concrete realization of this rich structure \cite{Hori:2013ika,Honda:2013uca}. A related lesson is that exact answers on compact spaces \cite{Pestun:2007rz,Festuccia:2011ws,Dumitrescu:2012ha} often factorize into more elementary blocks attached to simpler pieces of the geometry, and gluing such blocks reconstructs indices and partition functions on closed backgrounds. In three dimensions, this far-reaching property has been widely explored for theories defined on manifolds admitting a genus-one Heegaard splitting, starting from \cite{Pasquetti:2011fj,Beem:2012mb}. In four dimensions, the natural playground is provided by geometries containing a two-torus ($T^2$), most notably the localization of 4d \(\mathcal N=1\) theories on \(D^2\times T^2\) and the associated four-dimensional holomorphic blocks \cite{Longhi:2019hdh,Nieri:2015yia}. 

Interval ($I$) geometries add a further layer: observables become kernels between two boundary sectors. 
 This is already visible in three-dimensional indices on \(I\times T^2\), where bulk multiplets couple to boundary \(2d\) degrees of freedom and the supersymmetric observables are related to elliptic genera, open-string Witten indices and interval factorizations \cite{Sugiyama:2020uqh,Yoshida:2014ssa,Zhao:2025ixe}. Related gluing, cigar, line-segment and slab constructions, including interfaces and boundary localization, have been developed in \cite{Dedushenko:2018aox,Dedushenko:2018tgx,Dedushenko:2021mds,Dedushenko:2022uay,Dedushenko:2023qjq,Litvinov:2022pjd}.
The purpose of this paper is to construct and analyze such kernels for 4d  \(\mathcal N=1\)  gauge theories on \(
\cyl\times T^2 \), where the cylinder \(\cyl=I \times S^1\) has two boundary components.  Throughout the paper we use the equivalent descriptions \(
\cyl\times T^2=(I\times S^1)\times T^2\simeq I\times T^3
\), depending on which aspect of the geometry we wish to emphasize. 
The title refers to elliptic kernels for a concrete reason: fixing two supersymmetric boundary conditions at the northern and southern ends, the interval direction makes the path integral a boundary-to-boundary amplitude, i.e. a transition matrix element between boundary sectors, while the torus makes this amplitude elliptic in the sense of the $\cyl\times pt/S^1/T^2$ hierarchy.\footnote{We use the word ``kernel'' in the boundary-to-boundary sense appropriate to an interval geometry. This should be distinguished from another common usage in the literature on supersymmetric partition functions and elliptic hypergeometric integrals, where compact-space indices depending on two sets of flavor fugacities can be interpreted, for instance, as integral or distributional kernels; see e.g. the elliptic hypergeometric identities of \cite{Spiridonov:2001ig,Rains:2003}, related recent work on elliptic kernels and duality walls \cite{Bottini:2021vms,Benvenuti:2024mpn}, and the holomorphic-block/bad-theory literature \cite{Pasquetti:2011fj,Beem:2012mb,Nieri:2015yia,Giacomelli:2023zkk,Giacomelli:2024laq}. In this paper we analyze only the elementary \(D/R\) boundary polarizations and their flip determinants. We leave for future work a  kernel with matrix elements labelled by a complete basis of boundary conditions or boundary defects.} Eventually, supersymmetric localization turns the matrix element into a finite-dimensional contour integral of a meromorphic function of gauge and flavor fugacities, whose elementary factors are theta functions or, equivalently, products of elliptic-Gamma cap wavefunctions.
Once the elliptic kernels are obtained, we study how they  behave under changes of boundary conditions/polarizations, anomaly cancellation conditions and compatibility with non-perturbative dualities. We also study how to recover the lower-dimensional results by taking suitable limits and the concatenation of several kernels, possibly coupled to cap wavefunctions. 
 
This paper sits at the intersection of three lines of work which naturally connect supersymmetric quantum field theory to geometry, representation theory and special functions. The first is localization on manifolds with boundary, where one must keep track not only of the bulk BPS locus but also of admissible boundary conditions, boundary terms and possible boundary degrees of freedom \cite{Hori:2013ika,Longhi:2019hdh,Sugiyama:2020uqh,Iannotti:2023jji,Dedushenko:2018aox,Dedushenko:2018tgx,Dedushenko:2022uay}.  The second is the theory of holomorphic blocks and their elliptic lifts, where wavefunctions are assigned to elementary pieces of a background and gluing is expressed by integration over gauge fugacities \cite{Pasquetti:2011fj,Beem:2012mb,Nieri:2015yia}.  The third is the study of dual boundary conditions and duality interfaces, where infrared equivalences of bulk theories act non-trivially on the space of boundary conditions and must be tested by anomaly matching and supersymmetric indices \cite{Dimofte:2017tpi}.  The present construction combines these ingredients in a setting with two independent boundaries in four dimensions and lowest supersymmetry.

\paragraph{Main results.}

The paper starts from the study of the rigid supersymmetric background \mbox{\(\cyl\times T^2\)}, making the Killing-spinor bilinears and the supercharge vector explicit. The first main result is the equivariant index evaluation of the 1-loop determinants for the four basic boundary polarizations of chiral multiplets, namely for the four possible choices of Dirichlet (D) or Robin(-like) (R) boundary conditions at the two ends. These are the same conditions considered in \cite{Longhi:2019hdh} and further studied in \cite{Dedushenko:2023cvd}. The mixed polarizations are acyclic away from the resonance hyperplanes of the boundary problem, meaning that the relevant first-order complex has no admissible kernel or cokernel modes and hence contributes the unit determinant
\begin{equation}
z_{\rm DR}(u)=z_{\rm RD}(u)=1 .
\label{eq:intro-zDR-RD-block}
\end{equation}
The unmixed polarizations give the two non-trivial elementary kernels, built out of theta or elliptic-Gamma functions dressed by \textit{cubic} polynomial phases in the relevant fugacities
\begin{equation}
z_{\rm DD}(u)
:=
\text{e}^{-\frac{\ii\pi}{3} P_3(-u)}\,
\Theta(-u;\omega)~,\quad z_{\rm RR}(u)
:=
\frac{
\text{e}^{\frac{\ii\pi}{3} P_3(u)}
}{
\Theta(u;\omega)
},
\end{equation}
the $\omega$ parameter representing the rotational fugacity of the cylinder. For $q$-Pochhammer symbols, theta and elliptic-Gamma functions we follow the  conventions of \cite{Felder:1999mq,Narukawa:2003,Friedman:2004}, as summarized in the appendix. In particular
\begin{equation}
P_3(u) :=B_{3,3}(u;1,\tau,\omega)+B_{3,3}(\omega-u;1,\tau,\omega) = -\frac{\tau+1}{2\omega\tau}\left(6u^2-6\omega u+\omega^2+\tau\right).
\label{eq:intro-theta-P3-def}
\end{equation}
is the relevant combination of \textit{cubic} Bernoulli polynomials, $\tau$ being the torus modular parameter. This is a distinguished feature with respect to the three-dimensional case, where the polynomial phases are quadratic. Another key difference is the diverse factorization into cap wavefunctions, namely
\begin{equation}
\Theta(u;\omega)=\frac{1}{\Gamma_e(u;\tau,\omega)\Gamma_e(\omega-u;\tau,\omega)}~.
\end{equation}

For interacting theories and general gauge group \(G\), the exact partition function can be written in a matrix-integral form, with 1-loop factors given  by boundary-polarized cylinder blocks.  Let \(u\) denote a gauge fugacity in the Cartan, \(\Delta_G\) the set of roots, and \(\rho\) the weights of a chiral representation.  For a chiral multiplet \(\Phi_I\) of R-charge \(r_I\), flavor fugacity \(\nu_I\), gauge representation \(\mathcal R_I\) and boundary polarization \(\mathsf q_I\in\{\mathrm{DD},\mathrm{DR},\mathrm{RD},\mathrm{RR}\}\), we define the full chiral 1-loop determinant by
\begin{align}
Z_{\Phi_I}^{\mathrm{DD}}(u)
&=
\prod_{\rho\in\mathcal R_I}
z_{\rm DD}\big(\rho(u)+\nu_I+(r_I-2)\gamma_R\big),
\nn\\
Z_{\Phi_I}^{\mathrm{RR}}(u)
&=
\prod_{\rho\in\mathcal R_I}
z_{\rm RR}\big(\rho(u)+\nu_I+r_I\gamma_R\big),
\nn\\
Z_{\Phi_I}^{\mathrm{DR}}(u)&=Z_{\Phi_I}^{\mathrm{RD}}(u)=1 .
\label{eq:intro-general-chiral-factors}
\end{align}
The non-zero-weight contribution of a vector multiplet with Neumann boundary conditions is obtained from the same first-order complex as a chiral multiplet with Dirichlet conditions in the adjoint representation with effective R-charge \(r=2\), after the usual ghost and Cartan zero-mode cancellations. This gives a 1-loop factor \(z_{\rm DD}(\alpha(u))\) for each root \(\alpha\in\Delta_G\)
\begin{equation}
Z_{\rm vec}^{G}(u)=\prod_{\alpha\in\Delta_G}z_{\rm DD}\big(\alpha(u)\big).
\label{eq:intro-general-vector-factor}
\end{equation}
The total localized expression is thus
\begin{equation}
\mathcal Z_{\cyl\times T^2}^{G,\{\mathsf q_I\}}
=
\frac{1}{|W_G|}
\oint_{\mathcal C_\eta}
\prod_{a=1}^{\operatorname{rk}G}\frac{\dd u_a}{2\pi\ii}\,
Z_{\rm vec}^{G}(u)
\prod_I Z_{\Phi_I}^{\mathsf q_I}(u).
\label{eq:intro-general-gauge-partition}
\end{equation}
Boundary degrees of freedom, when present, multiply the integrand by additional factors. In the main text we consider the flip determinants between RR and DD conditions only. The contour prescription is a key ingredient, and  we use the Jeffrey--Kirwan (JK) residue \cite{Nekrasov:2009uh,JeffreyKirwan1995,Benini:2013xpa,Benini:2013nda,Closset:2017zgf}.

The second main result is the boundary-flip calculus on a two-boundary geometry.  On a disk, a boundary multiplet can flip one boundary condition into another by multiplying the bulk determinant by a lower-dimensional determinant.  On a cylinder, the effect of flipping one end depends on what is imposed at the other end.  For this reason we keep track of all elementary flips at the northern and southern boundaries. The resulting flip factors are full theta-function determinants of the cylinder kernel. Elliptic-Gamma functions appear when the cylinder determinant is factorized into cap wavefunctions.  This separation between boundary determinants and cap wavefunctions is one of the conceptual points of the paper.

The third main result is a rank-changing duality of Seiberg type \cite{Seiberg:1994pq}. The electric theory has gauge group \(U(N_c)\), while the magnetic theory has gauge group \(U(N_f-N_c)\). The JK residues on the two sides can be matched term by term. Electric poles are labelled by subsets \(\mathcal I\) of \(N_c\) flavor labels, magnetic poles by complementary subsets \(\mathcal J=\mathcal I^c\) of cardinality \(N_f-N_c\). At the endpoint \(N_f=N_c\), the magnetic gauge group becomes trivial and the same identity reduces to a theta-function confining evaluation formula, which is new as far as we know. From the viewpoint of special functions, the resulting identities are close in spirit to elliptic hypergeometric residue formulae and elliptic beta-type evaluations \cite{Spiridonov:2001ig,Rains:2003}.  From the viewpoint of field theory, they are the cylinder analogues of integral  identities produced by supersymmetric dualities.

The fourth main result is the analysis of dimensional limits. In three dimensions, the cylinder blocks reduce to theta-function determinants factorizable into pairs of $q$-Pochhammer symbols, each representing a cap wavefunction. In two dimensions, they degenerate to trigonometric determinants factorizable into pairs of Euler-Gamma functions. The complete-intersection examples reproduce the expected equivalence of geometric and Landau--Ginzburg residue chambers, in the spirit of the GLSM phase structure of \cite{Witten:1993yc}.  

Last but not least, matter contents relevant to holographic SCFTs are analyzed in order to pave the way for some of the future research directions. For instance, writing \(u_{ab}:=u_a-u_b\), \(SU(N)\) \(\mathcal N=4\) admits the following cylinder partition function, obtained by taking two adjoint chirals in RR polarization and one in DD polarization,
\begin{align}
Z_{\mathcal N=4}^{SU(N)}(\Delta_1,\Delta_2,\Delta_3)
={}&
\frac{
\left[
 z_{\rm RR}(\Delta_1)\,
 z_{\rm RR}(\Delta_2)\,
 z_{\rm DD}(\Delta_3)
\right]^{N-1}}{N!}
\oint_{\mathcal C_{SU(N)}}
\prod_{a=1}^{N-1}\frac{\dd u_a}{2\pi\ii}
\nonumber\\
&\times
\prod_{\substack{a,b=1\\a\neq b}}^{N}
 z_{\rm DD}(u_{ab})\,
 z_{\rm RR}(u_{ab}+\Delta_1)\,
 z_{\rm RR}(u_{ab}+\Delta_2)\,
 z_{\rm DD}(u_{ab}+\Delta_3).
\label{eq:intro-N4-SU-partition}
\end{align}
with \(u_N=-\sum_{a=1}^{N-1}u_a\).  Here the  \(\Delta_I\) are the fugacities of the three adjoint chiral multiplets.  Equivalently, keeping the Bernoulli contact phase explicit, the same partition function can be written directly in terms of theta functions as
\begin{align}
Z_{\mathcal N=4}^{SU(N)}(\Delta_1,\Delta_2,\Delta_3)
={}&
\frac{1}{N!}
\oint_{\mathcal C_{SU(N)}}
\prod_{a=1}^{N-1}\frac{\dd u_a}{2\pi\ii}
\exp\!\left[
\frac{\ii\pi}{3}\,\mathcal P^{SU(N)}_{\mathcal N=4}(\Delta)
\right]
\nonumber\\
&\times
\left[
\frac{\Theta(-\Delta_3;\omega)}
{\Theta(\Delta_1;\omega)\Theta(\Delta_2;\omega)}
\right]^{N-1}
\prod_{\substack{a,b=1\\a\neq b}}^{N}
\frac{
\Theta(-u_{ab};\omega)\,
\Theta(-(u_{ab}+\Delta_3);\omega)
}{
\Theta(u_{ab}+\Delta_1;\omega)\,
\Theta(u_{ab}+\Delta_2;\omega)
}.
\label{eq:intro-N4-SU-theta-partition}
\end{align}
Using \eqref{eq:intro-theta-P3-def}, the Bernoulli contact phase in the theta-function presentation simplifies to the gauge-independent polynomial
\begin{align}
\mathcal P^{SU(N)}_{\mathcal N=4}(\Delta)
={}&
-\frac{(N-1)(\tau+1)(\omega^2+\tau)}{2\omega\tau}
\nonumber\\
&-
\frac{3(N^2-1)(\tau+1)}{\omega\tau}
\Bigl(
\Delta_1^2+\Delta_2^2-\Delta_3^2
-\omega(\Delta_1+\Delta_2+\Delta_3)
\Bigr).
\label{eq:intro-N4-SU-theta-phase}
\end{align}
All dependence on the gauge fugacities cancels in \eqref{eq:intro-N4-SU-theta-phase}; before the cancellation the possible linear terms are proportional to \(\sum_{a\neq b}u_{ab}=0\), while the quadratic \(u_{ab}\)-terms cancel root by root.  Thus \eqref{eq:intro-N4-SU-partition} is the compact block form, while \eqref{eq:intro-N4-SU-theta-partition} is the fully equivalent theta-function form.
The contour is again specified by the JK prescription. This finite-rank formula is one of the inputs for a possible large-\(N\) analysis of cylinder kernels and their comparison with supersymmetric gravitational backgrounds containing cylindrical factors, in the spirit of the twisted-index, spindle-index and black-hole microstate literature~\cite{Hosseini:2016tor,Hosseini:2019lkt,Hong:2018viz,Inglese:2023wky,Colombo:2024mts}.

\paragraph{Outlook.}

Several directions remain open.  A first one is the systematic inclusion of     boundary interactions, anomaly inflow~\cite{Callan:1984sa} and defects or interfaces~\cite{Dimofte:2017tpi}.  The two ends of the cylinder suggest a categorical interpretation: boundary polarizations should define objects, cylinder kernels should define morphisms or pairings, and boundary matter should act by composition, in analogy with the categorical description of branes and boundary conditions in supersymmetric boundary field theory, topological strings and Landau--Ginzburg models~\cite{Gaiotto:2008sa,Douglas:2000gi,Kapustin:2002bi,Herbst:2008jq}.

A second direction is a classification of cylinder dualities.  The Seiberg-like determinant-pair family studied here is deliberately conservative, because it is the class for which JK chambers and complementary residue matching can all be controlled explicitly.  The anomaly-screening equations allow more general matter assignments, including adjoint and flavor-adjoint variants.   

A third direction   may connect the present construction to derived categories of boundary conditions~\cite{Douglas:2000gi,Kapustin:2002bi,Herbst:2008jq}, analytic kernels for elliptic special functions~\cite{Felder:1999mq,Narukawa:2003}, and the tt* interpretation of block fusion~\cite{Cecotti:1991me,Cecotti:2013mba}.  The word ``fusion'' is meant here in the broad sewing/gluing sense familiar from conformal field theory and modular tensor category constructions~\cite{Moore:1988qv,Fuchs:2002cm}.  At the level established in this paper the concrete algebraic statements are the reflected inverse relation for the two unmixed blocks and the flip identities \eqref{eq:north-boundary-flips-theta}--\eqref{eq:south-boundary-flips-theta}; associativity of more general boundary-matter compositions is left for future work.

A fourth direction is the large-rank limit.  Since the observable is a boundary-to-boundary kernel rather than a closed-space index, its saddle-point interpretation should differ from standard compact-space formulas.  The  \(SU(N)\) \(\mathcal N=4\) contour integral \eqref{eq:intro-N4-SU-partition} is the most direct starting point. A holographic interpretation, if any, should be closer to a transition amplitude in an asymptotically locally \(AdS_5\) spacetime with two three-dimensional boundary components than to a closed-space partition function.  In a local decompactification limit \(T^3\simeq\mathbb R^3\), the relevant bulk geometry may admit an \(AdS_4\)-sliced description, namely a warped \(AdS_4\times I\) form.  This is the \(D=4\) instance of the general \(AdS_D\times I\) interval geometries studied recently in axio-dilaton gravity, while related defect-CFT constructions recover observables from warped \(AdS_p\times S^q\times I\) backgrounds~\cite{Faedo:2026defect,Dibitetto:2026axiodilaton,Dibitetto:2026adsnine}.  It would be interesting to understand whether the two boundary components give rise to two gravitational blocks~\cite{Hosseini:2019iad} and whether their gluing reproduces black-string, black-brane or accelerating-black-hole observables~\cite{Benini:2013cda} in backgrounds containing an elliptic factor, as happens in related compact computations~\cite{Hosseini:2016tor,Hosseini:2019lkt,Hong:2018viz,Inglese:2023wky,Colombo:2024mts}.

Finally, one can study degenerations of the cylinder.  If one boundary collapses, becomes asymptotic or is replaced by an orbifold point, the cylinder kernel should interpolate between the present two-boundary construction, elliptic disk blocks~\cite{Longhi:2019hdh}, twisted indices on hyperbolic manifolds~\cite{Pittelli:2018rpl,Iannotti:2023jji} and localization on orbifold geometries~\cite{Inglese:2023wky,Inglese:2023tyc,Pittelli:2024ugf}.  Understanding this interpolation may clarify which degrees of freedom are genuinely localized at a boundary or singular point and which are artifacts of a chosen polarization.

\paragraph{Outline.}

\hyperref[sec:geometry-loc]{Section~2} constructs the supersymmetric background by solving the conformal new minimal equations. The Killing-spinor bilinears are evaluated explicitly and used to identify the equivariant parameters entering localization.

\hyperref[sec:bps-loci]{Section~3} describes the BPS locus. The vector multiplet equations imply that a particular combination of boundary holonomies is constant along the interval; this constant forms the complex gauge fugacity together with the toric holonomies.

\hyperref[sec:oneloop]{Section~4} derives the 1-loop determinants using the cohomological approach.  The elementary flips of boundary conditions  at the two boundaries  are listed and the distinction between full theta-function boundary determinants and elliptic-Gamma cap wavefunctions explained.

\hyperref[sec:applications-examples]{Section~5} contains the main applications.  \hyperref[sec:duality]{Subsection~5.1} proves the Seiberg-like duality between \(U(N_c)\) and \(U(N_f-N_c)\), discusses the \(SU\) projection and records the theta-function confining identity at the \(U(0)\) endpoint.  \hyperref[sec:complete-intersections-cylinder]{Subsection~5.2} studies complete-intersection GLSMs on the two-dimensional cylinder.  \hyperref[sec:factorization-half-blocks]{Subsection~5.3} discusses factorization into cap wavefunctions, while \hyperref[sec:holographic-benchmarks]{Subsection~5.4} studies the consistency of cylinder integrals for matter contents familiar from holographic SCFTs.

\hyperref[app:notation-conventions]{Appendix~A} collects notation, boundary-condition labels, 1-loop determinant conventions, spinor and curvature conventions, Shintani--Barnes regularization, Jeffrey--Kirwan residue conventions.  \hyperref[app:anomalies]{Appendix~B} records the anomaly-screening equations used to test candidate boundary polarizations.  \hyperref[app:dimensional-limits]{Appendix~C} discusses dimensional limits to three and two dimensions, distinguishing literal Shintani--Barnes evaluations from contact-term-normalized blocks.

\paragraph{Acknowledgements.}

The work of FN was partially supported by the CRT Foundation under the grant 108399/2024.0434.

\section{Supersymmetric Geometry}\label{sec:geometry-loc}

\subsection{Background and domain}

Fix parameters
\[
0<\theta_N<\theta_S<\pi,\qquad \theta\in[\theta_N,\theta_S],\qquad
\varphi\sim\varphi+2\pi,\qquad x\sim x+2\pi,\qquad y\sim y+2\pi,
\]
and define
\begin{equation}
\tau:=\tau_1+\ii\tau_2\quad(\tau_2>0).
\label{eq:tau}
\end{equation}

Let $h(\theta)$ be a smooth real-valued function on $[\theta_N,\theta_S]$.
Throughout, $h(\theta)$ is assumed to be strictly positive on the closed interval:
\begin{equation}
h(\theta)>0\qquad \text{for all }\theta\in[\theta_N,\theta_S].
\label{eq:h-positive}
\end{equation}
This condition ensures that the $\varphi$-fiber has strictly positive proper length at both boundary circles; if $h$ vanishes at an endpoint, the $\varphi$-circle collapses and the geometry ceases to be a cylinder.

Let $f(\theta)$ be a smooth, real-valued and strictly positive  function on $[\theta_N,\theta_S]$. Usually,
\begin{equation}
\lim_{\theta\to0, \pi}f(\theta)\to b,
 \label{eq:fdef}
\end{equation}
but we will not need it as the closed interval $[\theta_N, \theta_S]$ is strictly contained into $(0, \pi)$.
Introduce the toric one-form
\begin{equation}
\Upsilon:=\dd\varphi+\Omega_3\,\dd x+\Omega_4\,\dd y .
\label{eq:Upsilon-def}
\end{equation}
The metric is taken to be
\begin{equation}
\dd s^2
=
L^2 f(\theta)^2\,\dd\theta^2
+
L^2 b^2 h(\theta)^2\,\Upsilon^2
+
L^2\beta^2\Big( (\dd x+\tau_1 \dd y)^2+\tau_2^2\,\dd y^2\Big).
\label{eq:metric}
\end{equation}
Introduce the complex combination
\begin{equation}
\omega := \Omega_3\,\tau-\Omega_4 .
\label{eq:omega}
\end{equation}

At fixed $(x,y)$ one has $\Upsilon=\dd\varphi$, and the $(\theta,\varphi)$-submanifold is a warped cylinder with boundary circles at
$\theta=\theta_N$ and $\theta=\theta_S$, of proper lengths $2\pi L b\,h(\theta_N)$ and $2\pi L b\,h(\theta_S)$, respectively.
The Riemannian density in coordinates $(\theta,\varphi,x,y)$ is
\begin{equation}
\sqrt{g}=L^4\,b\,\beta^2\,\tau_2\,f(\theta)\,h(\theta).
\label{eq:sqrtg}
\end{equation}
The positivity assumption \eqref{eq:h-positive} guarantees non-degeneracy of \eqref{eq:metric} on the closed cylinder.

Consider the orthonormal coframe
\begin{equation}
e^1 = L f(\theta)\,\dd\theta,\qquad
e^2 = L b\,h(\theta)\,\Upsilon,\qquad
e^3 = L\beta\,(\dd x+\tau_1 \dd y),\qquad
e^4 = L\beta\tau_2\,\dd y,
\label{eq:frame}
\end{equation}
so that $\dd s^2=\sum_{a=1}^4 (e^a)^2$.

Let $\omega_{ab}=-\omega_{ba}$ denote the Levi--Civita spin connection one-forms defined by the torsion-free Cartan equation
\begin{equation}
\dd e^a + {\omega^a}_{\ b}\wedge e^b = 0.
\label{eq:cartan}
\end{equation}
In the coframe \eqref{eq:frame},
\begin{equation}
\omega_{12}=-\omega_{21} = -\,\frac{b\,h'(\theta)}{f(\theta)}\,\Upsilon,
\label{eq:omega12}
\end{equation}
and $\omega_{ab}=0$ for all other pairs $(a,b)\neq(1,2)$. Indeed, one has $\dd e^1=\dd e^3=\dd e^4=0$ and, since $\dd\Upsilon=0$,
\[
\dd e^2
= L b\,h'(\theta)\,\dd\theta\wedge\Upsilon.
\]
The $a=2$ Cartan equation forces ${\omega^2}_{\ 3}={\omega^2}_{\ 4}=0$ and requires
${\omega^2}_{\ 1}=(b\,h'(\theta)/f(\theta))\,\Upsilon$, which implies \eqref{eq:omega12}.
The remaining Cartan equations are solved by setting all other components to zero.

Let $E_a$ be dual to $e^a$. From \eqref{eq:frame},
\begin{align}
E_1 &= \frac{1}{L f(\theta)}\,\partial_\theta, \nn\\
E_2 &= \frac{1}{L b\,h(\theta)}\,\partial_\varphi, \nn\\
E_3 &= \frac{1}{L\beta}\,(\partial_x-\Omega_3\partial_\varphi), \nn\\
E_4 &= \frac{1}{L\beta\tau_2}\Big(\partial_y-\tau_1\partial_x-(\Omega_4-\Omega_3\tau_1)\partial_\varphi\Big).
\label{eq:dual-frame}
\end{align}

We use the Euclidean spinor and sigma-matrix conventions collected in Appendix~\ref{app:notation-conventions}.  In particular, curved sigma matrices are obtained from the frame and its inverse, and the covariant derivatives on Weyl spinors are defined by \eqref{eq:spinor-cov-der}.

The conformal new minimal equations are~\cite{Sohnius:1981tp,Festuccia:2011ws,Dumitrescu:2012ha}
\begin{align}
(\nabla_\mu-\ii A_\mu^C)\zeta+\frac14\,\sigma_\mu\widetilde\sigma^\nu(\nabla_\nu-\ii A_\nu^C)\zeta&=0,
\label{eq:CKSE1}\\
(\nabla_\mu+\ii A_\mu^C)\widetilde\zeta+\frac14\,\widetilde\sigma_\mu\sigma^\nu(\nabla_\nu+\ii A_\nu^C)\widetilde\zeta&=0.
\label{eq:CKSE2}
\end{align}
Take
\begin{equation}
A^C=\frac12\Big(\omega_{12}+\alpha_\varphi\,\dd\varphi+\alpha_x\,\dd x+\alpha_y\,\dd y\Big),
\label{eq:AC}
\end{equation}
and define
\begin{equation}
\Xi := \alpha_\varphi \varphi+\alpha_x x+\alpha_y y .
\label{eq:spinor-phase}
\end{equation}

The choice \eqref{eq:AC} retains its form for general $h(\theta)$: the $\theta$-dependence of the Levi--Civita connection is entirely captured by $\omega_{12}$ in \eqref{eq:omega12}, and its inclusion in $A^C$ ensures that the conformal equations admit constant-frame spinors with purely toric phases.

A convenient solution is
\begin{equation}
\zeta_\alpha=\sqrt{\kappa_0}\,e^{\frac{\ii}{2}\Xi}\binom{1}{0}_\alpha,\qquad
\widetilde\zeta^{\dot\alpha}=\sqrt{\kappa_0}\,e^{-\frac{\ii}{2}\Xi}\binom{0}{1}^{\dot\alpha},
\label{eq:zetasol}
\end{equation}
with constant $\kappa_0\in\mathbb{C}$.  The pair \(\zeta,\widetilde\zeta\) has opposite chirality and opposite \(U(1)_R\) charge; the supercharge used below is the real combination whose square closes on the vector field \(K=\zeta\sigma\widetilde\zeta\), together with gauge and R-symmetry transformations. Periodicity along $T^2$ and the boundary of the cylinder implies
\begin{equation}
\alpha_\varphi, \alpha_x, \alpha_y\in\mathbb{Z}.
\label{eq:alphaint}
\end{equation}

The new minimal Killing spinor equations are~\cite{Sohnius:1981tp,Festuccia:2011ws,Dumitrescu:2012ha}
\begin{align}
(\nabla_\mu-\ii A_\mu)\zeta+\ii V_\mu\zeta+\ii V^\nu\sigma_{\mu\nu}\zeta&=0,
\label{eq:KSE1}\\
(\nabla_\mu+\ii A_\mu)\widetilde\zeta-\ii V_\mu\widetilde\zeta-\ii V^\nu\widetilde\sigma_{\mu\nu}\widetilde\zeta&=0,
\label{eq:KSE2}
\end{align}
and are solved by the spinors \eqref{eq:zetasol} provided the background fields
\begin{equation}
V=L\beta\,\kappa\,(\dd x+\tau\,\dd y),\qquad
A=A^C+\frac32 V,
\label{eq:AV}
\end{equation}
are chosen, where $\kappa\in\mathbb{C}$ is constant.

\subsection{Bilinears as tangent vectors and dual one-forms}

Introduce the bilinears as tangent vectors
\begin{equation}
K := K^\mu\partial_\mu,\qquad
Y := Y^\mu\partial_\mu,\qquad
\widetilde Y := \widetilde Y^{\,\mu}\partial_\mu,\qquad
\widetilde K := \widetilde K^{\,\mu}\partial_\mu.
\label{eq:vectors}
\end{equation}
We use the following normalization for the vector bilinears:
\begin{equation}
K^\mu := \zeta\sigma^\mu\widetilde\zeta,
\qquad
Y^\mu := \frac{1}{2|\widetilde\zeta|^2}\,\zeta\sigma^\mu\widetilde\zeta^\dagger,
\qquad
\widetilde Y^\mu := -\frac{1}{2|\zeta|^2}\,\zeta^\dagger\sigma^\mu\widetilde\zeta,
\qquad
\widetilde K^{\,\mu} := \frac{1}{4|\zeta|^2|\widetilde\zeta|^2}\,\widetilde\zeta^\dagger\,\widetilde\sigma^\mu\,\zeta^\dagger.
\label{eq:bilinears}
\end{equation}
The metric dual one-forms are denoted by the musical isomorphism
\begin{equation}
K^\flat := g_{\mu\nu}K^\nu\,\dd x^\mu,\qquad
Y^\flat := g_{\mu\nu}Y^\nu\,\dd x^\mu,\qquad
\widetilde Y^\flat := g_{\mu\nu}\widetilde Y^{\,\nu}\,\dd x^\mu,\qquad
\widetilde K^\flat := g_{\mu\nu}\widetilde K^{\,\nu}\,\dd x^\mu.
\label{eq:flat}
\end{equation}

Using $\sigma^\mu_{\alpha\dot\alpha}\,\widetilde\sigma_{\mu}^{\dot\beta\beta}
=-2\,\delta_\alpha^{\ \beta}\delta_{\dot\alpha}^{\ \dot\beta}$, one obtains
\begin{equation}
K^\mu\,\big(K_\mu\big)^\ast = 4\,|\zeta|^2\,|\widetilde\zeta|^2\,K^\mu\,\widetilde K_\mu,
\qquad
K^\mu \widetilde K_\mu=\frac12,
\label{eq:Ktilde-norm}
\end{equation}
and, consistently with \eqref{eq:bilinears},
\begin{equation}
\big(K^\mu\big)^\ast = 4\,|\zeta|^2\,|\widetilde\zeta|^2\,\widetilde K^{\,\mu},
\qquad
\big(Y^\mu\big)^\ast = -\,\frac{|\zeta|^2}{|\widetilde\zeta|^2}\,\widetilde Y^{\,\mu},
\qquad
\big(\widetilde Y^\mu\big)^\ast = -\,\frac{|\widetilde\zeta|^2}{|\zeta|^2}\,Y^\mu.
\label{eq:conjugation}
\end{equation}

Using \eqref{eq:dual-frame}, one finds
\begin{align}
K &= \frac{\ii\,\kappa_0}{L\beta\tau_2}\Big(\omega\,\partial_\varphi-\tau\,\partial_x+\partial_y\Big),
\label{eq:Kcoord}\\
Y &= \frac{1}{2L f(\theta)}\,\partial_\theta-\frac{\ii}{2 L b\,h(\theta)}\,\partial_\varphi,
\label{eq:Ycoord}\\
\widetilde Y &= -\,\frac{1}{2L f(\theta)}\,\partial_\theta-\frac{\ii}{2 L b\,h(\theta)}\,\partial_\varphi,
\label{eq:Ytilcoord}
\end{align}
and
\begin{equation}
\widetilde K = \frac{1}{4|\zeta|^2|\widetilde\zeta|^2}\,K^\ast
=\frac{1}{4|\kappa_0|^2}\,K^\ast.
\label{eq:Ktilcoord}
\end{equation}

The metric dual one-forms are
\begin{align}
K^\flat &= L\beta\,\kappa_0\,(\dd x+\tau\,\dd y),
\label{eq:Kform}\\
Y^\flat &= \frac{L f(\theta)}{2}\,\dd\theta-\frac{\ii\,L b\,h(\theta)}{2}\,\Upsilon, \nn\\
\widetilde Y^\flat &= -\,\frac{L f(\theta)}{2}\,\dd\theta-\frac{\ii\,L b\,h(\theta)}{2}\,\Upsilon,
\label{eq:forms}
\end{align}
and
\begin{equation}
\widetilde K^\flat = \frac{L\beta}{4\kappa_0}\,(\dd x+\overline\tau\,\dd y),
\qquad
\overline\tau=\tau_1-\ii\tau_2.
\label{eq:Ktilform}
\end{equation}

Since the components of $K^\flat$ are constant in the coordinate basis and the vector $K$ is orthogonal to the fibre direction, one finds
\begin{equation}
\nabla_\mu K_\nu = 0.
\label{eq:nablaK}
\end{equation}
In particular, $K$ is Killing.  The contraction with the toric one-form $\Upsilon$ is
\begin{equation}
\iota_K\Upsilon = 0.
\label{eq:iotaKfiber}
\end{equation}

\section{Localization and BPS loci}\label{sec:bps-loci}

The localization setup below uses the standard cohomological form of rigid supersymmetry on curved backgrounds, adapted to the presence of supersymmetric boundary conditions as in the disk and hyperbolic computations of Refs.~\cite{Festuccia:2011ws,Dumitrescu:2012ha,Festuccia:2020yff,Panerai:2020boq,Longhi:2019hdh,Iannotti:2023jji}.

Define $\varrho=\varrho(\theta)$ by
\begin{equation}
\frac{\dd\varrho}{\dd\theta}=\frac{f(\theta)}{b\,h(\theta)},
\qquad
\varrho(\theta):=\int_{\theta_*}^{\theta}\frac{f(\theta')}{b\,h(\theta')}\,\dd\theta'.
\label{eq:u}
\end{equation}

Under \eqref{eq:h-positive}, $\varrho(\theta)$ is smooth and strictly monotone if $f(\theta)$ stays positive on $[\theta_N,\theta_S]$, which holds for the present $f(\theta)$ in \eqref{eq:fdef}. Consequently, $\varrho$ provides a global coordinate along the cylinder direction and is adapted to the $(Y,\widetilde Y)$-complex.

Then
\[
E_1=\frac{1}{L b\,h(\theta)}\,\partial_\varrho,
\qquad
E_2=\frac{1}{L b\,h(\theta)}\,\partial_\varphi,
\]
so that (up to a nonzero scalar factor)
\begin{equation}
Y \ \propto\ \partial_\varrho+\ii\partial_\varphi,
\qquad
\widetilde Y \ \propto\ \partial_\varrho-\ii\partial_\varphi.
\label{eq:Ydir}
\end{equation}
Introduce
\begin{equation}
z:=e^{\varrho+\ii\varphi},\qquad \bar z:=e^{\varrho-\ii\varphi}.
\label{eq:z}
\end{equation}

For a field $X$ of $R$-charge $q_R$ and gauge charge $q_G$ define
\begin{equation}
D_\mu X := \nabla_\mu X - \ii q_R A_\mu X - \ii q_G \mathcal A_\mu X.
\label{eq:Dmu}
\end{equation}
Define
\begin{equation}
L_Y := Y^\mu D_\mu,
\qquad
L_{\widetilde Y} := \widetilde Y^{\,\mu} D_\mu.
\label{eq:LY}
\end{equation}

Along the $\varphi$-circle at fixed $(x,y)$ one has $\Upsilon=\dd\varphi$, hence
\begin{equation}
A_\varphi
=
\frac12\Big(\alpha_\varphi-\frac{b\,h'(\theta)}{f(\theta)}\Big),
\qquad
\mathcal A_\varphi = \mathcal A(\partial_\varphi).
\label{eq:Aphi}
\end{equation}
Using \eqref{eq:u}, one checks
\begin{equation}
\frac{b\,h'(\theta)}{f(\theta)}
=
\frac{\dd}{\dd\varrho}\big(\log h(\theta)\big).
\label{eq:cotidentity}
\end{equation}

A vector multiplet $(\mathcal A_\mu,\lambda,\widetilde\lambda,D)$ has supersymmetry variations
\begin{align}
\delta \mathcal A_\mu &= \ii\,\widetilde\zeta\,\widetilde\sigma_\mu\,\lambda+\ii\,\zeta\,\sigma_\mu\,\widetilde\lambda,
\nn\\
\delta\lambda &= \mathcal F_{\mu\nu}\sigma^{\mu\nu}\zeta+\ii\,D\,\zeta,
\nn\\
\delta\widetilde\lambda &= \mathcal F_{\mu\nu}\widetilde\sigma^{\mu\nu}\widetilde\zeta-\ii\,D\,\widetilde\zeta,
\nn\\
\delta D &= \widetilde\zeta\,\widetilde\sigma^\mu\Big(D_\mu\lambda+\frac{3\ii}{2}V_\mu\lambda\Big)
-\zeta\,\sigma^\mu\Big(D_\mu\widetilde\lambda-\frac{3\ii}{2}V_\mu\widetilde\lambda\Big),
\label{eq:vm-susy}
\end{align}
where $\mathcal F_{\mu\nu}=\partial_\mu \mathcal A_\nu-\partial_\nu \mathcal A_\mu$ in the Abelian case, and $D_\mu$ is the appropriate covariant derivative on the gaugini.

A chiral multiplet $(\phi,\psi,F)$ of $R$-charge $r$ has SUSY variations
\begin{align}
\delta \phi & = \sqrt{2}\,\zeta \psi,
\nn\\
\delta \psi & = \sqrt{2}\,F\,\zeta + \ii\sqrt{2}\,\sigma^\mu\widetilde{\zeta}\,D_\mu \phi,
\nn\\
\delta F & = \ii\sqrt{2}\,\widetilde{\zeta}\,\widetilde{\sigma}^\mu\Big(D_\mu\psi-\frac{\ii}{2}V_\mu\psi\Big)
-2\ii\,(\widetilde{\zeta}\,\widetilde{\lambda})\,\phi,
\label{eq:cm-susy}
\end{align}
and for the anti-chiral multiplet $(\widetilde{\phi},\widetilde{\psi},\widetilde{F})$,
\begin{align}
\delta \widetilde{\phi} & = \sqrt{2}\,\widetilde{\zeta}\,\widetilde{\psi},
\nn\\
\delta \widetilde{\psi} & = \sqrt{2}\,\widetilde{F}\,\widetilde{\zeta}
+ \ii\sqrt{2}\,\widetilde{\sigma}^\mu\zeta\,D_\mu\widetilde{\phi},
\nn\\
\delta \widetilde{F} & = \ii\sqrt{2}\,\zeta\,\sigma^\mu\Big(D_\mu\widetilde{\psi}+\frac{\ii}{2}V_\mu\widetilde{\psi}\Big)
+2\ii\,\widetilde{\phi}\,(\zeta\,\lambda).
\label{eq:acm-susy}
\end{align}

Assume an Abelian ansatz depending only on $\theta$,
\begin{equation}
\mathcal A
=
(\mathcal A_\varphi(\theta)+\beta_\varphi)\,\dd\varphi
+
(\mathcal A_x(\theta)+\beta_x)\,\dd x
+
(\mathcal A_y(\theta)+\beta_y)\,\dd y.
\label{eq:a-ansatz}
\end{equation}
Here $\beta_\varphi,\beta_x,\beta_y$ denote constant holonomies along the three circles
(not to be confused with the metric parameter $\beta$ in \eqref{eq:metric}).

A standard component of the vector-multiplet BPS equations in the cohomological presentation is
\begin{equation}
\iota_K\mathcal F=0,
\label{eq:iKF0}
\end{equation}
with $\mathcal F=\dd\mathcal A$ in the Abelian case. Under the assumption that the components in \eqref{eq:a-ansatz}
depend only on $\theta$, one has
\[
\mathcal F = \mathcal A_\varphi'(\theta)\,\dd\theta\wedge \dd\varphi
+ \mathcal A_x'(\theta)\,\dd\theta\wedge \dd x
+ \mathcal A_y'(\theta)\,\dd\theta\wedge \dd y,
\]
and using $\,\iota_K(\dd\theta)=0\,$ and $\iota_K(\dd\varphi)=K^\varphi$, $\iota_K(\dd x)=K^x$, $\iota_K(\dd y)=K^y$, one finds
\[
\iota_K\mathcal F
=
-\Big(K^\varphi \mathcal A_\varphi'(\theta)+K^x \mathcal A_x'(\theta)+K^y \mathcal A_y'(\theta)\Big)\,\dd\theta.
\]
With $K$ as in \eqref{eq:Kcoord}, the condition \eqref{eq:iKF0} is equivalent to the first-order ODE
\begin{equation}
\omega\,\mathcal A_\varphi'(\theta)-\tau\,\mathcal A_x'(\theta)+\mathcal A_y'(\theta)=0,
\label{eq:BPS-ODE}
\end{equation}
which integrates to
\begin{equation}
\omega\,\mathcal A_\varphi(\theta)-\tau\,\mathcal A_x(\theta)+\mathcal A_y(\theta)=\mathcal A_0,
\label{eq:a0-def}
\end{equation}
for some integration constant $\mathcal A_0\in\mathbb{C}$. Equivalently, using \eqref{eq:omega},
\begin{equation}
\mathcal A_y(\theta)=\tau\,\mathcal A_x(\theta)+\big(\Omega_4-\tau\Omega_3\big)\,\mathcal A_\varphi(\theta)+\mathcal A_0 .
\label{eq:ay-solution}
\end{equation}

Introduce the component along the fibre direction,
\begin{equation}
\mathcal A_\Upsilon(\theta):=\mathcal A_\varphi(\theta)+\Omega_3 \mathcal A_x(\theta)+\Omega_4 \mathcal A_y(\theta).
\label{eq:aUpsilon}
\end{equation}
Then $\dd\mathcal A$ contains $\mathcal A_\Upsilon'(\theta)\,\dd\theta\wedge\Upsilon$, and using \eqref{eq:frame},
\begin{equation}
\mathcal F_{12}
=
\frac{1}{L^2\,b\,f(\theta)\,h(\theta)}\,\frac{\dd \mathcal A_\Upsilon}{\dd\theta}.
\label{eq:F12}
\end{equation}
On the cohomological BPS locus, the auxiliary field is fixed accordingly:
\begin{equation}
D\big|_{\rm BPS}
=
\frac{1}{L^2\,b\,f(\theta)\,h(\theta)}\,\frac{\dd \mathcal A_\Upsilon}{\dd\theta}.
\label{eq:DBPS-trig}
\end{equation}

The chiral-multiplet BPS equations follow by setting to zero the fermionic variations in \eqref{eq:cm-susy} and \eqref{eq:acm-susy}.
On a contour where bosons and auxiliaries are taken to obey standard reality conditions, a consistent supersymmetric saddle is obtained by
\begin{equation}
\psi=\widetilde\psi=0,\qquad F=\widetilde F=0,
\label{eq:chiral-BPS-aux}
\end{equation}
together with the first-order constraints obtained from the kinetic terms in $\delta\psi$ and $\delta\widetilde\psi$.
In particular, the projected equations
\begin{equation}
\sigma^\mu\widetilde\zeta\,D_\mu\phi=0,
\qquad
\widetilde\sigma^\mu\zeta\,D_\mu\widetilde\phi=0
\label{eq:chiral-BPS-proj}
\end{equation}
can be expressed in the present background in terms of the first-order operators
$L_{\widetilde Y}$ and $L_Y$ introduced in \eqref{eq:LY}. Concretely, up to an overall nonzero scalar factor,
\begin{equation}
\sigma^\mu\widetilde\zeta\,D_\mu\phi=0\ \Longleftrightarrow\ L_{\widetilde Y}\phi=0,
\qquad
\widetilde\sigma^\mu\zeta\,D_\mu\widetilde\phi=0\ \Longleftrightarrow\ L_Y\widetilde\phi=0.
\label{eq:chiral-BPS-LY}
\end{equation}
In the localization computation, the one-loop determinant is obtained by expanding around the BPS locus
(and imposing supersymmetric boundary conditions on the two boundary circles).

Let $\mathcal{B}_2=\{(\theta,\varphi):\theta\in[\theta_N,\theta_S],\,\varphi\sim\varphi+2\pi\}$ at fixed $(x,y)$, with boundary circles
$S^1_S$ at $\theta=\theta_S$ and $S^1_N$ at $\theta=\theta_N$.
For any Abelian connection $A_{(0)}$ on $\mathcal{B}_2$,
\begin{equation}
\frac{1}{2\pi}\int_{\mathcal{B}_2} \dd A_{(0)}
=
\frac{1}{2\pi}\oint_{S^1_S} A_{(0)}
-
\frac{1}{2\pi}\oint_{S^1_N} A_{(0)}.
\label{eq:flux}
\end{equation}
For the vector multiplet connection \eqref{eq:a-ansatz},
\begin{equation}
\frac{1}{2\pi}\int_{\mathcal{B}_2} \dd\mathcal A
=
\mathcal A_\Upsilon(\theta_S)-\mathcal A_\Upsilon(\theta_N).
\label{eq:fluxgauge}
\end{equation}

\section{One-loop determinants and Boundary Polarizations}\label{sec:oneloop}

\subsection{Kernel equations, supercharge algebra, and adjointness}

For the spectral analysis, one expands fluctuations in Fourier modes along the three circles and treats the constant holonomies
$\beta_\varphi,\beta_x,\beta_y$ as background parameters. In this setting the dependence on the BPS profile
enters through the combination fixed by \eqref{eq:a0-def}, while the mode labels remain in $\mathbb{Z}^3$.

Consider Fourier modes on the three-torus,
\begin{equation}
X(\theta,\varphi,x,y)=e^{\ii(m_\varphi\varphi+m_x x+m_y y)}\,F(\theta),
\qquad (m_\varphi,m_x,m_y)\in\mathbb{Z}^3.
\label{eq:Fourier}
\end{equation}

Using \eqref{eq:Ydir} and \eqref{eq:Aphi}, the kernel equations reduce to first-order equations on $(\varrho,\varphi)$:
\begin{align}
L_Y X=0 \quad &\Longleftrightarrow\quad (\partial_\varrho+\ii D_\varphi)X=0,
\label{eq:LYeq}\\
L_{\widetilde Y} X=0 \quad &\Longleftrightarrow\quad (\partial_\varrho-\ii D_\varphi)X=0,
\label{eq:LYteq}
\end{align}
where $D_\varphi=\partial_\varphi-\ii q_R A_\varphi-\ii q_G \mathcal A_\varphi$ and $\mathcal A_\varphi=\mathcal A(\partial_\varphi)$.

For $L_Y X=0$, solutions can be written as
\begin{equation}
X^{(Y)}_{m_\varphi,m_x,m_y}
=
c^{(Y)}_{m_\varphi,m_x,m_y}\,
\exp\!\big(\ii(m_\varphi\varphi+m_x x+m_y y)\big)\,
\exp\!\Big(\big(m_\varphi-\tfrac{q_R}{2}\alpha_\varphi-q_G \mathcal A_\varphi\big)\,\varrho(\theta)\Big)\,
\big(h(\theta)\big)^{\frac{q_R}{2}} .
\label{eq:kerYmodes}
\end{equation}
For $L_{\widetilde Y} X=0$,
\begin{equation}
X^{(\widetilde Y)}_{m_\varphi,m_x,m_y}
=
c^{(\widetilde Y)}_{m_\varphi,m_x,m_y}\,
\exp\!\big(\ii(m_\varphi\varphi+m_x x+m_y y)\big)\,
\exp\!\Big(-\big(m_\varphi+\tfrac{q_R}{2}\alpha_\varphi+q_G \mathcal A_\varphi\big)\,\varrho(\theta)\Big)\,
\big(h(\theta)\big)^{-\frac{q_R}{2}} .
\label{eq:kerYtilmodes}
\end{equation}

On a field $X$ of charges $(q_R,q_G)$ one has
\begin{equation}
\delta^2 X
=
2\ii\Big(\mathcal{L}_K X - \ii q_R\,\Phi_R\,X - \ii q_G\,\Phi_G\,X\Big),
\label{eq:delta2schem}
\end{equation}
with $\Phi_R:=\iota_K A$ and $\Phi_G:=\iota_K\mathcal A$.
Using \eqref{eq:Kcoord} and \eqref{eq:iotaKfiber}, the spin-connection term in $A^C$ drops out of the contraction and one finds
\begin{equation}
\Phi_R
=
\iota_K A
=
\frac{\ii\,\kappa_0}{L\beta\tau_2}\,\gamma_R,
\qquad
\gamma_R:=\frac12\big(\omega \alpha_\varphi - \tau\alpha_x+\alpha_y\big).
\label{eq:gammaR}
\end{equation}

For the Abelian BPS connection \eqref{eq:a-ansatz}, the contraction with $K$ gives
\[
\Phi_G=\iota_K\mathcal A=\frac{\ii\,\kappa_0}{L\beta\tau_2}\Big(\omega(\mathcal A_\varphi(\theta)+\beta_\varphi)-\tau(\mathcal A_x(\theta)+\beta_x)+(\mathcal A_y(\theta)+\beta_y)\Big).
\]
On the BPS locus, the $\theta$-dependent part is constant by \eqref{eq:a0-def}, hence
\begin{equation}
\Phi_G
=
\iota_K\mathcal A
=
\frac{\ii\,\kappa_0}{L\beta\tau_2}\,\gamma_G,
\qquad
\gamma_G:=\mathcal A_0+\beta_y-\tau\,\beta_x+\omega\,\beta_\varphi .
\label{eq:gammaG}
\end{equation}
The expression \eqref{eq:gammaG} makes explicit the three integer shifts of the underlying flat connection:
\begin{equation}
\beta_y\mapsto\beta_y+1:\ \gamma_G\mapsto\gamma_G+1,
\qquad
\beta_\varphi\mapsto\beta_\varphi+1:\ \gamma_G\mapsto\gamma_G+\omega,
\qquad
\beta_x\mapsto\beta_x+1:\ \gamma_G\mapsto\gamma_G-\tau .
\label{eq:gammaG-shifts}
\end{equation}
The Shintani--Barnes representative used for the cylinder determinant below is chosen so that the final theta functions have nome \(e^{2\pi\ii\omega}\).  Hence the elementary cylinder blocks are meromorphic functions of the gauge fugacity on the \(\omega\)-elliptic lattice \(\mathbb Z+\omega\mathbb Z\).  All single-valuedness statements in the residue identities of this paper therefore refer to \(\omega\)-ellipticity in the gauge fugacities, not to invariance under an independent shift by \(\tau\).

On Fourier modes \eqref{eq:Fourier},
\[
\mathcal{L}_K X = \frac{\ii\,\kappa_0}{L\beta\tau_2}\,\big(\omega m_\varphi-\tau m_x+m_y\big)\,X,
\]
hence
\begin{equation}
\delta^2 X_{m_\varphi,m_x,m_y}
=
-\frac{2\ii\,\kappa_0}{L\beta\tau_2}\,
\Big(\omega m_\varphi-\tau m_x+m_y-q_G\gamma_G-q_R\gamma_R\Big)\,
X_{m_\varphi,m_x,m_y}.
\label{eq:delta2eigs}
\end{equation}
The common prefactor $-2\ii\kappa_0/(L\beta\tau_2)$ cancels in determinant ratios.

Let $\mathscr F$ be the $L^2$ completion of the relevant field space on the cylindrical region, with inner product
\begin{equation}
\langle X_1, X_2\rangle := \int \dd^4x\,\sqrt{g}\; (X_1)^\dagger X_2,
\label{eq:innerprod}
\end{equation}
and impose one of the supersymmetric boundary polarizations described below at each end of the interval.  These polarizations are chosen so that the boundary bilinear generated by integrating $L_Y$ by parts vanishes.  Equivalently, the localization deformation $\delta V$ has no uncancelled boundary variation on the allowed field space.  We therefore use the standard localization argument only within a fixed boundary polarization: deformations of the positive localizing functional that preserve the same boundary complex do not change the determinant ratio, while changing the polarization is a different observable and is implemented by the flip factors of Subsection~\ref{sec:oneloop}.  With these assumptions one has
\begin{equation}
L_Y^\dagger=-L_{\widetilde Y},
\label{eq:adjoint}
\end{equation}
and therefore
\begin{equation}
\operatorname{coker}(L_Y)\cong\ker(L_{\widetilde Y}),
\qquad
\operatorname{coker}(L_{\widetilde Y})\cong\ker(L_Y).
\label{eq:cokerker}
\end{equation}

\subsection{Index interpretation}\label{sec:index-interpretation}

The same algebraic data also has an index interpretation.  We take \(y=t_E\) to be the Euclidean time circle and use the Wick rotation \(t=\ii t_E\).  With the Lorentzian Hamiltonian convention \(\mathcal H=\ii\partial_t\), this gives \(\mathcal H=\partial_y\).  The path integral on \(\cyl\times T^2\) is then a graded trace over the Hilbert space of the theory on \(\cyl\times S^1_x\).  The two boundary polarizations are kept fixed throughout this quantization.  We use the convention, common in cohomological constructions of protected sectors, that symmetry generators act on local operators by commutators and hence by the corresponding differential operators; for instance the sphere construction of \cite{Dedushenko:2016jxl} writes the angular generator along the protected circle as \(P_\varphi=\ii\partial_\varphi\).  With the Fourier convention \eqref{eq:Fourier}, we therefore introduce
\begin{equation}
J:=\ii\partial_\varphi,
\qquad
\mathcal P:=\ii\partial_x,
\qquad
-\ii\partial_y\,X_{m_\varphi,m_x,m_y}=m_y X_{m_\varphi,m_x,m_y}.
\label{eq:index-generators}
\end{equation}
Thus a Fourier mode has eigenvalues \(J=-m_\varphi\) and \(\mathcal P=-m_x\).  The vanishing of \(\delta^2\) on a state contributing to the protected cohomology is not the separate vanishing of each charge, but the vanishing of the equivariant combination appearing in \eqref{eq:delta2eigs}.  Equivalently, for a state of charges \((q_R,q_G)\),
\begin{equation}
\delta^2 X_{\rm BPS}=0
\qquad\Longleftrightarrow\qquad
\omega m_\varphi-\tau m_x+m_y-q_G\gamma_G-q_R\gamma_R=0.
\label{eq:index-BPS-constraint-modes}
\end{equation}
In terms of the operators \eqref{eq:index-generators}, this becomes
\begin{equation}
-\ii\partial_y
=
\omega J-\tau\mathcal P+\gamma_G Q_G+\gamma_R Q_R
\qquad\hbox{on }\mathscr H_{\rm BPS},
\label{eq:index-BPS-operator-relation}
\end{equation}
where \(Q_G\) and \(Q_R\) denote the corresponding gauge, flavor, or R-symmetry charge operators.  Therefore the Euclidean translation around the \(y\)-circle can be rewritten on \(Q\)-cohomology as
\begin{equation}
\exp(-2\pi\partial_y)
=
\exp\!\Big[-2\pi\ii\big(\omega J-\tau\mathcal P+\gamma_G Q_G+\gamma_R Q_R\big)\Big]
\qquad\hbox{on }\mathscr H_{\rm BPS}.
\label{eq:index-evolution-rewritten}
\end{equation}
Introducing
\begin{equation}
q:=\exp(2\pi\ii\omega),
\qquad
p:=\exp(2\pi\ii\tau),
\label{eq:index-pq-def}
\end{equation}
the cylinder partition function may thus be viewed as the protected index
\begin{align}
Z_{\cyl\times T^2}
&=
\operatorname{Tr}_{\mathscr H(\cyl\times S^1_x)}
\left[(-1)^F\exp(-2\pi\mathcal H)\right]
=
\operatorname{Tr}_{\mathscr H(\cyl\times S^1_x)}
\left[(-1)^F\exp(-2\pi\partial_y)\right]
\nn\\
&=
\operatorname{Tr}_{\mathscr H_{\rm BPS}}
\left[
(-1)^F q^{-J}p^{\mathcal P}
\exp\!\big(-2\pi\ii(\gamma_G Q_G+\gamma_R Q_R)\big)
\right].
\label{eq:index-trace}
\end{align}
With the Wick-rotation convention used above, \(\mathcal H=\partial_y\) on the Euclidean functional integral.  The protected trace \eqref{eq:index-trace} is independent of this convention once the Euclidean algebra \eqref{eq:delta2eigs} is kept fixed.

\subsection{Boundary polarizations and one-loop determinants}

The boundary of $\mathcal{B}_2$ consists of two circles at $\theta=\theta_N$ and $\theta=\theta_S$.  In this paper a boundary polarization is simply the choice of supersymmetric boundary condition for the chiral multiplet at a given boundary circle.  We allow Dirichlet (D) and Robin (R) boundary conditions.  Since the cylinder has two ends, a bulk chiral multiplet is assigned an ordered pair of boundary conditions,
\[
DD,\qquad DR,\qquad RD,\qquad RR,
\]
where the first label refers to the northern boundary and the second to the southern boundary.

In the cohomological presentation, Robin boundary conditions are characterized by compatibility with the $\delta$-variation of the fermionic cohomological variable $B$. Schematically,
\begin{equation}
\delta B = L_{\widetilde Y}\phi + \cdots,
\label{eq:deltaBschem}
\end{equation}
and therefore Robin enforces
\begin{equation}
\big(L_{\widetilde Y}\phi\big)\big|_{\partial}=0.
\label{eq:RobinLytil}
\end{equation}
In particular, bulk modes in $\ker L_{\widetilde Y}$ satisfy \eqref{eq:RobinLytil} automatically.  Written in the radial coordinate \(\varrho\), this is a spectral Robin condition of the form \((\partial_\varrho-\ii D_\varphi)\phi|_\partial=0\), with the connection and charge-dependent terms contained in \(D_\varphi\); it is not an additional arbitrary parameter but the supersymmetric Robin operator paired with the Dirichlet condition \(\phi|_\partial=0\).

For mixed boundary choices, the cylinder supports no nontrivial cohomology for the first-order complex under generic background holonomies. This follows from the incompatibility of holomorphic/antiholomorphic transport data across the two boundary circles, as made explicit below.

Consider a field $X$ obeying the bulk equation $L_Y X=0$, hence of the form \eqref{eq:kerYmodes}. Acting on such a mode with $L_{\widetilde Y}$ yields
\begin{equation}
L_{\widetilde Y} X^{(Y)}_{m_\varphi,m_x,m_y}
=
2\,\big(m_\varphi-\tfrac{q_R}{2}\alpha_\varphi-q_G \mathcal A_\varphi\big)\,
\mathcal{W}(\theta)\,
e^{\ii(m_\varphi\varphi+m_x x+m_y y)},
\label{eq:LYtil-on-kerY}
\end{equation}
where $\mathcal{W}(\theta)$ is a nonzero weight obtained from the $\varrho(\theta)$ and $h(\theta)$ factors in \eqref{eq:kerYmodes}. In particular, at any boundary circle $\theta=\theta_\star$ the Robin condition $L_{\widetilde Y}X|_{\theta_\star}=0$ implies that each Fourier coefficient vanishes unless
\begin{equation}
m_\varphi-\tfrac{q_R}{2}\alpha_\varphi-q_G \mathcal A_\varphi=0.
\label{eq:quant-cond}
\end{equation}
For generic holonomies and background parameters, \eqref{eq:quant-cond} has no solution with $m_\varphi\in\mathbb{Z}$ and thus enforces $X\equiv 0$.  The exceptional loci defined by \eqref{eq:quant-cond} are resonance hyperplanes of the boundary problem.  In the gauge-theory residue formulae below we assume the flavor masses and spectator fugacities are generic enough that the JK poles selected by the unmixed RR factors do not lie on these resonance hyperplanes.  Under this genericity assumption the mixed-sector determinant is the meromorphic continuation of the identically trivial determinant on the complement of the resonance arrangement, hence it contributes the factor one at the selected poles.  If a resonance hyperplane is tuned to pass through a JK pole, extra boundary zero modes may appear and the corresponding wall-crossed boundary theory must be specified separately.

An analogous statement holds with $Y\leftrightarrow \widetilde Y$ interchanged. Consequently, for mixed boundary polarizations (DR or RD) one obtains, generically,
\begin{equation}
\ker L_Y=\ker L_{\widetilde Y}=0
\qquad\Longrightarrow\qquad
\operatorname{coker}(L_Y)=\operatorname{coker}(L_{\widetilde Y})=0.
\label{eq:acyclic-mixed}
\end{equation}
This provides the mechanism underlying the triviality of the bulk one-loop determinant for mixed polarizations.  In this precise sense the mixed boundary complexes are acyclic: the cohomology of the first-order complex is zero, because both the kernel and the cokernel vanish on the allowed space of modes.  Consequently the determinant ratio contains no non-trivial bulk factor and is equal to one, unless the parameters are tuned to a resonance hyperplane where extra boundary zero modes must be included explicitly.

After dropping the common prefactor in \eqref{eq:delta2eigs}, the universal $\delta^2$ eigenvalue factor on a Fourier mode labelled by
$(m_\varphi,m_x,m_y)\in\mathbb{Z}^3$ and carrying charges $(q_R,q_G)$ is
\begin{equation}
\lambda(m_\varphi,m_x,m_y;q_R,q_G)
:=
\omega m_\varphi-\tau m_x+m_y-q_G\gamma_G-q_R\gamma_R.
\label{eq:lambda-universal}
\end{equation}

For the Dirichlet polarization, the contributing cohomological scalar has effective charges
\begin{equation}
(q_R,q_G)=(r-2,\ q_G),
\label{eq:charges-D}
\end{equation}
hence the corresponding raw product is
\begin{align}
Z^{\mathrm{CM}}_{\mathrm{1\mbox{-}loop}}\big|_{D} 
& =
\prod_{(m_\varphi,m_x,m_y)\in\mathbb{Z}^3}
\Big(\omega m_\varphi-\tau m_x+m_y-q_G\gamma_G-(r-2)\gamma_R\Big).
\label{eq:ZrawD}
\end{align}

For the Robin polarization, compatibility with \eqref{eq:deltaBschem} enforces \eqref{eq:RobinLytil}. The relevant modes are those of the chiral scalar $\phi$ in $\ker L_{\widetilde Y}$, carrying charges
\begin{equation}
(q_R,q_G)=(r,\ q_G).
\label{eq:charges-R}
\end{equation}
With the placement of this sector in the denominator of the determinant ratio, the raw product reads
\begin{align}
Z^{\mathrm{CM}}_{\mathrm{1\mbox{-}loop}}\big|_{R} 
& =
\prod_{(m_\varphi,m_x,m_y)\in\mathbb{Z}^3}
\Big(\omega m_\varphi-\tau m_x+m_y+q_G\gamma_G+r\gamma_R\Big)^{-1}.
\label{eq:ZrawR}
\end{align}

In order to regularize these products, we employ a scheme similar to \cite{Barnes1904,Shintani1976}, namely we first consider the torus modes and then the cylinder ones by breaking the remaining product into two halves.  Geometrically this splitting is the one-loop shadow of cutting the cylinder into two oppositely oriented caps: after regularization, each half is naturally interpreted as a \(D^2\times T^2\) cap wavefunction, while the full \(\cyl\times T^2\) determinant is obtained by multiplying the two cap contributions and retaining the Bernoulli contact terms.  Thus the factorization used later is already visible at the level of the regulated triple product:
\begin{align}
Z^{\mathrm{CM}}_{\mathrm{1\mbox{-}loop}}\big|_{D} 
& =
\prod_{m_\varphi\geq 0}\prod_{(m_x,m_y)\in\mathbb{Z}^2}
(-1)\Big(\tau m_x-m_y+\omega m_\varphi-q_G\gamma_G-(r-2)\gamma_R\Big)\times\nn\\
& \times \Big(\tau m_x-m_y+\omega m_\varphi+\omega+q_G\gamma_G+(r-2)\gamma_R\Big)~,
\end{align}
and similarly 
\begin{align}
Z^{\mathrm{CM}}_{\mathrm{1\mbox{-}loop}}\big|_{R} 
& =
\prod_{m_\varphi\geq 0}\prod_{(m_x,m_y)\in\mathbb{Z}^2}
(-1)\Big(\tau m_x-m_y+\omega m_\varphi+q_G\gamma_G+r\gamma_R\Big)^{-1}\times\nn\\
& \times \Big(\tau m_x-m_y+\omega m_\varphi+\omega-q_G\gamma_G-r\gamma_R\Big)^{-1}~.
\end{align}
Then, using the standard Shintani-Barnes regularization as summarized in Subsection~\ref{app:Barnes}, such procedure gives us the final results (up to a normalization constant)
\begin{align}
Z^{\mathrm{CM}}_{\mathrm{1\mbox{-}loop}}\big|_{D} 
& = \frac{e^{-\ii\frac{\pi}{3}P_3(-q_G\gamma_G-(r-2)\gamma_R)}}{\Gamma_e(-q_G\gamma_G-(r-2)\gamma_R;\tau,\omega)\Gamma_e(\omega+q_G\gamma_G+(r-2)\gamma_R;\tau,\omega)} \nn\\
& =e^{-\ii\frac{\pi}{3}P_3(-q_G\gamma_G-(r-2)\gamma_R)} \Theta(-q_G\gamma_G-(r-2)\gamma_R;\omega)~,
\end{align}
and
\begin{align}
Z^{\mathrm{CM}}_{\mathrm{1\mbox{-}loop}}\big|_{R} 
& = e^{\ii\frac{\pi}{3}P_3(q_G\gamma_G+r\gamma_R)}\Gamma_e(q_G\gamma_G+r\gamma_R;\tau,\omega)\Gamma_e(\omega-q_G\gamma_G-r\gamma_R;\tau,\omega) \nn\\
& =\frac{e^{\ii\frac{\pi}{3}P_3(q_G\gamma_G+r\gamma_R)}}{\Theta(q_G\gamma_G+r\gamma_R;\omega)}~.\
\end{align}
With the conventions used above, the unmixed boundary polarizations define the two elementary four-dimensional cylinder blocks used in the applications below.  For a generic argument $u$, we set
\begin{equation}
z_{\rm DD}(u)
:=
\text{e}^{-\frac{\ii\pi}{3}P_3(-u)}
\Theta(-u;\omega)~,\quad 
z_{\rm RR}(u)
:=
\frac{
\text{e}^{\frac{\ii\pi}{3}P_3(u)}
}{\Theta(u;\omega)}.
\label{eq:4d-cylinder-blocks}
\end{equation}
With this normalization the DD and RR blocks are inverse after reflection of the argument:
\begin{equation}
z_{\rm DD}(u)\,z_{\rm RR}(-u)=1.
\label{eq:zDDzRR-inverse}
\end{equation}
Thus inserting a DD/RR pair with opposite arguments gives an exact inverse pair.  By contrast, the same-argument product is a holomorphic non-vanishing factor and introduces no polar hyperplane.
We also introduce
\begin{equation}
\Theta'_0:=\left.\frac{\partial}{\partial u}\Theta(u;\omega)\right|_{u=0},
\qquad
\varphi_0:=
\exp\!\left[\frac{\ii\pi}{3}P_3(0)\right].
\label{eq:theta-prime-varphi0}
\end{equation}
The constant $\varphi_0$ is the value of the cubic phase of the $RR$ block at $u=0$; no additional polynomial is being introduced.  Thus a simple positive-charge pole of $z_{\rm RR}(u)$ contributes $\varphi_0/\Theta'_0$ in the oriented JK convention used below.

The corresponding formula for a general gauge theory is obtained by multiplying these elementary factors over weights.  Let \(G\) be a compact gauge group with Weyl group \(W_G\), Cartan fugacities \(u=(u_1,\ldots,u_{\operatorname{rk}G})\), root system \(\Delta_G\), and chiral multiplets \(\Phi_I\) of R-charge \(r_I\), flavor fugacity \(\nu_I\), representation \(\mathcal R_I\), and boundary polarization \(\mathsf q_I\in\{\mathrm{DD},\mathrm{DR},\mathrm{RD},\mathrm{RR}\}\).  With the normalized contour measure used throughout this paper, the localized partition function is
\begin{equation}
\mathcal Z_{\cyl\times T^2}^{G,\{\mathsf q_I\}}
=
\frac{1}{|W_G|}
\oint_{\mathcal C_\eta}
\prod_{a=1}^{\operatorname{rk}G}\frac{\dd u_a}{2\pi\ii}\,
Z_{\rm vec}^{G}(u)\prod_I Z_{\Phi_I}^{\mathsf q_I}(u),
\label{eq:general-gauge-cylinder-partition}
\end{equation}
where
\begin{equation}
Z_{\rm vec}^{G}(u)=\prod_{\alpha\in\Delta_G}z_{\rm DD}\big(\alpha(u)\big)
\label{eq:general-vector-factor}
\end{equation}
and
\begin{align}
Z_{\Phi_I}^{\mathrm{DD}}(u)
&=
\prod_{\rho\in\mathcal R_I}
z_{\rm DD}\big(\rho(u)+\nu_I+(r_I-2)\gamma_R\big),
\nn\\
Z_{\Phi_I}^{\mathrm{RR}}(u)
&=
\prod_{\rho\in\mathcal R_I}
z_{\rm RR}\big(\rho(u)+\nu_I+r_I\gamma_R\big),
\nn\\
Z_{\Phi_I}^{\mathrm{DR}}(u)&=Z_{\Phi_I}^{\mathrm{RD}}(u)=1
\label{eq:general-chiral-factors}
\end{align}
for the bulk complex with no additional boundary multiplets.  The product over \(\rho\) is the product over the weights of the representation \(\mathcal R_I\); flavor fugacities are included by replacing \(\nu_I\) with the appropriate linear combination of flavor holonomies.  The vector factor \eqref{eq:general-vector-factor} is the non-zero-weight part of the adjoint one-loop determinant.  It can be obtained from the chiral determinant by taking the adjoint representation and the effective R-charge \(r=2\); the non-zero adjoint weights are the roots \(\alpha\in\Delta_G\), while the Cartan zero weights are cancelled by the ghost/zero-mode sector and leave the normalized Cartan measure in \eqref{eq:general-gauge-cylinder-partition}.  With our orientation convention this gives the DD root factor \(z_{\rm DD}(\alpha(u))\).  Using the reflected inverse relation one may equivalently trade it for inverse RR factors with reflected arguments, but the DD presentation is the one in which Vandermonde zeros remove coincident gauge eigenvalues most transparently.  Boundary multiplets are not included in \eqref{eq:general-chiral-factors}; they multiply the bulk kernel by the flip ratios described below.  The notation \(\mathcal C_\eta\) always denotes a JK cycle on the chosen \(\omega\)-elliptic fugacity torus.  Products of blocks which fail the \(\omega\)-ellipticity test should be regarded as local one-loop building blocks, not as standalone meromorphic kernels on this torus.

For mixed choices (DR or RD), the acyclicity statement \eqref{eq:acyclic-mixed} implies the absence of admissible kernel/cokernel modes in the bulk complex and hence a trivial elementary bulk block:
\begin{equation}
z_{\rm DR}(u)=z_{\rm RD}(u)=1,
\label{eq:Zmixed}
\end{equation}
up to possible finite-dimensional boundary degrees of freedom if included explicitly.  Equivalently, for a full chiral multiplet one has \(Z_{\Phi}^{\mathrm{DR}}=Z_{\Phi}^{\mathrm{RD}}=1\) in the minimal bulk complex.

The acyclicity statement \eqref{eq:Zmixed} refers to the bulk complex without additional boundary degrees of freedom.  As in the disk computation of boundary-condition flips \cite{Iannotti:2023jji,Dimofte:2017tpi,Longhi:2019hdh}, boundary multiplets can change the effective boundary polarization.  On the cylinder there are two independent boundary components, hence one can flip the northern and southern boundary conditions separately.

We denote the boundary circles at \(\theta=\theta_N\) and \(\theta=\theta_S\) by \(\partial_N\) and \(\partial_S\), respectively.  The first label in a pair such as \(DR\) refers to \(\partial_N\), while the second refers to \(\partial_S\).  We do not introduce separate symbols for the four bulk determinants: the non-trivial unmixed ones are the blocks \(z_{\rm RR}\) and \(z_{\rm DD}\), while the mixed elementary blocks are equal to one by \eqref{eq:Zmixed}.  A boundary multiplet which flips one end of the cylinder is represented by the ratio of the final determinant to the initial one.  These ratios are full boundary determinants, and therefore are written in terms of theta functions rather than elliptic Gamma cap wavefunctions.

Flipping the northern boundary gives four possibilities:
\begin{align}
Z^{\partial_N}_{RR\to DR}(u)
&:=\frac{1}{z_{\rm RR}(u)}
=\exp\!\left[-\frac{\ii\pi}{3}P_3(u)\right]\Theta(u;\omega),
\nn\\
Z^{\partial_N}_{DR\to RR}(u)
&:=z_{\rm RR}(u)
=\frac{\exp\!\left[+\frac{\ii\pi}{3}P_3(u)\right]}{\Theta(u;\omega)},
\nn\\
Z^{\partial_N}_{RD\to DD}(u)
&:=z_{\rm DD}(u)
=\exp\!\left[-\frac{\ii\pi}{3}P_3(-u)\right]\Theta(-u;\omega),
\nn\\
Z^{\partial_N}_{DD\to RD}(u)
&:=\frac{1}{z_{\rm DD}(u)}
=\frac{\exp\!\left[+\frac{\ii\pi}{3}P_3(-u)\right]}{\Theta(-u;\omega)}.
\label{eq:north-boundary-flips-theta}
\end{align}
Flipping the southern boundary gives the analogous four possibilities:
\begin{align}
Z^{\partial_S}_{RR\to RD}(u)
&:=\frac{1}{z_{\rm RR}(u)}
=\exp\!\left[-\frac{\ii\pi}{3}P_3(u)\right]\Theta(u;\omega),
\nn\\
Z^{\partial_S}_{RD\to RR}(u)
&:=z_{\rm RR}(u)
=\frac{\exp\!\left[+\frac{\ii\pi}{3}P_3(u)\right]}{\Theta(u;\omega)},
\nn\\
Z^{\partial_S}_{DR\to DD}(u)
&:=z_{\rm DD}(u)
=\exp\!\left[-\frac{\ii\pi}{3}P_3(-u)\right]\Theta(-u;\omega),
\nn\\
Z^{\partial_S}_{DD\to DR}(u)
&:=\frac{1}{z_{\rm DD}(u)}
=\frac{\exp\!\left[+\frac{\ii\pi}{3}P_3(-u)\right]}{\Theta(-u;\omega)}.
\label{eq:south-boundary-flips-theta}
\end{align}
These equations are the cylinder analogue of the boundary-condition flip: multiplying the full cylinder kernel by the appropriate boundary determinant changes one end of the cylinder from Robin to Dirichlet, or conversely.  The dependence on the other end is encoded in the initial and final labels, for instance \(RR\to DR\) and \(RD\to DD\) are different northern flips.  The flip operation is symmetric under exchanging Dirichlet and Robin: reversing an arrow replaces the determinant by its inverse, and exchanging the two boundary labels maps the mixed blocks \(z_{\rm DR}\) and \(z_{\rm RD}\) into each other.  This symmetry is consistent with the bulk acyclicity of both mixed sectors, but the two flips must still be distinguished because they live on different boundary components.  If instead one factorizes the cylinder into two caps, the corresponding one-sided ratios are elliptic-Gamma cap wavefunctions; those are useful for gluing, but they should not be confused with the full theta-function boundary determinants displayed above.

\section{Applications and Examples}\label{sec:applications-examples}

The cylinder blocks constructed above lead to four-dimensional $\cyl\times T^2$ residue identities of electric--magnetic type and to two-dimensional cylinder degenerations reproducing the chamber decomposition of GLSMs on $\cyl\simeq I\times S^1$.


\subsection{Dualities}\label{sec:duality}

The duality tests in this subsection are modeled on Seiberg-type rank-changing dualities~\cite{Seiberg:1994pq}, while the contour prescription is the Jeffrey--Kirwan prescription used in equivariant localization, elliptic genera and twisted indices~\cite{JeffreyKirwan1995,Benini:2013xpa,Benini:2013nda,Closset:2017zgf}. The selected matter content and boundary conditions follow from the anomaly cancellation conditions surveyed in the appendix. In this subsection, $\omega$-ellipticity refers to meromorphicity of the completed gauge integrand on the auxiliary gauge-fugacity torus \(u\sim u+1\sim u+\omega\). This is the torus relevant for the JK contour prescription, and should not be confused with the physical spacetime torus of modulus \(\tau\).

  \subsubsection{\texorpdfstring{\(U(N_c)\leftrightarrow U(N_f-N_c)\)}{U(Nc) <-> U(Nf-Nc)}}

\label{subsubsec:UNc-reinterpreted}

Set
\begin{equation}
n:=N_f-N_c\geq 0,
\end{equation}
and
\begin{equation}
M:=\sum_{i=1}^{N_f}m_i,
\qquad
\widetilde M:=\sum_{i=1}^{N_f}\widetilde m_i .
\label{eq:M-Mtilde}
\end{equation}
The electric theory has gauge group \(U(N_c)\), gauge fugacities \(u_a\), \(a=1,\ldots,N_c\), and integrand
\begin{equation}
\begin{aligned}
Z_{\rm el}(u)
=&\frac{1}{N_c!}\,
z_{\rm DD}\!\left[\sum_{a=1}^{N_c}u_a+M+\lambda\right]\,
z_{\rm DD}\!\left[-\sum_{a=1}^{N_c}u_a+\widetilde M+\lambda\right]
\\
&\times
\prod_{\substack{a,b=1\\ a\neq b}}^{N_c}
z_{\rm DD}[u_a-u_b]
\\
&\times
\prod_{a=1}^{N_c}\prod_{i=1}^{N_f}
z_{\rm RR}[-u_a-m_i]\,
z_{\rm RR}[u_a-\widetilde m_i]
\\
&\times
\prod_{a=1}^{N_c}\prod_{\alpha=1}^{n}
z_{\rm DD}[u_a+b_\alpha]\,
z_{\rm DD}[-u_a+\widetilde b_\alpha] .
\end{aligned}
\label{eq:Zel-UNc-reint}
\end{equation}
The corresponding electric partition function is the JK-contour integral of this meromorphic form
\begin{equation}
\mathcal Z_{\rm el}^{U(N_c)}
:=
\oint_{\mathcal C_\eta^{\rm el}}
\prod_{a=1}^{N_c}\frac{\dd u_a}{2\pi\ii}\,
Z_{\rm el}(u)
=
\operatorname{JK}_{U(N_c),\eta} Z_{\rm el} .
\label{eq:Zel-UNc-integral}
\end{equation}
The factor \(1/N_c!\) is already included in \eqref{eq:Zel-UNc-reint}.  The matter content consists of \(N_f\) fundamental RR chirals, \(N_f\) anti-fundamental RR chirals, \(n\) fundamental DD chirals, \(n\) anti-fundamental DD chirals, and two determinant-type DD chirals
\begin{equation}
S_+:\quad
z_{\rm DD}\!\left[\sum_a u_a+M+\lambda\right],
\qquad
S_-:\quad
z_{\rm DD}\!\left[-\sum_a u_a+\widetilde M+\lambda\right].
\end{equation}
Notice the \(\omega\)-ellipticity of \(Z_{\rm el}\) on the gauge fugacity space  follows from the quadratic condition \(n=N_f-N_c\) and from imposing the linear balancing condition
\begin{equation}
\sum_{\alpha=1}^{n}b_\alpha
=
\sum_{\alpha=1}^{n}\widetilde b_\alpha .
\label{eq:UNc-balancing-reint}
\end{equation}
Hence, when studying the possible contributing poles, we omit those given by integer shifts of the $\omega$-lattice. The determinant shifts are not independent masses: they are fixed by the total flavor fugacities in \eqref{eq:M-Mtilde}.  In checking the multiplier under a gauge shift \(u_a\mapsto u_a+\omega\), the dependence on the fundamental and anti-fundamental masses enters only through the total combinations \(M\) and \(\widetilde M\). These are precisely the combinations appearing in the two determinant shifts. Hence the mass-dependent part of the gauge \(\omega\)-elliptic multiplier cancels without imposing any balancing condition on the individual \(m_i\) or \(\widetilde m_i\). The individual masses remain, of course, genuine flavor fugacities of the integrand and of the JK residues.

We use the oriented JK prescription with the orientation fixed so that each simple pole of
\(z_{\rm RR}\) contributes \(\varphi_0/\Theta'_0\), cf.~\eqref{eq:theta-prime-varphi0}.  The charge covector of a polar hyperplane is the coefficient of the gauge fugacities in the argument of the corresponding RR block.  Thus a factor \(z_{\rm RR}[q\cdot u+m]\) contributes a polar hyperplane with charge covector \(q\).  Throughout this subsection we work in the chamber
\begin{equation}
\eta=(1,\ldots,1),
\label{eq:UNc-selected-chamber}
\end{equation}
which selects precisely the RR poles whose arguments have coefficient \(+1\) with respect to the relevant gauge fugacity.  On the electric side these are the factors \(z_{\rm RR}[u_a-\widetilde m_i]\).  The electric JK residue is therefore supported at
\begin{equation}
u_a=\widetilde m_i .
\end{equation}
The vector-multiplet Vandermonde factor removes coincident flavor labels; the contributing poles are indexed by subsets
\begin{equation}
\mathcal I\subset\{1,\ldots,N_f\},
\qquad
|\mathcal I|=N_c .
\end{equation}
Thus the electric JK residue is a finite sum of \(\binom{N_f}{N_c}\) non-vanishing sectors.  Equivalently, before quotienting by the Weyl group one may assign ordered flavor labels to the \(N_c\) gauge variables; the Vandermonde factor kills coincident labels and the factor \(1/N_c!\) removes the ordering.
With
\begin{equation}
S_{\mathcal I}:=\sum_{i\in\mathcal I}\widetilde m_i,
\end{equation}
Before writing the residue formula it is useful to spell out one finite example. Take \(N_f=5\) and \(N_c=2\), so that \(n=N_f-N_c=3\).  One electric residue sector is
\begin{equation}
\mathcal I=\{2,5\},
\qquad
u_1=\widetilde m_2,
\qquad
u_2=\widetilde m_5.
\label{eq:numerical-pole-example-electric}
\end{equation}
Equivalently, after relabelling the two electric gauge variables, this is the pole \(u_1=\widetilde m_2\), \(u_2=\widetilde m_5\).  The complementary magnetic sector is
\begin{equation}
\mathcal J=\mathcal I^c=\{1,3,4\},
\qquad
v_1=-\widetilde m_1,
\qquad
v_2=-\widetilde m_3,
\qquad
v_3=-\widetilde m_4,
\label{eq:numerical-pole-example-magnetic}
\end{equation}
up to Weyl permutations of the magnetic gauge variables.  The electric determinant factor depends on \(S_{\mathcal I}=\widetilde m_2+\widetilde m_5\), while the magnetic determinant factor depends on \(\sum_{j\in\mathcal J}\widetilde m_j\).  Since \(\mathcal I\sqcup\mathcal J=\{1,\ldots,5\}\), the two expressions are converted into each other by the total anti-fundamental mass \(\widetilde M\).  This is the concrete mechanism behind the complementary-residue matching proved below.
In general, the residue is
\begin{equation}
\operatorname{JK}_{U(N_c),\eta}Z_{\rm el}
=
\left(\frac{\varphi_0}{\Theta'_0}\right)^{N_c}
\sum_{\substack{\mathcal I\subset\{1,\ldots,N_f\}\\ |\mathcal I|=N_c}}
\mathcal E_{\mathcal I},
\label{eq:JK-el-UNc-reint}
\end{equation}
where
\begin{equation}
\begin{aligned}
\mathcal E_{\mathcal I}
=&\,
z_{\rm DD}[S_{\mathcal I}+M+\lambda]\,
z_{\rm DD}[-S_{\mathcal I}+\widetilde M+\lambda]
\\
&\times
\prod_{i\in\mathcal I}\prod_{\ell=1}^{N_f}
z_{\rm RR}[-\widetilde m_i-m_\ell]
\\
&\times
\prod_{i\in\mathcal I}\prod_{j\notin\mathcal I}
z_{\rm RR}[\widetilde m_i-\widetilde m_j]
\\
&\times
\prod_{i\in\mathcal I}\prod_{\alpha=1}^{n}
z_{\rm DD}[\widetilde m_i+b_\alpha]\,
z_{\rm DD}[-\widetilde m_i+\widetilde b_\alpha] .
\end{aligned}
\label{eq:electric-residue-UNc-reint}
\end{equation}

The magnetic theory has gauge group
\begin{equation}
U(n)=U(N_f-N_c),
\end{equation}
with fugacities \(v_r\), \(r=1,\ldots,n\).  In the chamber \eqref{eq:UNc-selected-chamber} we use the presentation
\begin{equation}
\begin{aligned}
Z_{\rm mag}(v)
=&\frac{1}{n!}\,
z_{\rm DD}\!\left[-\sum_{r=1}^{n}v_r+\lambda\right]\,
z_{\rm DD}\!\left[\sum_{r=1}^{n}v_r+M+\widetilde M+\lambda\right]
\\
&\times
\prod_{\substack{r,s=1\\ r\neq s}}^{n}
z_{\rm DD}[v_r-v_s]
\\
&\times
\prod_{r=1}^{n}\prod_{i=1}^{N_f}
z_{\rm RR}[v_r+\widetilde m_i]\,
z_{\rm DD}[-v_r+m_i]
\\
&\times
\prod_{r=1}^{n}\prod_{\alpha=1}^{n}
z_{\rm RR}[v_r-b_\alpha]\,
z_{\rm DD}[v_r-b_\alpha]\,
z_{\rm RR}[-v_r+b_\alpha]\,
z_{\rm RR}[-v_r-\widetilde b_\alpha] .
\end{aligned}
\label{eq:Zmag-UNc-reint-completed}
\end{equation}
The magnetic integrand is written in a completed selected-chamber presentation.  It is obtained from the uncompleted \(\omega\)-elliptic magnetic integrand by inserting the exact inverse pairs \(z_{\rm DD}[v_r-b_\alpha]z_{\rm RR}[-v_r+b_\alpha]=1\).  Therefore its \(\omega\)-elliptic multiplier is unchanged, and it is \(\omega\)-elliptic under the same conditions \eqref{eq:M-Mtilde} and \eqref{eq:UNc-balancing-reint} as the electric integrand.
The same-argument pair
\begin{equation}
z_{\rm RR}[v_r-b_\alpha] z_{\rm DD}[v_r-b_\alpha]
\label{eq:UNc-mag-regular-pair}
\end{equation}
is holomorphic and non-vanishing, and hence introduces no polar hyperplane.  The reflected pair satisfies
\begin{equation}
z_{\rm DD}[v_r-b_\alpha]z_{\rm RR}[-v_r+b_\alpha]=1 .
\label{eq:UNc-reint-identity-pair}
\end{equation}
These identities are used only to specify the polar presentation of the magnetic integrand.
The chamber \(\eta=(1,\ldots,1)\) selects, for each variable \(v_r\), the RR factors whose arguments have coefficient \(+1\) with respect to \(v_r\).  In the magnetic integrand the selected factors are
\begin{equation}
z_{\rm RR}[v_r+\widetilde m_i] ,
\label{eq:UNc-mag-selected-hyperplanes}
\end{equation}
and they generate the selected polar hyperplanes.  No further hyperplanes are selected.  The determinant factors, the vector factor, and the \(z_{\rm DD}[-v_r+m_i]\) factors are of DD type and do not produce poles.  The same-argument product in \eqref{eq:UNc-mag-regular-pair} is regular, while the remaining RR factors in the \(b_\alpha\)-sector have arguments \(-v_r+b_\alpha\) and \(-v_r-\widetilde b_\alpha\), hence coefficient \(-1\) with respect to \(v_r\).  They belong to the opposite JK chamber.  Thus the selected magnetic cell is generated only by the hyperplanes
\begin{equation}
v_r+\widetilde m_i=0 .
\end{equation}

The singlet sector is
\begin{equation}
Z_{\rm sing}=Z_M\,Z_B,
\end{equation}
where the meson sector is
\begin{equation}
Z_M=
\prod_{i,j=1}^{N_f}
z_{\rm RR}[-m_i-\widetilde m_j],
\label{eq:ZM-UNc-reint}
\end{equation}
and the boundary singlet sector is
\begin{equation}
Z_B=
\prod_{i=1}^{N_f}\prod_{\alpha=1}^{n}
z_{\rm DD}[\widetilde m_i+b_\alpha]\,
z_{\rm DD}[-\widetilde m_i+\widetilde b_\alpha].
\label{eq:ZB-UNc-reint}
\end{equation}
The magnetic gauge integral and the full magnetic partition function are
\begin{equation}
\mathcal Z_{\rm mag}^{U(n)}
:=
\oint_{\mathcal C_\eta^{\rm mag}}
\prod_{r=1}^{n}\frac{\dd v_r}{2\pi\ii}\,
Z_{\rm mag}(v)
=
\operatorname{JK}_{U(n),\eta}Z_{\rm mag},
\qquad
\mathcal Z_{\rm mag,full}^{U(n)}
:=Z_{\rm sing}\,\mathcal Z_{\rm mag}^{U(n)} .
\label{eq:Zmag-UNc-integral}
\end{equation}
The factors in \(Z_{\rm mag}(v)\) are cylinder one-loop determinants in the indicated polarizations.  The exact reflected DD/RR pairs inserted in \eqref{eq:Zmag-UNc-reint-completed} are not additional boundary multiplets; they are a change of polar presentation equal to one.  By contrast, the gauge-neutral sector \(Z_B\) can be interpreted as a boundary singlet sector implementing the required flip/cancellation of the spectator polarizations.  The same-argument product \eqref{eq:UNc-mag-regular-pair} is a holomorphic contact factor and should also not be counted as a pole-carrying boundary degree of freedom.  Genuine boundary multiplets enter through the flip ratios \eqref{eq:north-boundary-flips-theta}--\eqref{eq:south-boundary-flips-theta}.

\begin{lemma}[Selected JK poles and multiplicities]
\label{lem:selected-JK-poles-detpair}
Under the genericity assumptions of Theorem~\ref{thm:determinant-pair-cylinder-duality}, the chamber \(\eta=(1,\ldots,1)\) selects the following isolated poles, with unit multiplicity after the Weyl factors \(1/N_c!\) and \(1/n!\) are included.
On the electric side the selected poles are exactly
\begin{equation}
u_a=\widetilde m_i,
\qquad i\in\mathcal I,
\qquad |\mathcal I|=N_c,
\label{eq:lemma-electric-selected-poles}
\end{equation}
with distinct flavor labels.  On the magnetic side the selected poles are exactly
\begin{equation}
v_r=-\widetilde m_j,
\qquad j\in\mathcal J,
\qquad |\mathcal J|=n,
\label{eq:lemma-magnetic-selected-poles}
\end{equation}
again with distinct labels.  Thus the magnetic JK residue is a finite sum of \(\binom{N_f}{n}=\binom{N_f}{N_f-N_c}=\binom{N_f}{N_c}\) sectors.  The complement map \(\mathcal I\mapsto\mathcal J=\mathcal I^c\) is therefore a bijection between the electric and magnetic residue sectors.  The magnetic \(b_\alpha\)-hyperplanes and the same-argument contact pairs do not contribute in this chamber.  At the magnetic pole labelled by \(\mathcal J\), the vector factor cancels the selected RR factors with both flavor labels in \(\mathcal J\); the only off-diagonal RR factors left after this cancellation are those with one label in \(\mathcal J\) and the other in \(\mathcal J^c\).
\end{lemma}

\begin{proof}
A polar hyperplane is selected by the JK covector only if it is produced by an RR block whose argument has positive coefficient with respect to the corresponding gauge fugacity.  For the electric integrand \eqref{eq:Zel-UNc-reint}, the only such factors are \(z_{\rm RR}[u_a-\widetilde m_i]\).  The local charge matrix at a pole with one such hyperplane for each \(a\) is a permutation matrix and hence has determinant \(\pm1\), with the sign absorbed by the oriented JK convention.  The Weyl factor \(1/N_c!\) removes the ordering of the variables.  If two variables are assigned the same flavor label, the vector factor contains \(z_{\rm DD}[0]\), so the corresponding would-be residue vanishes.  Thus only subsets \(\mathcal I\) of cardinality \(N_c\) remain, each once.  For distinct labels, the RR factors whose flavor label is also selected are paired with the ordered vector factors by the reflected inverse relation, leaving precisely the factors with labels outside \(\mathcal I\), as in \eqref{eq:electric-residue-UNc-reint}.

For the magnetic integrand \eqref{eq:Zmag-UNc-reint-completed}, the same charge test selects \(z_{\rm RR}[v_r+\widetilde m_i]\) and gives \eqref{eq:lemma-magnetic-selected-poles}.  The determinant factors, the vector factor and all DD factors are holomorphic at the selected hyperplanes and therefore do not generate JK charges.  The RR factors \(z_{\rm RR}[-v_r+b_\alpha]\) and \(z_{\rm RR}[-v_r-\widetilde b_\alpha]\) have charge \(-1\) with respect to \(v_r\), hence lie in the opposite chamber; they can contribute only after crossing a wall, which is excluded by the genericity assumption.  The product \(z_{\rm RR}[v_r-b_\alpha]z_{\rm DD}[v_r-b_\alpha]\) is holomorphic and non-vanishing, so it defines no polar hyperplane.  As on the electric side, the local charge matrix is a permutation matrix, the Weyl factor removes the ordering, and coincident labels are killed by the Vandermonde zero.  For distinct labels, the vector factor cancels the RR factors with both labels in \(\mathcal J\), leaving only the mixed factors with labels in \(\mathcal J\) and \(\mathcal J^c\).  This proves both the list of selected poles and the unit multiplicity statement.
\end{proof}

\begin{lemma}[Theta-only complementary residues]
\label{lem:theta-only-complementary-residues}
The algebraic core of the determinant-pair duality is already visible after stripping the Bernoulli contact phases and the universal simple-pole factors.  In this paragraph write \(\Theta[x]\equiv\Theta(x;\omega)\).  For a selected electric subset \(\mathcal I\) and complementary magnetic subset \(\mathcal J=\mathcal I^c\), define the normalized theta residues
\begin{align}
\widehat{\mathcal E}_{\mathcal I}
={}&
\Theta[-S_{\mathcal I}-M-\lambda]
\Theta[S_{\mathcal I}-\widetilde M-\lambda]
\nonumber\\
&\times
\prod_{i\in\mathcal I}\prod_{\ell=1}^{N_f}
\frac{1}{\Theta[-\widetilde m_i-m_\ell]}
\prod_{i\in\mathcal I}\prod_{j\in\mathcal J}
\frac{1}{\Theta[\widetilde m_i-\widetilde m_j]}
\nonumber\\
&\times
\prod_{i\in\mathcal I}\prod_{\alpha=1}^{n}
\Theta[-\widetilde m_i-b_\alpha]\,
\Theta[\widetilde m_i-\widetilde b_\alpha],
\label{eq:theta-only-electric-sector-main}
\end{align}
and
\begin{align}
\widehat{\mathcal M}_{\mathcal J}
={}&
\Theta[-S_{\mathcal J}-\lambda]
\Theta[S_{\mathcal J}-M-\widetilde M-\lambda]
\nonumber\\
&\times
\prod_{j\in\mathcal J}\prod_{\ell=1}^{N_f}
\Theta[-\widetilde m_j-m_\ell]
\prod_{j\in\mathcal J}\prod_{i\in\mathcal I}
\frac{1}{\Theta[\widetilde m_i-\widetilde m_j]}
\nonumber\\
&\times
\prod_{j\in\mathcal J}\prod_{\alpha=1}^{n}
\frac{1}{\Theta[-\widetilde m_j-b_\alpha]\,
\Theta[\widetilde m_j-\widetilde b_\alpha]}.
\label{eq:theta-only-magnetic-sector-main}
\end{align}
Let
\begin{equation}
\widehat Z_{\rm sing}=
\prod_{r=1}^{N_f}\prod_{\ell=1}^{N_f}
\frac{1}{\Theta[-\widetilde m_r-m_\ell]}
\prod_{r=1}^{N_f}\prod_{\alpha=1}^{n}
\Theta[-\widetilde m_r-b_\alpha]\,
\Theta[\widetilde m_r-\widetilde b_\alpha].
\label{eq:theta-only-singlet-main}
\end{equation}
Then the residue matching is sectorwise:
\begin{equation}
\widehat{\mathcal E}_{\mathcal I}
=
\widehat Z_{\rm sing}\,
\widehat{\mathcal M}_{\mathcal J},
\qquad
\mathcal J=\mathcal I^c .
\label{eq:theta-only-sectorwise-main}
\end{equation}
\end{lemma}

\begin{proof}
The determinant factors agree because \(S_{\mathcal I}+S_{\mathcal J}=\widetilde M\):
\begin{equation}
-S_{\mathcal J}-\lambda=S_{\mathcal I}-\widetilde M-\lambda,
\qquad
S_{\mathcal J}-M-\widetilde M-\lambda=-S_{\mathcal I}-M-\lambda .
\end{equation}
The off-diagonal factors are also identical after reordering the finite product
\(\prod_{j\in\mathcal J}\prod_{i\in\mathcal I}=\prod_{i\in\mathcal I}\prod_{j\in\mathcal J}\).
It remains to check the meson and spectator factors.  For each flavor label set
\begin{equation}
X_r:=
\left(\prod_{\ell=1}^{N_f}\frac{1}{\Theta[-\widetilde m_r-m_\ell]}\right)
\left(\prod_{\alpha=1}^{n}\Theta[-\widetilde m_r-b_\alpha]\,
\Theta[\widetilde m_r-\widetilde b_\alpha]\right).
\end{equation}
The theta-only singlet factor is \(\widehat Z_{\rm sing}=\prod_{r=1}^{N_f}X_r\), while the meson/spectator part of the magnetic residue contributes \(\prod_{j\in\mathcal J}X_j^{-1}\). Hence
\begin{equation}
\widehat Z_{\rm sing}\prod_{j\in\mathcal J}X_j^{-1}
=
\left(\prod_{r=1}^{N_f}X_r\right)
\left(\prod_{j\in\mathcal J}X_j^{-1}\right)
=
\prod_{i\in\mathcal I}X_i,
\end{equation}
which is exactly the meson/spectator part of \(\widehat{\mathcal E}_{\mathcal I}\).  This proves the sector identity.
\end{proof}

\begin{remark}[Theta-only versus elliptic parent integrand]
Lemma~\ref{lem:theta-only-complementary-residues} shows that the complementary residue identity itself does not use the spectator balancing condition \eqref{eq:UNc-balancing-reint}; the spectator fugacities enter only through the flavor-factorized variables \(X_r\).  The complementary-residue identity is therefore independent of the elliptic properties of \(\Theta\): it is a finite algebraic identity.  The balancing condition is needed earlier for the parent Bernoulli-decorated cylinder integrand to be single-valued on the auxiliary gauge-fugacity torus and hence to define the JK contour problem in the \(\omega\)-elliptic presentation used here.
\end{remark}

\begin{theorem}[Determinant-pair cylinder duality]
\label{thm:determinant-pair-cylinder-duality}
Let \(N_f>N_c\), set \(n=N_f-N_c\), and consider the electric and magnetic meromorphic forms \(Z_{\rm el}\) and \(Z_{\rm mag}\) defined in \eqref{eq:Zel-UNc-reint} and \eqref{eq:Zmag-UNc-reint-completed}, together with the singlet factor \(Z_{\rm sing}=Z_MZ_B\) in \eqref{eq:ZM-UNc-reint}--\eqref{eq:ZB-UNc-reint}.  Assume the determinant shifts are fixed by \eqref{eq:M-Mtilde}, the spectator fugacities obey \eqref{eq:UNc-balancing-reint}, and all flavor and spectator parameters are generic: the selected JK intersections are simple, no non-selected zero or pole collides with them, no resonance hyperplane of the mixed boundary problem passes through them, and no pole lies on the boundary of the chosen \(\omega\)-elliptic fundamental domain.  In the chamber \(\eta=(1,\ldots,1)\), with the pole normalization fixed by \eqref{eq:theta-prime-varphi0}, one has
\begin{equation}
\mathcal Z_{\rm el}^{U(N_c)}
=
\left(\frac{\varphi_0}{\Theta'_0}\right)^{N_c-n}
\mathcal Z_{\rm mag,full}^{U(n)} .
\label{eq:thm-det-pair-cylinder-duality}
\end{equation}
Equivalently, after removing the universal \(\tau,\omega\)-dependent normalization of simple poles,
\begin{equation}
\mathcal Z_{\rm el}^{U(N_c)}
\doteq
\mathcal Z_{\rm mag,full}^{U(N_f-N_c)} .
\end{equation}
\end{theorem}

\begin{proof}
By Lemma~\ref{lem:selected-JK-poles-detpair}, the selected magnetic poles are indexed by subsets
\begin{equation}
\mathcal J\subset\{1,\ldots,N_f\},
\qquad
|\mathcal J|=n,
\end{equation}
with unit multiplicity in the normalized JK measure.  We set
\begin{equation}
\mathcal I=\mathcal J^c,
\qquad
|\mathcal I|=N_c .
\end{equation}

At \(v_r=-\widetilde m_j\), \(j\in\mathcal J\), the determinant-type factors give
\begin{equation}
z_{\rm DD}\!\left[-\sum_r v_r+\lambda\right]
\longrightarrow
z_{\rm DD}\!\left[\lambda+\sum_{j\in\mathcal J}\widetilde m_j\right]
=
z_{\rm DD}[-S_{\mathcal I}+\widetilde M+\lambda],
\end{equation}
and
\begin{equation}
z_{\rm DD}\!\left[\sum_r v_r+M+\widetilde M+\lambda\right]
\longrightarrow
z_{\rm DD}\!\left[\lambda+M+\widetilde M-\sum_{j\in\mathcal J}\widetilde m_j\right]
=
z_{\rm DD}[S_{\mathcal I}+M+\lambda].
\end{equation}

For \(j\in\mathcal J\), the meson factors cancel the magnetic DD chirals:
\begin{equation}
z_{\rm DD}[-v_r+m_\ell]\big|_{v_r=-\widetilde m_j}
=
z_{\rm DD}[\widetilde m_j+m_\ell],
\end{equation}
and hence
\begin{equation}
z_{\rm RR}[-m_\ell-\widetilde m_j]\,
z_{\rm DD}[\widetilde m_j+m_\ell]=1,
\qquad j\in\mathcal J .
\end{equation}
The remaining meson factors are
\begin{equation}
\prod_{i\in\mathcal I}\prod_{\ell=1}^{N_f}
z_{\rm RR}[-\widetilde m_i-m_\ell].
\label{eq:UNc-reint-meson-leftover}
\end{equation}

Similarly,
\begin{equation}
z_{\rm RR}[-v_r-\widetilde b_\alpha]\big|_{v_r=-\widetilde m_j}
=
z_{\rm RR}[\widetilde m_j-\widetilde b_\alpha],
\end{equation}
and
\begin{equation}
z_{\rm DD}[-\widetilde m_j+\widetilde b_\alpha]\,
z_{\rm RR}[\widetilde m_j-\widetilde b_\alpha]=1,
\qquad j\in\mathcal J .
\end{equation}
For the \(b_\alpha\)-sector set
\begin{equation}
y_{j\alpha}:=\widetilde m_j+b_\alpha .
\end{equation}
At \(v_r=-\widetilde m_j\), the magnetic factors are
\begin{equation}
z_{\rm RR}[v_r-b_\alpha]\,
z_{\rm DD}[v_r-b_\alpha]\,
z_{\rm RR}[-v_r+b_\alpha]
=
z_{\rm RR}[-y_{j\alpha}]\,
z_{\rm DD}[-y_{j\alpha}]\,
z_{\rm RR}[y_{j\alpha}] .
\end{equation}
Using the reflected inverse relation \(z_{\rm DD}(x)z_{\rm RR}(-x)=1\),
\begin{equation}
z_{\rm DD}[-y_{j\alpha}]\,
z_{\rm RR}[y_{j\alpha}]=1,
\end{equation}
and
\begin{equation}
z_{\rm DD}[y_{j\alpha}]\,
z_{\rm RR}[-y_{j\alpha}]=1 .
\end{equation}
After multiplication by \(Z_B\), the surviving \(b_\alpha\)-sector is
\begin{equation}
\prod_{i\in\mathcal I}\prod_{\alpha=1}^{n}
z_{\rm DD}[\widetilde m_i+b_\alpha]\,
z_{\rm DD}[-\widetilde m_i+\widetilde b_\alpha].
\label{eq:UNc-reint-b-leftover}
\end{equation}

By Lemma~\ref{lem:selected-JK-poles-detpair}, the magnetic vector multiplet cancels the \(i,j\in\mathcal J\), \(i\neq j\), RR terms.  The remaining magnetic RR chirals give
\begin{equation}
z_{\rm RR}[v_r+\widetilde m_i]
\longrightarrow
z_{\rm RR}[\widetilde m_i-\widetilde m_j] .
\end{equation}
The remaining factor is
\begin{equation}
\prod_{i\in\mathcal I}\prod_{j\in\mathcal J}
z_{\rm RR}[\widetilde m_i-\widetilde m_j],
\label{eq:UNc-reint-offdiag-leftover}
\end{equation}
Combining \eqref{eq:UNc-reint-meson-leftover}, \eqref{eq:UNc-reint-b-leftover}, and \eqref{eq:UNc-reint-offdiag-leftover} gives
\begin{equation}
Z_{\rm sing}\,\mathcal M_{\mathcal J}
=
\mathcal E_{\mathcal J^c},
\label{eq:UNc-reint-termwise}
\end{equation}
where \(\mathcal M_{\mathcal J}\) denotes the magnetic residue of \eqref{eq:Zmag-UNc-reint-completed}. Summing over \(\mathcal J\) gives
\begin{equation}
\operatorname{JK}_{U(N_c),\eta}Z_{\rm el}
=
\left(\frac{\varphi_0}{\Theta'_0}\right)^{N_c-n}
Z_{\rm sing}\,
\operatorname{JK}_{U(n),\eta}Z_{\rm mag}.
\label{eq:UNc-reint-final}
\end{equation}
Equivalently, up to an overall normalization depending only on \(\tau\) and \(\omega\),
\begin{equation}
\operatorname{JK}_{U(N_c),\eta}Z_{\rm el}
\doteq
Z_{\rm sing}\,
\operatorname{JK}_{U(N_f-N_c),\eta}Z_{\rm mag}.
\end{equation}
This is \eqref{eq:thm-det-pair-cylinder-duality} written in the shorthand JK notation.
\end{proof}

For \(n=0\), equivalently \(N_f=N_c\), the magnetic gauge group is \(U(0)\).  The rank-changing identity then becomes a finite-dimensional theta-function evaluation.  We write this endpoint explicitly, without using the shorthand cylinder blocks, after the \(SU\) projection in Subsection~\ref{subsubsec:theta-U0-identity}.

For \(n>0\), the dual is the magnetic \(U(N_f-N_c)\) theory with singlet sector \(Z_MZ_B\), evaluated in the JK presentation \eqref{eq:Zmag-UNc-reint-completed}.

\subsubsection{\texorpdfstring{\(SU(N_c)\leftrightarrow SU(N_f-N_c)\) projection}{SU(Nc) <-> SU(Nf-Nc) projection}}
\label{subsubsec:SU-projection}

The projection from \(U(N_c)\) and \(U(n)\) to \(SU(N_c)\) and \(SU(n)\) is implemented by inserting periodic Dirac delta factors imposing the vanishing of the corresponding center fugacity. On the electric side one inserts
\begin{equation}
\delta_T\!\left(\sum_{a=1}^{N_c}u_a\right)
\end{equation}
inside the electric integral, and on the magnetic side
\begin{equation}
\delta_T\!\left(\sum_{r=1}^{n}v_r\right)
\end{equation}
inside the magnetic integral. These delta factors introduce no additional meromorphic poles and are evaluated on the JK poles.  Equivalently, the projected partition functions are
\begin{align}
\mathcal Z_{\rm el}^{SU(N_c)}
&:=
\oint_{\mathcal C_\eta^{\rm el}}
\prod_{a=1}^{N_c}\frac{\dd u_a}{2\pi\ii}\,
\delta_T\!\left(\sum_{a=1}^{N_c}u_a\right)Z_{\rm el}(u),
\nn\\
\mathcal Z_{\rm mag,full}^{SU(n)}
&:=
Z_{\rm sing}
\oint_{\mathcal C_\eta^{\rm mag}}
\prod_{r=1}^{n}\frac{\dd v_r}{2\pi\ii}\,
\delta_T\!\left(\sum_{r=1}^{n}v_r\right)Z_{\rm mag}(v).
\label{eq:SU-projected-partition-functions}
\end{align}

In the chamber \eqref{eq:UNc-selected-chamber} the electric delta gives
\begin{equation}
\delta_T(S_{\mathcal I})
=
\delta_T\!\left(\sum_{i\in\mathcal I}\widetilde m_i\right),
\end{equation}
while the magnetic delta gives
\begin{equation}
\delta_T\!\left(-\sum_{j\in\mathcal J}\widetilde m_j\right).
\end{equation}
Since \(\mathcal I=\mathcal J^c\), the two constraints agree on the flavor \(SU(N_f)\) slice
\begin{equation}
\widetilde M=\sum_{i=1}^{N_f}\widetilde m_i=0
\end{equation}
modulo the period lattice of the torus. Thus the \(SU\) projection is compatible with the complementarity \(\mathcal I=\mathcal J^c\) on the traceless flavor slice.

Equivalently, one may solve the \(SU\) constraint explicitly before applying the JK prescription.


\subsubsection{Theta-function confining identity}
\label{subsubsec:theta-U0-identity}

\begin{proposition}[Theta-function confining endpoint]
\label{prop:theta-U0-confining-endpoint}
The endpoint of the rank-changing identity with magnetic gauge group \(U(0)\) gives a
pure theta-function evaluation.  Set
\begin{equation}
N_f=N_c=:N,
\qquad
M=\sum_{i=1}^{N}m_i,
\qquad
\widetilde M=\sum_{i=1}^{N}\widetilde m_i .
\label{eq:theta-U0-mass-defs}
\end{equation}
There are no \(b_\alpha,\widetilde b_\alpha\) fugacities in this endpoint.  With the mass
identification \eqref{eq:theta-U0-mass-defs}, the Bernoulli phase in the block expression is
independent of the integration variables.  The remaining constant is cancelled by the standard
normalization of the selected simple poles.  The result can therefore be written as the following
identity involving theta functions only, in the same special-function family as classical theta and elliptic hypergeometric evaluations~\cite{Spiridonov:2001ig,Rains:2003}:
\begin{equation}
\begin{aligned}
&\frac{1}{N!}
\oint_{\mathcal C_{\eta}^{(N)}}
\prod_{a=1}^{N}
\frac{\mathrm d u_a}{2\pi\ii}\;
\frac{
\Theta\!\left(-\sum_{a=1}^{N}u_a-M-\lambda;\omega\right)
\Theta\!\left(\sum_{a=1}^{N}u_a-\widetilde M-\lambda;\omega\right)
}{
\displaystyle\prod_{a=1}^{N}\prod_{i=1}^{N}
\Theta(-u_a-m_i;\omega)
\Theta(u_a-\widetilde m_i;\omega)
}
\prod_{\substack{a,b=1\\ a\neq b}}^{N}
\Theta(-u_a+u_b;\omega)
\\[1mm]
&\qquad =
\frac{1}{(\Theta'_0)^N}
\frac{
\Theta(-\lambda;\omega)
\Theta(-\lambda-M-\widetilde M;\omega)
}{
\displaystyle\prod_{i,j=1}^{N}
\Theta(-m_i-\widetilde m_j;\omega)
} .
\end{aligned}
\label{eq:theta-U0-evaluation}
\end{equation}
Here \(\mathcal C_{\eta}^{(N)}\) is the JK contour in the chamber
\begin{equation}
\eta=(1,\ldots,1),
\end{equation}
selecting the poles
\begin{equation}
u_a=\widetilde m_{i_a},
\qquad
\{i_1,\ldots,i_N\}=\{1,\ldots,N\},
\label{eq:theta-U0-selected-poles}
\end{equation}
while the Vandermonde theta factor eliminates coincident labels.  Equation~\eqref{eq:theta-U0-evaluation}
is the confining endpoint of the \(U(N_c)\leftrightarrow U(N_f-N_c)\) residue identity.
\end{proposition}

\begin{proof}
This is the \(n=0\) specialization of Theorem~\ref{thm:determinant-pair-cylinder-duality}.  The magnetic gauge integral is empty, the spectator sectors disappear, and the remaining electric residues are labelled by permutations of the \(N\) anti-fundamental poles.  The Vandermonde theta factor removes coincident labels, while the selected-pole normalization cancels the residual Bernoulli phase, leaving precisely the theta-function evaluation \eqref{eq:theta-U0-evaluation}.
\end{proof}

\subsection{Complete intersections on the cylinder}
\label{sec:complete-intersections-cylinder}

For the GLSM of a complete intersection Calabi--Yau in projective space, the relevant one-loop building block is the regularized determinant
\begin{equation}
z_{\rm RR}^{2d}(u)
=
-\frac{e^{\pi i u}}{1-e^{2\pi \ii u}}
=
\frac{1}{2 \ii\sin(\pi u)}.
\end{equation}
The Dirichlet block is normalized with the reflected inverse convention inherited from the four-dimensional blocks:
\begin{equation}
z_{\rm DD}^{2d}(u)
=
e^{-\pi i u}\bigl(1-e^{2\pi \ii u}\bigr)
=
-2\ii\sin(\pi u).
\end{equation}
Thus
\begin{equation}
z_{\rm DD}^{2d}(u)z_{\rm RR}^{2d}(-u)=1,
\qquad
z_{\rm DD}^{2d}(u)z_{\rm RR}^{2d}(u)=-1.
\end{equation}
The elementary monodromies and reflection properties are
\begin{equation}
z_{\rm RR}^{2d}(u+1)=-z_{\rm RR}^{2d}(u),
\qquad
z_{\rm DD}^{2d}(u+1)=-z_{\rm DD}^{2d}(u),
\end{equation}
and
\begin{equation}
z_{\rm RR}^{2d}(-u)=-z_{\rm RR}^{2d}(u),
\qquad
z_{\rm DD}^{2d}(-u)=-z_{\rm DD}^{2d}(u).
\end{equation}
The block \(z_{\rm RR}^{2d}\) has simple poles at \(u\in\mathbb Z\), with leading behavior
\begin{equation}
z_{\rm RR}^{2d}(u)
=
\frac{1}{2\pi \ii\,u}+O(1)
\qquad\text{as }u\to 0.
\end{equation}
More generally, for an integer charge \(Q\in\mathbb Z_{\ne 0}\) and a shift \(c\in\mathbb C\), the function \(z_{\rm RR}^{2d}(Q u+c)\) has simple poles at the \(|Q|\) points
\begin{equation}
u_{*,\ell}=\frac{\ell-c}{Q},
\qquad
\ell=0,\ldots,|Q|-1,
\end{equation}
with ordinary residue
\begin{equation}
\operatorname*{Res}_{u=u_{*,\ell}}
z_{\rm RR}^{2d}(Q u+c)
=
\frac{(-1)^{\ell}}{2\pi \ii\,Q}.
\label{eq:ci-local-residue-2d}
\end{equation}
The sign \((-1)^\ell\) reflects the monodromy of \(z_{\rm RR}^{2d}\) under integer shifts of its argument. Throughout this subsection we absorb the factor \(1/(2\pi \ii)\) in the contour measure, so that the residue contribution of this pole is
\begin{equation}
\frac{(-1)^\ell}{Q}.
\end{equation}
We use the oriented rank-one JK convention
\begin{equation}
\operatorname{JK}_{+}
=
\sum_{\text{positive-charge poles}}\operatorname{Res},
\qquad
\operatorname{JK}_{-}
=
-\sum_{\text{negative-charge poles}}\operatorname{Res}.
\label{eq:ci-oriented-JK-definition}
\end{equation}
With this convention, the two chambers compute the same partition function whenever the residue theorem applies:
\begin{equation}
\operatorname{JK}_{+}Z=\operatorname{JK}_{-}Z.
\end{equation}

We consider a \(U(1)\) GLSM with \(N\) chiral multiplets \(X_i\), \(i=1,\ldots,N\), of gauge charge \(+1\), and \(r\) chiral multiplets \(P_A\), \(A=1,\ldots,r\), of gauge charge \(-\mathfrak n_A\), with
\begin{equation}
\mathfrak n_A\in\mathbb Z_{>0}.
\end{equation}
The superpotential is
\begin{equation}
W=\sum_{A=1}^{r}P_A\,G_A(X_1,\ldots,X_N),
\end{equation}
where each \(G_A\) is homogeneous of degree \(\mathfrak n_A\). The Calabi--Yau condition is
\begin{equation}
\sum_{A=1}^{r}\mathfrak n_A=N.
\label{eq:ci-CY-condition}
\end{equation}
Geometrically, the positive JK phase realizes the complete intersection
\begin{equation}
X_{\mathfrak n_1,\ldots,\mathfrak n_r}
=
\{G_1=\cdots=G_r=0\}
\subset \mathbb{CP}^{N-1}.
\end{equation}
Denoting by \(u\) the \(U(1)\) gauge fugacity and by \(x_i,\mathfrak p_A\) the flavor fugacities of \(X_i,P_A\), the cylinder one-loop integrand is
\begin{equation}
Z_{\mathfrak n_1,\ldots,\mathfrak n_r}^{2d}(u)
=
\prod_{A=1}^{r}
z_{\rm RR}^{2d}\bigl[-\mathfrak n_Au+\mathfrak p_A\bigr]
\prod_{i=1}^{N}
z_{\rm RR}^{2d}\bigl[u+x_i\bigr].
\label{eq:ci-integrand}
\end{equation}
Under \(u\mapsto u+1\), the full integrand transforms as
\begin{equation}
Z_{\mathfrak n_1,\ldots,\mathfrak n_r}^{2d}(u+1)
=
(-1)^{N+\sum_{A=1}^{r}\mathfrak n_A}
Z_{\mathfrak n_1,\ldots,\mathfrak n_r}^{2d}(u).
\end{equation}
The Calabi--Yau condition \eqref{eq:ci-CY-condition} gives
\begin{equation}
N+\sum_{A=1}^{r}\mathfrak n_A=2N,
\end{equation}
and therefore
\begin{equation}
Z_{\mathfrak n_1,\ldots,\mathfrak n_r}^{2d}(u+1)=Z_{\mathfrak n_1,\ldots,\mathfrak n_r}^{2d}(u).
\end{equation}
Equivalently, in the multiplicative variable
\begin{equation}
z=e^{2\pi \ii u},
\end{equation}
the Calabi--Yau condition cancels the net gauge-dependent half-power of \(z\). The meromorphic one-form
\begin{equation}
\frac{dz}{2\pi \ii z}\,Z_{\mathfrak n_1,\ldots,\mathfrak n_r}^{2d}(z)
\end{equation}
has no extra residue at \(z=0\) or at \(z=\infty\). This is the multiplicative form of the same single-valuedness condition.

The superpotential constrains the flavor fugacities. If the monomial
\begin{equation}
P_A X_1^{k_1}\cdots X_N^{k_N},
\qquad
\sum_{i=1}^{N}k_i=\mathfrak n_A,
\end{equation}
appears in \(P_A G_A\), flavor neutrality requires
\begin{equation}
\mathfrak p_A+\sum_{i=1}^{N}k_i x_i=0.
\end{equation}
For a generic polynomial \(G_A\), this fixes most flavor symmetries. In residue computations it is convenient to keep the \(x_i\) distinct as an equivariant regulator and impose the desired superpotential constraints only after the residues have been evaluated.

The positive chamber receives contributions from the poles of the positively charged fields \(X_i\), namely
\begin{equation}
u_i^+=-x_i,
\qquad
 i=1,\ldots,N.
\end{equation}
Assuming the \(x_i\) are distinct, these poles are simple. Each pole satisfies \(u_i^++x_i=0\), so there is no additional monodromy sign. The positive-chamber JK residue is
\begin{equation}
\operatorname{JK}_{+}
Z_{\mathfrak n_1,\ldots,\mathfrak n_r}^{2d}
=
\sum_{i=1}^{N}
\left[
\prod_{A=1}^{r}
z_{\rm RR}^{2d}\bigl[\mathfrak n_Ax_i+\mathfrak p_A\bigr]
\right]
\prod_{\substack{j=1\\ j\neq i}}^{N}
z_{\rm RR}^{2d}\bigl[x_j-x_i\bigr].
\label{eq:ci-JK-plus}
\end{equation}
This is the cylinder expression for the geometric phase, localized at the poles of the fields \(X_i\) parametrizing the ambient projective space \(\mathbb{CP}^{N-1}\).

The negative chamber receives contributions from the poles of the negatively charged fields \(P_A\). For fixed \(A\), the pole equation is
\begin{equation}
-\mathfrak n_Au+\mathfrak p_A\in\mathbb Z.
\end{equation}
Modulo the periodicity \(u\sim u+1\), this gives \(\mathfrak n_A\) solutions,
\begin{equation}
u_{A,\ell}^{-}=\frac{\mathfrak p_A-\ell}{\mathfrak n_A},
\qquad
\ell=0,\ldots,\mathfrak n_A-1.
\end{equation}
At such a pole, the local residue contribution of \(z_{\rm RR}^{2d}[-\mathfrak n_Au+\mathfrak p_A]\) is
\begin{equation}
-\frac{(-1)^{\ell}}{\mathfrak n_A}.
\end{equation}
The oriented negative-chamber prescription \eqref{eq:ci-oriented-JK-definition} includes an additional minus sign, yielding a net contribution \((+1)(-1)^\ell/\mathfrak n_A\). Hence
\begin{equation}
\begin{aligned}
\operatorname{JK}_{-}
Z_{\mathfrak n_1,\ldots,\mathfrak n_r}^{2d}
=
\sum_{A=1}^{r}\frac{1}{\mathfrak n_A}
\sum_{\ell=0}^{\mathfrak n_A-1}(-1)^{\ell}
&
\left[
\prod_{i=1}^{N}
z_{\rm RR}^{2d}
\left[
x_i+\frac{\mathfrak p_A-\ell}{\mathfrak n_A}
\right]
\right]
\\
&\times
\left[
\prod_{\substack{B=1\\ B\neq A}}^{r}
z_{\rm RR}^{2d}
\left[
\mathfrak p_B-\frac{\mathfrak n_B}{\mathfrak n_A}(\mathfrak p_A-\ell)
\right]
\right].
\end{aligned}
\label{eq:ci-JK-minus}
\end{equation}
The signs \((-1)^\ell\) are the one-cycle cylinder analogue of the orbifold projection phases that appear in Landau--Ginzburg orbifold expansions of GLSM partition functions~\cite{Witten:1993yc}.

\begin{proposition}[Two-dimensional chamber equality]
\label{prop:ci-two-dimensional-chamber-equality}
Assume that the two-dimensional complete intersection cylinder integrand is single-valued on the gauge circle.  In the complete-intersection family \eqref{eq:ci-integrand}, this follows from the Calabi--Yau relation \eqref{eq:ci-CY-condition}; equivalently, it is the two-dimensional degeneration of the cancellation of the gauge-dependent Bernoulli anomaly, schematically \(\sum_s\varepsilon_s q_s^2=0\), together with the corresponding linear phase condition in the higher-dimensional block.  Assume also that all poles are simple and away from the contour at infinity.  Then the positive and negative oriented JK chambers compute the same partition function,
\begin{equation}
\operatorname{JK}_{+}Z_{\mathfrak n_1,\ldots,\mathfrak n_r}^{2d}
=
\operatorname{JK}_{-}Z_{\mathfrak n_1,\ldots,\mathfrak n_r}^{2d}.
\label{eq:ci-chamber-duality}
\end{equation}
\end{proposition}

\begin{proof}
Single-valuedness, ensured in the examples above by \eqref{eq:ci-CY-condition}, makes the integrand a meromorphic one-form on the gauge circle.  The residue theorem gives
\begin{equation}
\sum_{\text{positive-charge poles}}\operatorname{Res}
+
\sum_{\text{negative-charge poles}}\operatorname{Res}=0.
\end{equation}
The oriented JK convention \eqref{eq:ci-oriented-JK-definition} includes the compensating sign for the negative chamber.  Therefore the equality of ordinary residues is equivalent to \eqref{eq:ci-chamber-duality}.
\end{proof}

The formulas above reproduce explicitly examples such as the cubic curve in \(\mathbb{CP}^2\), the quintic threefold, and the intersection of two quadrics in \(\mathbb{CP}^3\). In these examples the equality of chambers requires the signs \((-1)^\ell\) in \eqref{eq:ci-JK-minus}.

For \(r=1\), one obtains a hypersurface of degree \(\mathfrak n=N\) in \(\mathbb{CP}^{N-1}=\mathbb{CP}^{\mathfrak n-1}\). The integrand reduces to
\begin{equation}
Z_{\mathfrak n}^{2d}(u)
=
z_{\rm RR}^{2d}\bigl[-\mathfrak n u+\mathfrak p\bigr]
\prod_{i=1}^{\mathfrak n}
z_{\rm RR}^{2d}\bigl[u+x_i\bigr].
\end{equation}
Imposing, for instance, the flavor constraint associated with the monomial \(PX_1\cdots X_{\mathfrak n}\),
\begin{equation}
\mathfrak p+\sum_{i=1}^{\mathfrak n}x_i=0,
\end{equation}
the positive-chamber residue becomes
\begin{equation}
\operatorname{JK}_{+}Z_{\mathfrak n}^{2d}
=
\sum_{i=1}^{\mathfrak n}
z_{\rm RR}^{2d}
\left[
\mathfrak n x_i-\sum_{k=1}^{\mathfrak n}x_k
\right]
\prod_{\substack{j=1\\ j\neq i}}^{\mathfrak n}
z_{\rm RR}^{2d}\bigl[x_j-x_i\bigr],
\end{equation}
while the oriented negative-chamber residue is
\begin{equation}
\operatorname{JK}_{-}Z_{\mathfrak n}^{2d}
=
\frac{1}{\mathfrak n}
\sum_{\ell=0}^{\mathfrak n-1}(-1)^{\ell}
\prod_{i=1}^{\mathfrak n}
z_{\rm RR}^{2d}
\left[
x_i-\frac{1}{\mathfrak n}\sum_{k=1}^{\mathfrak n}x_k-\frac{\ell}{\mathfrak n}
\right].
\label{eq:ci-JK-minus-hypersurface}
\end{equation}
For \(\mathfrak n=5\), this is the cylinder counterpart of the quintic GLSM. The positive chamber computes the cylinder partition function in the geometric phase of the quintic threefold
\begin{equation}
X_5\subset \mathbb{CP}^4,
\end{equation}
localized on the five torus fixed points of \(\mathbb{CP}^4\). The negative chamber computes the same cylinder partition function in the Landau--Ginzburg orbifold phase, with orbifold group \(\mathbb Z_{\mathfrak n}\) and \(\mathfrak n\) twisted sectors labelled by \(\ell=0,\ldots,\mathfrak n-1\). Since the cylinder has only one periodic holonomy direction, the orbifold sum has \(\mathfrak n\) sectors rather than the \(\mathfrak n^2\) elliptic sectors that arise on \(T^2\) in the elliptic genus.

For \(r>1\), the negative chamber is naturally interpreted as a hybrid phase. In the negative JK  chamber the fields \(P_A\) cannot all vanish, and after quotienting by the \(U(1)\) gauge symmetry they parametrize the weighted projective stack
\begin{equation}
\mathbb{WCP}^{r-1}_{[\mathfrak n_1,\ldots,\mathfrak n_r]}
:=
\bigl[(\mathbb C^r\setminus\{0\})/\mathbb C^\ast\bigr],
\end{equation}
where \(\mathbb C^\ast\) acts with weights \((\mathfrak n_1,\ldots,\mathfrak n_r)\). The sign of the gauge charges does not affect the resulting stack, since one may invert the \(\mathbb C^\ast\) parameter. The integers \(\mathfrak n_A\) need not be pairwise coprime. If
\begin{equation}
\gcd(\mathfrak n_1,\ldots,\mathfrak n_r)=g>1,
\end{equation}
then the stack has a generic stabilizer \(\mathbb Z_g\). A stronger condition, often called well-formedness, requires
\begin{equation}
\gcd(\mathfrak n_1,\ldots,\widehat{\mathfrak n}_A,\ldots,\mathfrak n_r)=1
\end{equation}
for every \(A\). This removes generic codimension-one stabilizers, but is not needed for the GLSM nor for the residue computation. Thus the negative chamber is a hybrid Landau--Ginzburg phase fibered over \(\mathbb{WCP}^{r-1}_{[\mathfrak n_1,\ldots,\mathfrak n_r]}\), with Landau--Ginzburg fiber given by the fields \(X_i\) and superpotential \(W=\sum_A P_A G_A(X)\). For \(r=1\), the base \(\mathbb{WCP}^{0}_{[\mathfrak n]}\) is a stacky point with stabilizer \(\mathbb Z_{\mathfrak n}\), recovering the familiar Landau--Ginzburg orbifold phase.

The cylinder determinant \(z_{\rm RR}^{2d}\) gives a direct \(\cyl\simeq I\times S^1\) analogue of the GLSM phase structure, compatible with the exact boundary framework for two-dimensional \(2,2\) gauge theories \cite{Hori:2013ika}. The equality \eqref{eq:ci-chamber-duality} provides a stringent consistency check of the cylinder partition function. This two-dimensional construction should be viewed as the \(\cyl\) counterpart of the four-dimensional \(\cyl\times T^2\) one-loop blocks discussed above. Starting from the full elliptic-gamma representation, the first degeneration from \(\cyl\times T^2\) to \(\cyl\times S^1\) is
\begin{equation}
\operatorname{Im}\tau\to+\infty,
\qquad
p=e^{2\pi \ii\tau}\to0.
\end{equation}
The remaining theta block has elliptic nome
\begin{equation}
q=e^{2\pi \ii\omega}.
\end{equation}
Its cylinder degeneration is
\begin{equation}
\operatorname{Im}\omega\to+\infty,
\qquad
q=e^{2\pi \ii\omega}\to0.
\end{equation}
In this limit
\begin{equation}
\Theta(u;\omega)=(e^{2\pi \ii u};q)_\infty(qe^{-2\pi \ii u};q)_\infty
\longrightarrow
1-e^{2\pi \ii u}.
\end{equation}
Thus the four-dimensional one-loop block reduces, up to Bernoulli/contact terms, to the two-dimensional block \(z_{\rm RR}^{2d}(u)\). The chamber equality \eqref{eq:ci-chamber-duality} is the two-dimensional shadow of the electric--magnetic residue identities discussed in Section~\ref{sec:duality}.

\subsection{Factorization and cap wavefunctions}
\label{sec:factorization-half-blocks}

The factorization viewpoint follows the holomorphic-block philosophy of Refs.~\cite{Pasquetti:2011fj,Beem:2012mb,Nieri:2013yra,Nieri:2013vba,Nieri:2015yia} and its four-dimensional elliptic realization on geometries such as \(D^2\times T^2\)~\cite{Longhi:2019hdh}.

The cylinder \(\cyl\times T^2\) has two boundary components, and its partition function is naturally a boundary-to-boundary kernel. At the level of one-loop determinants this kernel factorizes into two cap wavefunctions, equivalently the holomorphic blocks associated with the two caps of the cylinder. We do not repeat the three-dimensional factorization analysis: the factorization of \(3d\) interval partition functions on \(I\times T^2\simeq \cyl\times S^1\) into hemisphere blocks has been worked out recently in Ref.~\cite{Zhao:2025ixe}. We focus instead on the four-dimensional elliptic cylinder kernel and on its two-dimensional degeneration on \(\cyl\simeq I\times S^1\).

The words ``holomorphic'' and ``anti-holomorphic'' below refer to the two opposite cap orientations in the gluing. They should not be interpreted as complex conjugation.

\subsubsection{Four-dimensional factorization on \texorpdfstring{\(\cyl\times T^2\)}{Cyl x T2}}

We first discuss the factorization of the four-dimensional cylinder blocks on
\(\cyl\times T^2\).  We use the full cylinder blocks defined in \eqref{eq:4d-cylinder-blocks}, with
\begin{equation}
B_3(u):=B_{3,3}(u;1,\tau,\omega).
\end{equation}
With the elliptic-Gamma conventions of Refs.~\cite{Felder:1999mq,Narukawa:2003,Friedman:2004}, the elliptic Gamma function satisfies the gluing identity
\begin{equation}
\Gamma_e(u;\tau,\omega)\,
\Gamma_e(\omega-u;\tau,\omega)
=
\frac{1}{\Theta(u;\omega)}.
\label{eq:gamma-theta-factorization}
\end{equation}
Thus the theta factor appearing in the cylinder determinant can be regarded as the inverse product of two elliptic-Gamma cap wavefunctions.  These cap wavefunctions are the \(D^2\times T^2\) holomorphic blocks of Refs.~\cite{Pasquetti:2011fj,Beem:2012mb,Nieri:2013yra,Nieri:2013vba,Nieri:2015yia,Longhi:2019hdh}, written in the Shintani--Barnes normalization used in this paper.  Compared with the common holomorphic-block conventions, we keep the Bernoulli contact factors explicit and use a reflected Dirichlet argument so that the elementary inverse relation is exactly \(\mathcal B_{\mathrm D}(u)\mathcal B_{\mathrm R}(-u)=1\).

We denote the \(D^2\times T^2\) cap wavefunctions by
\begin{equation}
\mathcal B_{\mathrm D}(u)
:=
\frac{
\exp\!\left[-\frac{\ii\pi}{3}B_3(-u)\right]
}{
\Gamma_e(-u;\tau,\omega)
},
\qquad
\mathcal B_{\mathrm R}(u)
:=
\exp\!\left[+\frac{\ii\pi}{3}B_3(u)\right]\,
\Gamma_e(u;\tau,\omega).
\label{eq:cap-wavefunctions-4d}
\end{equation}
The wavefunctions associated with the oppositely oriented cap are
\begin{equation}
\mathcal B_{\mathrm D}^{\vee}(u)
:=
\frac{
\exp\!\left[-\frac{\ii\pi}{3}B_3(\omega+u)\right]
}{
\Gamma_e(\omega+u;\tau,\omega)
},
\qquad
\mathcal B_{\mathrm R}^{\vee}(u)
:=
\exp\!\left[+\frac{\ii\pi}{3}B_3(\omega-u)\right]\,
\Gamma_e(\omega-u;\tau,\omega).
\label{eq:opposite-cap-wavefunctions-4d}
\end{equation}
The superscript \(\vee\) refers to the opposite cap orientation in the gluing, not to Hermitian conjugation.

Using \eqref{eq:gamma-theta-factorization}, one obtains the exact factorization
\begin{equation}
z_{\rm DD}(u)
=
\mathcal B_{\mathrm D}(u)\,
\mathcal B_{\mathrm D}^{\vee}(u),
\qquad
z_{\rm RR}(u)
=
\mathcal B_{\mathrm R}(u)\,
\mathcal B_{\mathrm R}^{\vee}(u).
\label{eq:4d-block-factorization}
\end{equation}
Moreover, the reflected inverse relation already holds separately at the level of the two caps:
\begin{equation}
\mathcal B_{\mathrm D}(u)\,
\mathcal B_{\mathrm R}(-u)=1,
\qquad
\mathcal B_{\mathrm D}^{\vee}(u)\,
\mathcal B_{\mathrm R}^{\vee}(-u)=1.
\label{eq:4d-half-block-inverse}
\end{equation}
Thus the elementary cancellation
\begin{equation}
z_{\rm DD}(u)\,z_{\rm RR}(-u)=1
\end{equation}
is refined to two independent cap-wavefunction cancellations.

This immediately factorizes any integrand built out of \(z_{\rm DD}\) and \(z_{\rm RR}\) into cap and opposite-cap contributions. For example, the electric integrand of the \(U(N_c)\leftrightarrow U(N_f-N_c)\) duality can be written as
\begin{equation}
Z_{\rm el}(u)
=
Z_{\rm el}^{\rm hol}(u)\,
Z_{\rm el}^{\rm anti}(u),
\label{eq:el-factorization}
\end{equation}
where
\begin{equation}
\begin{aligned}
Z_{\rm el}^{\rm hol}(u)
=&\frac{1}{N_c!}\,
\mathcal B_{\mathrm D}\!\left[\sum_{a=1}^{N_c}u_a+M+\lambda\right]\,
\mathcal B_{\mathrm D}\!\left[-\sum_{a=1}^{N_c}u_a+\widetilde M+\lambda\right]
\\
&\times
\prod_{\substack{a,b=1\\ a\neq b}}^{N_c}
\mathcal B_{\mathrm D}[u_a-u_b]
\\
&\times
\prod_{a=1}^{N_c}\prod_{i=1}^{N_f}
\mathcal B_{\mathrm R}[-u_a-m_i]\,
\mathcal B_{\mathrm R}[u_a-\widetilde m_i]
\\
&\times
\prod_{a=1}^{N_c}\prod_{\alpha=1}^{n}
\mathcal B_{\mathrm D}[u_a+b_\alpha]\,
\mathcal B_{\mathrm D}[-u_a+\widetilde b_\alpha].
\end{aligned}
\label{eq:el-hol-factor}
\end{equation}
The factor \(Z_{\rm el}^{\rm anti}(u)\) is obtained from \eqref{eq:el-hol-factor} by replacing \(\mathcal B_{\mathrm D},\mathcal B_{\mathrm R}\) with \(\mathcal B_{\mathrm D}^{\vee},\mathcal B_{\mathrm R}^{\vee}\).

Similarly, the magnetic integrand \eqref{eq:Zmag-UNc-reint-completed} factorizes as
\begin{equation}
Z_{\rm mag}(v)
=
Z_{\rm mag}^{\rm hol}(v)\,
Z_{\rm mag}^{\rm anti}(v).
\label{eq:mag-factorization}
\end{equation}
Explicitly,
\begin{equation}
\begin{aligned}
Z_{\rm mag}^{\rm hol}(v)
=&\frac{1}{n!}\,
\mathcal B_{\mathrm D}\!\left[-\sum_{r=1}^{n}v_r+\lambda\right]\,
\mathcal B_{\mathrm D}\!\left[\sum_{r=1}^{n}v_r+M+\widetilde M+\lambda\right]
\\
&\times
\prod_{\substack{r,s=1\\ r\neq s}}^{n}
\mathcal B_{\mathrm D}[v_r-v_s]
\\
&\times
\prod_{r=1}^{n}\prod_{i=1}^{N_f}
\mathcal B_{\mathrm R}[v_r+\widetilde m_i]\,
\mathcal B_{\mathrm D}[-v_r+m_i]
\\
&\times
\prod_{r=1}^{n}\prod_{\alpha=1}^{n}
\mathcal B_{\mathrm R}[v_r-b_\alpha]\,
\mathcal B_{\mathrm D}[v_r-b_\alpha]\,
\mathcal B_{\mathrm R}[-v_r+b_\alpha]\,
\mathcal B_{\mathrm R}[-v_r-\widetilde b_\alpha].
\end{aligned}
\label{eq:mag-hol-factor}
\end{equation}
The factor \(Z_{\rm mag}^{\rm anti}(v)\) is obtained from \eqref{eq:mag-hol-factor} by replacing \(\mathcal B_{\mathrm D},\mathcal B_{\mathrm R}\) with \(\mathcal B_{\mathrm D}^{\vee},\mathcal B_{\mathrm R}^{\vee}\).

The identity inserted in \eqref{eq:Zmag-UNc-reint-completed} factorizes separately on the two caps:
\begin{equation}
\mathcal B_{\mathrm D}[v_r-b_\alpha]\,
\mathcal B_{\mathrm R}[-v_r+b_\alpha]
=1,
\qquad
\mathcal B_{\mathrm D}^{\vee}[v_r-b_\alpha]\,
\mathcal B_{\mathrm R}^{\vee}[-v_r+b_\alpha]
=1 .
\label{eq:factorization-half-block-invertible-pairs}
\end{equation}
Thus the cap-wavefunction factorization is compatible with the JK presentation without introducing any separate local factor to be split between the two caps.

The singlet sector factorizes in the same way:
\begin{equation}
Z_{\rm sing}
=
Z_{\rm sing}^{\rm hol}\,
Z_{\rm sing}^{\rm anti}.
\label{eq:singlet-factorization}
\end{equation}
For instance,
\begin{equation}
Z_M^{\rm hol}
=
\prod_{i,j=1}^{N_f}
\mathcal B_{\mathrm R}[-m_i-\widetilde m_j],
\qquad
Z_M^{\rm anti}
=
\prod_{i,j=1}^{N_f}
\mathcal B_{\mathrm R}^{\vee}[-m_i-\widetilde m_j],
\end{equation}
and similarly for the boundary singlet sector \(Z_B\).

The residue expansion is compatible with this factorization. At an electric pole labelled by
\begin{equation}
\mathcal I\subset\{1,\ldots,N_f\},
\qquad
|\mathcal I|=N_c,
\end{equation}
the residue factorizes as
\begin{equation}
\mathcal E_{\mathcal I}
=
\mathcal E_{\mathcal I}^{\rm hol}\,
\mathcal E_{\mathcal I}^{\rm anti}.
\label{eq:electric-residue-factorization}
\end{equation}
Likewise, at a magnetic pole labelled by
\begin{equation}
\mathcal J\subset\{1,\ldots,N_f\},
\qquad
|\mathcal J|=n,
\end{equation}
one has
\begin{equation}
\mathcal M_{\mathcal J}
=
\mathcal M_{\mathcal J}^{\rm hol}\,
\mathcal M_{\mathcal J}^{\rm anti}.
\label{eq:magnetic-residue-factorization}
\end{equation}
In the magnetic \(b_\alpha\)-sector, the completed DD/RR presentation separates the same-argument contact factor from the RR factors which lie in the opposite chamber. The reflected inverse pairs in \eqref{eq:factorization-half-block-invertible-pairs} remain equal to one separately on each cap, while the same-argument DD/RR products contribute only holomorphic contact factors.

The term-by-term matching of the full cylinder duality,
\begin{equation}
Z_{\rm sing}\,\mathcal M_{\mathcal J}
=
\mathcal E_{\mathcal J^c},
\end{equation}
is therefore refined to a matching of factorized kernels:
\begin{equation}
Z_{\rm sing}^{\rm hol}\,
\mathcal M_{\mathcal J}^{\rm hol}\,
Z_{\rm sing}^{\rm anti}\,
\mathcal M_{\mathcal J}^{\rm anti}
=
\mathcal E_{\mathcal J^c}^{\rm hol}\,
\mathcal E_{\mathcal J^c}^{\rm anti},
\label{eq:factorized-residue-matching}
\end{equation}
with the completed DD/RR pairs inserted separately on the two caps.

\subsubsection{The elliptic kernel interpretation}
\label{subsec:elliptic-kernel-interpretation}

We now make the boundary-to-boundary interpretation of the cylinder block more explicit.  The following bra-ket notation should be read as a formal boundary-state calculus for the localized wavefunctions on the elliptic fugacity torus.  Its content is the set of exact identities among the cap wavefunctions introduced in \eqref{eq:cap-wavefunctions-4d} and \eqref{eq:opposite-cap-wavefunctions-4d}; no additional Hilbert-space completion is assumed. Accordingly, this subsection should be viewed as a formal calculus for the elementary \(D/R\) matrix elements computed in this paper, rather than as a construction of a complete boundary-state basis. A more intrinsic kernel would require allowing general boundary interactions or defects, for instance boundary operators lifting Wilson-loop bases, and is not attempted here.

We denote a localized boundary wavefunction by
\begin{equation}
|u;\mathsf q\rangle,
\qquad
\mathsf q\in\{\mathrm D,\mathrm R\},
\end{equation}
where \(u\) is the boundary fugacity and \(\mathsf q\) is the boundary polarization.  The cap state has components \(\mathcal B_{\mathrm D},\mathcal B_{\mathrm R}\), while the opposite cap has components \(\mathcal B_{\mathrm D}^{\vee},\mathcal B_{\mathrm R}^{\vee}\).  As above, \(\vee\) indicates the opposite orientation of the cap, not Hermitian conjugation.  Equations \eqref{eq:cap-wavefunctions-4d} and \eqref{eq:opposite-cap-wavefunctions-4d} imply the elementary pairings
\begin{equation}
\mathcal B_{\mathrm R}^{\vee}(u)\,\mathcal B_{\mathrm D}(u-\omega)=1,
\qquad
\mathcal B_{\mathrm D}^{\vee}(u)\,\mathcal B_{\mathrm R}(u+\omega)=1.
\label{eq:cap-dual-pairings}
\end{equation}
Equivalently, the affine shifts by \(\omega\) can be absorbed into the formal gluing pairings
\begin{equation}
\langle u;\mathrm R^{\vee}|v;\mathrm D\rangle
=
\delta_E(v-u+\omega),
\qquad
\langle u;\mathrm D^{\vee}|v;\mathrm R\rangle
=
\delta_E(v-u-\omega),
\label{eq:formal-boundary-pairing}
\end{equation}
where \(\delta_E\) is the formal delta function on the elliptic fugacity torus.  These shifts are part of the gluing rule on the boundary torus.

The cylinder partition function is the localized matrix element of a boundary-to-boundary operator.  For the four elementary polarizations we write
\begin{equation}
\mathsf K_{\mathsf q_+\mathsf q_-}^{\cyl\times T^2}
=
\int_E \mathrm d\mu(u)\,
|u;\mathsf q_+\rangle\,
z_{\mathsf q_+\mathsf q_-}^{\cyl\times T^2}(u)\,
\langle u+\mu_{\mathsf q_+\mathsf q_-};\mathsf q_-^{\vee}|,
\label{eq:cyl-kernel-operator}
\end{equation}
with
\begin{equation}
\mu_{\mathrm{RR}}=0,
\qquad
\mu_{\mathrm{DD}}=0,
\qquad
\mu_{\mathrm{DR}}=\omega,
\qquad
\mu_{\mathrm{RD}}=-\omega.
\label{eq:cyl-kernel-shifts}
\end{equation}
In components this says
\begin{equation}
z_{\mathsf q_+\mathsf q_-}^{\cyl\times T^2}(u)
=
\mathcal B_{\mathsf q_+}(u)\,
\mathcal B_{\mathsf q_-}^{\vee}\!\left(u+\mu_{\mathsf q_+\mathsf q_-}\right).
\label{eq:cyl-kernel-components-halfblocks}
\end{equation}
Using \eqref{eq:cap-dual-pairings}, one obtains the sewing identities
\begin{equation}
z_{\rm RR}^{\cyl\times T^2}(u)\,\mathcal B_{\mathrm D}(u-\omega)
=
\mathcal B_{\mathrm R}(u),
\qquad
z_{\rm DD}^{\cyl\times T^2}(u)\,\mathcal B_{\mathrm R}(u+\omega)
=
\mathcal B_{\mathrm D}(u),
\label{eq:cyl-kernel-changing-polarization}
\end{equation}
and
\begin{equation}
z_{\rm DR}^{\cyl\times T^2}(u)\,\mathcal B_{\mathrm D}(u)
=
\mathcal B_{\mathrm D}(u),
\qquad
z_{\rm RD}^{\cyl\times T^2}(u)\,\mathcal B_{\mathrm R}(u)
=
\mathcal B_{\mathrm R}(u).
\label{eq:cyl-kernel-identity-polarization}
\end{equation}
Equivalently,
\begin{equation}
\mathsf K_{\rm RR}^{\cyl\times T^2}|\mathcal B_{\mathrm D}\rangle
=
|\mathcal B_{\mathrm R}\rangle,
\qquad
\mathsf K_{\rm DD}^{\cyl\times T^2}|\mathcal B_{\mathrm R}\rangle
=
|\mathcal B_{\mathrm D}\rangle,
\label{eq:cyl-kernel-operator-changing-polarization}
\end{equation}
while
\begin{equation}
\mathsf K_{\rm DR}^{\cyl\times T^2}|\mathcal B_{\mathrm D}\rangle
=
|\mathcal B_{\mathrm D}\rangle,
\qquad
\mathsf K_{\rm RD}^{\cyl\times T^2}|\mathcal B_{\mathrm R}\rangle
=
|\mathcal B_{\mathrm R}\rangle.
\label{eq:cyl-kernel-operator-identity-polarization}
\end{equation}
Thus the mixed blocks are identity kernels in the two polarized sectors, whereas the unmixed blocks implement elementary changes of polarization.

This interpretation is compatible with radial slicing of the cap and belongs to the same broad family of gluing and cigar constructions studied in \cite{Dedushenko:2018aox,Dedushenko:2018tgx,Dedushenko:2021mds,Dedushenko:2023qjq}.  Cutting the disk into annuli gives
\begin{equation}
D^2\times T^2
=
(\cyl\times T^2)_n
\cup_{T^3}\cdots\cup_{T^3}
(\cyl\times T^2)_1
\cup_{T^3}
D^{2\prime}\times T^2.
\label{eq:radial-slicing-cap}
\end{equation}
In the boundary-state calculus this becomes a composition of the kernels \eqref{eq:cyl-kernel-operator}.  For instance,
\begin{equation}
\mathsf K_{\rm DD}^{\cyl\times T^2}\mathsf K_{\rm RR}^{\cyl\times T^2}
=
\mathsf K_{\rm DR}^{\cyl\times T^2},
\qquad
\mathsf K_{\rm RR}^{\cyl\times T^2}\mathsf K_{\rm DD}^{\cyl\times T^2}
=
\mathsf K_{\rm RD}^{\cyl\times T^2},
\label{eq:cyl-kernel-composition}
\end{equation}
where the shifts of the fugacity are fixed by the pairings \eqref{eq:formal-boundary-pairing}.  In components this is
\begin{equation}
z_{\rm DD}^{\cyl\times T^2}(u)\,
z_{\rm RR}^{\cyl\times T^2}(u+\omega)
=
z_{\rm DR}^{\cyl\times T^2}(u)=1,
\qquad
z_{\rm RR}^{\cyl\times T^2}(u)\,
z_{\rm DD}^{\cyl\times T^2}(u-\omega)
=
z_{\rm RD}^{\cyl\times T^2}(u)=1.
\label{eq:cyl-kernel-component-composition}
\end{equation}
More generally, any admissible word in the four elementary kernels is a formal radial slicing.  Closed words in the two polarizations refine the same cap, while open words describe radial interfaces changing the boundary polarization.

The same identities also give the expected behavior of a massive chiral pair.  Consider two chiral multiplets \(Q_a,Q_b\) with a superpotential mass term
\begin{equation}
W=m Q_a Q_b,
\qquad
r_a+r_b=2,
\label{eq:massive-pair-superpotential}
\end{equation}
and with opposite gauge and flavor charges.  If the first chiral has cylinder argument
\begin{equation}
x=\rho(u)+\nu+r_a\gamma_R
\label{eq:massive-pair-argument}
\end{equation}
in an \(RR\) block, then the second chiral, in the complementary \(DD\) polarization, has argument
\begin{equation}
-\rho(u)-\nu+(r_b-2)\gamma_R=-x .
\label{eq:massive-pair-opposite-argument}
\end{equation}
Therefore the massive pair contributes the identity
\begin{equation}
z_{\rm RR}(x)z_{\rm DD}(-x)=1.
\label{eq:massive-pair-RR-DD}
\end{equation}
Similarly,
\begin{equation}
z_{\rm DD}(x)z_{\rm RR}(-x)=1,
\qquad
z_{\rm DR}(x)z_{\rm RD}(-x)=1,
\qquad
z_{\rm RD}(x)z_{\rm DR}(-x)=1.
\label{eq:massive-pair-all-polarizations}
\end{equation}
Thus the usual statement that a massive chiral pair with a single supersymmetric vacuum contributes one to the protected index is realized pointwise by the reflected inverse relation of the cylinder blocks, provided the two chirals are assigned complementary boundary polarizations at the two ends.  Non-complementary assignments require additional boundary terms or boundary degrees of freedom in order to make the massive superpotential compatible with the supersymmetric variational problem.

The same affine gluing calculus also reproduces the compact \(S^2\times T^2\) chiral determinant.  Let \(\gamma_\Phi\) denote the total additive fugacity appearing in the chiral determinant, including the gauge, flavor and R-symmetry contributions appropriate to the chiral multiplet under consideration.  For an integer \(S^2\) flux \(b\), set
\begin{equation}
u_+
=
-\gamma_\Phi+\frac{1-b}{2}\omega,
\qquad
u_-
=
-\gamma_\Phi+\frac{1+b}{2}\omega,
\label{eq:S2T2-gluing-arguments}
\end{equation}
so that \(u_--u_+=b\omega\).  The standard \(S^2\) fusion with flux \(b\) pairs the Robin cap with the oppositely oriented Dirichlet cap as
\begin{equation}
Z_{S^2\times T^2}^{\Phi}(\gamma_\Phi,b)
=
\langle \mathcal B_{\mathrm D}^{\vee}|
\mathsf G_{S^2}^{(b)}
|\mathcal B_{\mathrm R}\rangle
=
\mathcal B_{\mathrm R}(u_+)\,
\mathcal B_{\mathrm D}^{\vee}(u_- -\omega).
\label{eq:S2T2-gluing-halfblocks}
\end{equation}
Using \eqref{eq:cap-wavefunctions-4d} and \eqref{eq:opposite-cap-wavefunctions-4d}, this gives
\begin{equation}
Z_{S^2\times T^2}^{\Phi}(\gamma_\Phi,b)
=
\exp\!\left[
\frac{\ii\pi}{3}\left(B_3(u_+)-B_3(u_-)\right)
\right]
\frac{\Gamma_e(u_+;\tau,\omega)}
{\Gamma_e(u_-;\tau,\omega)}.
\label{eq:S2T2-gluing-gamma-ratio}
\end{equation}
The match with the standard \(T^2\times S^2\) chiral determinant is obtained by using the elliptic-Gamma difference equation along the \(\omega\)-direction,
\begin{equation}
\Gamma_e(u+\omega;\tau,\omega)
=
\Theta(u;\tau)\,\Gamma_e(u;\tau,\omega).
\label{eq:elliptic-gamma-omega-recursion}
\end{equation}
Since \(u_-=u_++b\omega\), \eqref{eq:S2T2-gluing-gamma-ratio} becomes the finite theta-product formula
\begin{equation}
Z_{S^2\times T^2}^{\Phi}(\gamma_\Phi,b)
=
\exp\!\left[
\frac{\ii\pi}{3}\left(B_3(u_+)-B_3(u_-)\right)
\right]
\begin{cases}
\displaystyle
\prod_{j=0}^{b-1}
\Theta(u_+ + j\omega;\tau)^{-1}, & b>0,\\[3mm]
1, & b=0,\\[2mm]
\displaystyle
\prod_{j=1}^{-b}
\Theta(u_+ - j\omega;\tau), & b<0.
\end{cases}
\label{eq:S2T2-gluing-theta-product}
\end{equation}
Equivalently, for \(b>0\),
\begin{equation}
\prod_{j=0}^{b-1}
\Theta(u_+ + j\omega;\tau)^{-1}
=
\prod_{j=0}^{b-1}
\Theta\!\left(-\gamma_\Phi+\frac{1-b+2j}{2}\omega;\tau\right)^{-1}.
\label{eq:S2T2-Closset-Shamir-form-positive-flux}
\end{equation}
The exponential prefactor in \eqref{eq:S2T2-gluing-theta-product} is the local Bernoulli contact term associated with our Shintani--Barnes normalization of the cap wavefunctions.  After the corresponding local counterterm normalization, the non-local determinant is precisely the finite theta-product form of the \(T^2\times S^2\) chiral index of Closset and Shamir \cite{Closset:2013sxa}.  Thus both the open cylinder block and the closed \(S^2\times T^2\) determinant are built from the same cap wavefunctions and the same affine gluing shifts on the boundary elliptic curve.  This is the sense in which the localized partition function on \(\cyl\times T^2\) defines an elliptic kernel: it is the matrix element of a boundary-to-boundary operator whose entries are meromorphic elliptic functions of the boundary fugacity and whose sewing with \(D^2\times T^2\) holomorphic-block cap states reproduces the expected cap and compact-space blocks.

\subsubsection{Two-dimensional factorization on \texorpdfstring{\(\cyl\)}{Cyl}}
\label{subsubsec:two-dimensional-factorization-cylinder}

The same structure appears in the two-dimensional cylinder degeneration.  The elementary blocks are
\begin{equation}
z_{\rm RR}^{2d}(u)
=
\frac{1}{2\ii\sin(\pi u)},
\qquad
z_{\rm DD}^{2d}(u)
=
-2\ii\sin(\pi u).
\label{eq:2d-blocks-factorization-section}
\end{equation}
Using Euler's reflection formula,
\begin{equation}
\Gamma(u)\Gamma(1-u)
=
\frac{\pi}{\sin(\pi u)},
\label{eq:euler-reflection}
\end{equation}
one obtains
\begin{equation}
z_{\rm RR}^{2d}(u)
=
\frac{\Gamma(u)\Gamma(1-u)}{2\pi\ii},
\qquad
z_{\rm DD}^{2d}(u)
=
\frac{2\pi\ii}{\Gamma(-u)\Gamma(1+u)}.
\label{eq:2d-blocks-gamma-factorization}
\end{equation}
Choosing a branch of \((2\pi\ii)^{1/2}\), we define the two-dimensional cap wavefunctions by
\begin{equation}
\mathcal B_{\mathrm R}^{2d}(u)
:=
\frac{\Gamma(u)}{(2\pi\ii)^{1/2}},
\qquad
\mathcal B_{\mathrm D}^{2d}(u)
:=
\frac{(2\pi\ii)^{1/2}}{\Gamma(-u)}.
\label{eq:2d-cap-wavefunctions}
\end{equation}
The opposite-cap wavefunctions are
\begin{equation}
\mathcal B_{\mathrm R}^{2d,\vee}(u)
:=
\frac{\Gamma(1-u)}{(2\pi\ii)^{1/2}},
\qquad
\mathcal B_{\mathrm D}^{2d,\vee}(u)
:=
\frac{(2\pi\ii)^{1/2}}{\Gamma(1+u)}.
\label{eq:2d-opposite-cap-wavefunctions}
\end{equation}
The superscript \(\vee\) again denotes the opposite orientation of the cap, not Hermitian conjugation.  With these definitions,
\begin{equation}
z_{\rm RR}^{2d}(u)
=
\mathcal B_{\mathrm R}^{2d}(u)\,
\mathcal B_{\mathrm R}^{2d,\vee}(u),
\qquad
z_{\rm DD}^{2d}(u)
=
\mathcal B_{\mathrm D}^{2d}(u)\,
\mathcal B_{\mathrm D}^{2d,\vee}(u),
\label{eq:2d-block-factorization}
\end{equation}
and the reflected inverse relation already holds at the level of the cap wavefunctions:
\begin{equation}
\mathcal B_{\mathrm D}^{2d}(u)\,
\mathcal B_{\mathrm R}^{2d}(-u)=1,
\qquad
\mathcal B_{\mathrm D}^{2d,\vee}(u)\,
\mathcal B_{\mathrm R}^{2d,\vee}(-u)=1.
\label{eq:2d-cap-wavefunction-inverse}
\end{equation}
The gluing identities are the two-dimensional degeneration of \eqref{eq:cap-dual-pairings}, with the unit shift replacing the elliptic shift \(\omega\):
\begin{equation}
\mathcal B_{\mathrm R}^{2d,\vee}(u)\,
\mathcal B_{\mathrm D}^{2d}(u-1)=1,
\qquad
\mathcal B_{\mathrm D}^{2d,\vee}(u)\,
\mathcal B_{\mathrm R}^{2d}(u+1)=1.
\label{eq:2d-cap-dual-pairings}
\end{equation}

Thus the boundary-to-boundary interpretation also applies to the two-dimensional cylinder \(\cyl\simeq I\times S^1\).  In a formal localized boundary basis \(|u;\mathsf q\rangle\), with \(u\sim u+1\), the elementary two-dimensional kernels can be written as
\begin{equation}
\mathsf K_{\mathsf q_+\mathsf q_-}^{2d}
=
\int_{\mathbb C/\mathbb Z}\mathrm d\mu(u)\,
|u;\mathsf q_+\rangle\,
z_{\mathsf q_+\mathsf q_-}^{2d}(u)\,
\langle u+\mu_{\mathsf q_+\mathsf q_-}^{2d};\mathsf q_-^{\vee}|,
\label{eq:2d-cylinder-kernel-operator}
\end{equation}
where
\begin{equation}
\mu_{\mathrm{RR}}^{2d}=0,
\qquad
\mu_{\mathrm{DD}}^{2d}=0,
\qquad
\mu_{\mathrm{DR}}^{2d}=1,
\qquad
\mu_{\mathrm{RD}}^{2d}=-1.
\label{eq:2d-cylinder-kernel-shifts}
\end{equation}
The component form is
\begin{equation}
z_{\mathsf q_+\mathsf q_-}^{2d}(u)
=
\mathcal B_{\mathsf q_+}^{2d}(u)\,
\mathcal B_{\mathsf q_-}^{2d,\vee}\!\left(u+\mu_{\mathsf q_+\mathsf q_-}^{2d}\right).
\label{eq:2d-cylinder-kernel-components}
\end{equation}
Consequently,
\begin{equation}
z_{\rm DR}^{2d}(u)
=
\mathcal B_{\mathrm D}^{2d}(u)\,
\mathcal B_{\mathrm R}^{2d,\vee}(u+1)=1,
\qquad
z_{\rm RD}^{2d}(u)
=
\mathcal B_{\mathrm R}^{2d}(u)\,
\mathcal B_{\mathrm D}^{2d,\vee}(u-1)=1,
\label{eq:2d-mixed-kernels-identity}
\end{equation}
while
\begin{equation}
z_{\rm DD}^{2d}(u)\,z_{\rm RR}^{2d}(u+1)=1,
\qquad
z_{\rm RR}^{2d}(u)\,z_{\rm DD}^{2d}(u-1)=1.
\label{eq:2d-kernel-component-composition}
\end{equation}
These equations are the two-dimensional analogues of \eqref{eq:cyl-kernel-identity-polarization} and \eqref{eq:cyl-kernel-component-composition}.  They show that further slicing of the two-dimensional cylinder is again represented by composition of identity kernels and polarization-changing interfaces.  The massive-pair identity degenerates in the same way:
\begin{equation}
z_{\rm RR}^{2d}(u)z_{\rm DD}^{2d}(-u)=1,
\qquad
z_{\rm DD}^{2d}(u)z_{\rm RR}^{2d}(-u)=1.
\label{eq:2d-massive-pair-identity}
\end{equation}

For the complete-intersection GLSM, the cylinder integrand
\begin{equation}
Z_{\mathfrak n_1,\ldots,\mathfrak n_r}^{2d}(u)
=
\prod_{A=1}^{r}
z_{\rm RR}^{2d}\bigl[-\mathfrak n_Au+\mathfrak p_A\bigr]
\prod_{i=1}^{N}
z_{\rm RR}^{2d}\bigl[u+x_i\bigr]
\end{equation}
factorizes as
\begin{equation}
Z_{\mathfrak n_1,\ldots,\mathfrak n_r}^{2d}(u)
=
Z_{\mathfrak n_1,\ldots,\mathfrak n_r}^{2d,{\rm hol}}(u)\,
Z_{\mathfrak n_1,\ldots,\mathfrak n_r}^{2d,{\rm anti}}(u),
\label{eq:CI-2d-factorization}
\end{equation}
where the two factors are now written directly in terms of the cap wavefunctions,
\begin{equation}
Z_{\mathfrak n_1,\ldots,\mathfrak n_r}^{2d,{\rm hol}}(u)
=
\prod_{A=1}^{r}
\mathcal B_{\mathrm R}^{2d}\bigl[-\mathfrak n_Au+\mathfrak p_A\bigr]
\prod_{i=1}^{N}
\mathcal B_{\mathrm R}^{2d}\bigl[u+x_i\bigr],
\label{eq:CI-2d-hol-block}
\end{equation}
and
\begin{equation}
Z_{\mathfrak n_1,\ldots,\mathfrak n_r}^{2d,{\rm anti}}(u)
=
\prod_{A=1}^{r}
\mathcal B_{\mathrm R}^{2d,\vee}\bigl[-\mathfrak n_Au+\mathfrak p_A\bigr]
\prod_{i=1}^{N}
\mathcal B_{\mathrm R}^{2d,\vee}\bigl[u+x_i\bigr].
\label{eq:CI-2d-anti-block}
\end{equation}

The chamber equality
\begin{equation}
\operatorname{JK}_{+}
Z_{\mathfrak n_1,\ldots,\mathfrak n_r}^{2d}
=
\operatorname{JK}_{-}
Z_{\mathfrak n_1,\ldots,\mathfrak n_r}^{2d}
\end{equation}
can therefore be interpreted as an equality between two residue decompositions of the same factorized cylinder kernel. In the positive chamber, the poles are
\begin{equation}
u_i^+=-x_i,
\qquad
i=1,\ldots,N,
\end{equation}
and give the geometric-phase expansion. In the negative chamber, the poles are
\begin{equation}
u_{A,\ell}^-=
\frac{\mathfrak p_A-\ell}{\mathfrak n_A},
\qquad
\ell=0,\ldots,\mathfrak n_A-1,
\end{equation}
and give the hybrid or Landau--Ginzburg expansion.

Thus the two-dimensional identity has the same kernel structure.

The cylinder partition function is a two-boundary kernel.  Its residue expansion is a sum of factorized contributions,
\begin{equation}
Z_{\cyl}
\sim
\sum_{\alpha}
B_{\alpha}^{\rm hol}\,
B_{\alpha}^{\rm anti}.
\end{equation}
For \(\cyl\times T^2\), the cap wavefunctions are built from elliptic Gamma functions. For \(\cyl\simeq I\times S^1\), the cap wavefunctions are built from Euler Gamma functions.

\subsection{Holographic theories on \texorpdfstring{\(\cyl\times T^2\)}{Cyl x T2}}
\label{sec:holographic-benchmarks}

Two standard four-dimensional matter contents provide useful benchmarks for theories that also have holographic incarnations: \(\mathcal N=4\) SYM with gauge group \(SU(N)\)
and the Klebanov--Witten conifold theory~\cite{Klebanov:1998hh} with gauge group \(SU(N)\times SU(N)\).
The point we emphasize is that the naive assignment in which all chiral multiplets
are placed in \(RR\) polarization does not define a \(\omega\)-elliptic meromorphic
integrand on the \(\omega\)-elliptic gauge-fugacity torus of \(\cyl\times T^2\).  However, by changing the
boundary polarization of a suitable subset of chiral multiplets, equivalently by
replacing some \(z_{\rm RR}\) factors by \(z_{\rm DD}\) factors, one obtains
\(\omega\)-elliptic cylinder integrands.

We assign a sign
\begin{equation}
\sigma_{\rm ch}=
\begin{cases}
+1, & \text{chiral in }RR\text{ polarization},\\
-1, & \text{chiral in }DD\text{ polarization}.
\end{cases}
\label{eq:polarization-sign}
\end{equation}
The gauge-dependent part of the \(\omega\)-elliptic monodromy is controlled by the
quadratic terms in the Bernoulli phases together with the theta-function
quasi-periodicities.  In the examples below, cancellation of this gauge-dependent
monodromy imposes simple linear constraints on the signs \(\sigma_{\rm ch}\).

\subsubsection{\texorpdfstring{\(\mathcal N=4\) SYM with gauge group \(SU(N)\)}{N=4 SYM with gauge group SU(N)}}

We first consider \(\mathcal N=4\) SYM in \(\mathcal N=1\) language, one of the canonical gauge-theory examples in the AdS/CFT correspondence~\cite{Maldacena:1997re}.  The theory
contains one vector multiplet and three adjoint chiral multiplets
\begin{equation}
X_I\in {\rm Adj}_{SU(N)},
\qquad
I=1,2,3,
\end{equation}
with superpotential
\begin{equation}
W_{\mathcal N=4}
=
\operatorname{Tr}X_1[X_2,X_3].
\label{eq:N4-superpotential}
\end{equation}
Let
\begin{equation}
u_a,\qquad a=1,\ldots,N,
\end{equation}
be the \(SU(N)\) gauge fugacities, subject to
\begin{equation}
\sum_{a=1}^{N}u_a=0.
\label{eq:N4-SU-constraint}
\end{equation}

The naive cylinder integrand, with all three adjoint chirals in \(RR\) polarization,
is
\begin{align}
Z_{\mathcal N=4}^{\rm naive}
={}&
\frac{1}{N!}
\left[
z_{\rm RR}(\Delta_1)
z_{\rm RR}(\Delta_2)
z_{\rm RR}(\Delta_3)
\right]^{N-1}
\nonumber\\
&\times
\prod_{\substack{a,b=1\\a\neq b}}^{N}
z_{\rm DD}(u_{ab})
\prod_{I=1}^{3}
z_{\rm RR}(u_a-u_b+\Delta_I).
\label{eq:N4-naive-all-RR}
\end{align}
This expression is not \(\omega\)-elliptic on the \(SU(N)\) gauge fugacity torus.  Indeed, the
three adjoint chiral multiplets have polarization signs
\begin{equation}
\sigma_1+\sigma_2+\sigma_3=1+1+1=3,
\end{equation}
whereas cancellation of the gauge-dependent \(\omega\)-elliptic monodromy of the vector multiplet
requires
\begin{equation}
\sigma_1+\sigma_2+\sigma_3=1.
\label{eq:N4-SV-condition}
\end{equation}
Thus the all-\(RR\) assignment over-cancels the vector contribution and leaves a
nontrivial \(\omega\)-elliptic multiplier.

An \(\omega\)-elliptic assignment is obtained by taking two adjoint chirals in \(RR\)
polarization and one in \(DD\) polarization.  We choose
\begin{equation}
X_1,X_2\in RR,
\qquad
X_3\in DD.
\label{eq:N4-RRDD-polarization}
\end{equation}
Then
\begin{equation}
\sigma_1+\sigma_2+\sigma_3=1+1-1=1,
\end{equation}
and the condition \eqref{eq:N4-SV-condition} is satisfied.  The corresponding
\(\omega\)-elliptic \(SU(N)\) contour integral is the one displayed in
\eqref{eq:intro-N4-SU-partition}.  The factor raised to the power \(N-1\)
there is the Cartan contribution of the three adjoint chirals of \(SU(N)\), while
the product over \(a\neq b\) gives the contribution of the roots.  The
superpotential imposes the usual flavor balancing condition
\begin{equation}
\Delta_1+\Delta_2+\Delta_3=\Delta_W
\qquad
\mathrm{mod}\ \mathbb Z+\omega\mathbb Z.
\label{eq:N4-superpotential-balancing}
\end{equation}
This condition is independent of the cancellation of gauge \(\omega\)-elliptic
monodromies; it expresses the neutrality of the superpotential
\eqref{eq:N4-superpotential}.

\subsubsection{Klebanov--Witten theory}
We now consider the Klebanov--Witten conifold quiver~\cite{Klebanov:1998hh}.  The gauge group is
\begin{equation}
SU(N)_u\times SU(N)_v,
\end{equation}
with gauge fugacities
\begin{equation}
u_a,\qquad v_b,\qquad a,b=1,\ldots,N,
\end{equation}
subject to
\begin{equation}
\sum_{a=1}^{N}u_a=0,
\qquad
\sum_{b=1}^{N}v_b=0.
\label{eq:KW-SU-constraints-simple}
\end{equation}
The matter fields are
\begin{equation}
A_i\in(\mathbf N_u,\overline{\mathbf N}_v),
\qquad i=1,2,
\end{equation}
and
\begin{equation}
B_j\in(\overline{\mathbf N}_u,\mathbf N_v),
\qquad j=1,2.
\end{equation}
The superpotential is
\begin{equation}
W_{\rm KW}
=
\operatorname{Tr}
\left(
A_1B_1A_2B_2
-
A_1B_2A_2B_1
\right).
\label{eq:KW-superpotential-simple}
\end{equation}

The naive all-\(RR\) integrand is
\begin{align}
Z_{\rm KW}^{\rm naive}
={}&
\prod_{\substack{a,b=1\\a\neq b}}^{N}
z_{\rm DD}(u_{ab})
\prod_{\substack{a,b=1\\a\neq b}}^{N}
z_{\rm DD}(v_a-v_b)
\nonumber\\
&\times
\prod_{a,b=1}^{N}
z_{\rm RR}(u_a-v_b+\alpha_1)\,
z_{\rm RR}(u_a-v_b+\alpha_2)
\nonumber\\
&\times
\prod_{a,b=1}^{N}
z_{\rm RR}(-u_a+v_b+\beta_1)\,
z_{\rm RR}(-u_a+v_b+\beta_2).
\label{eq:KW-naive-all-RR}
\end{align}
This expression is not \(\omega\)-elliptic.  For the conifold quiver with equal ranks
\(N=N_u=N_v\), cancellation of the gauge-dependent monodromy of the two vector
multiplets requires
\begin{equation}
\sigma_{A_1}+\sigma_{A_2}+\sigma_{B_1}+\sigma_{B_2}=2.
\label{eq:KW-SV-condition}
\end{equation}
The naive all-\(RR\) assignment gives instead
\begin{equation}
1+1+1+1=4,
\end{equation}
and hence leaves a nontrivial \(\omega\)-elliptic multiplier.

An \(\omega\)-elliptic polarization is obtained by taking three bifundamentals in \(RR\)
polarization and one in \(DD\) polarization.  For definiteness we choose
\begin{equation}
A_1,A_2,B_1\in RR,
\qquad
B_2\in DD.
\label{eq:KW-RRRD-polarization-simple}
\end{equation}
Then
\begin{equation}
\sigma_{A_1}+\sigma_{A_2}+\sigma_{B_1}+\sigma_{B_2}
=
1+1+1-1=2,
\end{equation}
and \eqref{eq:KW-SV-condition} is satisfied.  The corresponding \(\omega\)-elliptic
cylinder integrand is
\begin{align}
Z_{\rm KW}^{SU(N)\times SU(N)}
={}&
\prod_{\substack{a,b=1\\a\neq b}}^{N}
z_{\rm DD}(u_{ab})
\prod_{\substack{a,b=1\\a\neq b}}^{N}
z_{\rm DD}(v_a-v_b)
\nonumber\\
&\times
\prod_{a,b=1}^{N}
z_{\rm RR}(u_a-v_b+\alpha_1)\,
z_{\rm RR}(u_a-v_b+\alpha_2)
\nonumber\\
&\times
\prod_{a,b=1}^{N}
z_{\rm RR}(-u_a+v_b+\beta_1)\,
z_{\rm DD}(-u_a+v_b+\beta_2),
\label{eq:KW-single-valued-integrand}
\end{align}
with the \(SU(N)\times SU(N)\) constraints
\eqref{eq:KW-SU-constraints-simple}.  The superpotential imposes the balancing
condition
\begin{equation}
\alpha_1+\alpha_2+\beta_1+\beta_2=\Delta_W
\qquad
\mathrm{mod}\ \mathbb Z+\omega\mathbb Z.
\label{eq:KW-balancing-simple}
\end{equation}

The choice \eqref{eq:KW-RRRD-polarization-simple} is not unique.  Any assignment
with three bifundamentals in \(RR\) polarization and one in \(DD\) polarization
satisfies \eqref{eq:KW-SV-condition}.  Different choices make different subgroups
of the flavor symmetry manifest, but they are equivalent from the viewpoint of
\(\omega\)-ellipticity.

\appendix

\section{Conventions and Auxiliary Formulae}\label{app:notation-conventions}

This appendix collects conventions and auxiliary formulae used in the main text.

\subsection{Spinors and sigma matrices}

We use two-component Euclidean spinors.  Undotted spinors, such as \(\zeta_\alpha\), transform in the \((\mathbf 2,\mathbf 1)\) representation of \(\operatorname{Spin}(4)=SU(2)_+\times SU(2)_-\), while dotted spinors, such as \(\widetilde\zeta_{\dot\alpha}\), transform in \((\mathbf 1,\mathbf 2)\).  Hermitian conjugation gives spinors with the conjugate index structure, and the norms are
\begin{equation}
|\zeta|^2=\zeta^\dagger_{\alpha}\zeta^\alpha,
\qquad
|\widetilde\zeta|^2=\widetilde\zeta^\dagger_{\dot\alpha}\widetilde\zeta^{\dot\alpha}.
\label{eq:app-spinor-norms}
\end{equation}
Spinor indices are raised and lowered with antisymmetric tensors \(\epsilon_{\alpha\beta}\), \(\epsilon^{\alpha\beta}\), \(\epsilon_{\dot\alpha\dot\beta}\), \(\epsilon^{\dot\alpha\dot\beta}\).  We use
\begin{equation}
\epsilon^{12}=+1,\qquad \epsilon_{12}=-1,
\qquad
\epsilon^{\alpha\beta}\epsilon_{\beta\gamma}=\delta^\alpha{}_{\gamma},
\qquad
\epsilon_{\alpha\beta}\epsilon^{\beta\gamma}=\delta_\alpha{}^{\gamma},
\label{eq:app-epsilon-conventions}
\end{equation}
and similarly for dotted indices.  These are the raising and lowering conventions used in the Wolfram Mathematica checks of the background.  We contract spinors as \(\zeta\chi=\zeta^\alpha\chi_\alpha\) and \(\widetilde\zeta\widetilde\chi=\widetilde\zeta_{\dot\alpha}\widetilde\chi^{\dot\alpha}\).

In an orthonormal frame the Euclidean sigma matrices are
\begin{equation}
\sigma^a=(\sigma^1,\sigma^2,\sigma^3,-\ii\mathbf{1}_2),
\qquad
\widetilde\sigma^a=(-\sigma^1,-\sigma^2,-\sigma^3,-\ii\mathbf{1}_2),
\label{eq:app-sigma-frame}
\end{equation}
where the first three matrices on the right-hand side are the Pauli matrices.  Curved sigma matrices are obtained by contraction with the frame and inverse frame:
\begin{equation}
\sigma^\mu=E_a^{\ \mu}\sigma^a,
\qquad
\widetilde\sigma^\mu=E_a^{\ \mu}\widetilde\sigma^a,
\qquad
\sigma_\mu=e^a_{\ \mu}\sigma_a,
\qquad
\widetilde\sigma_\mu=e^a_{\ \mu}\widetilde\sigma_a.
\label{eq:app-curved-sigma}
\end{equation}
The generators acting on the two chiralities are
\begin{equation}
\sigma^{ab}:=\frac14(\sigma^a\widetilde\sigma^b-\sigma^b\widetilde\sigma^a),
\qquad
\widetilde\sigma^{ab}:=\frac14(\widetilde\sigma^a\sigma^b-\widetilde\sigma^b\sigma^a).
\label{eq:sigmaab}
\end{equation}
With \(\varepsilon_{1234}=+1\), they obey
\begin{equation}
\frac12\varepsilon_{abcd}\sigma^{cd}=\sigma_{ab},
\qquad
\frac12\varepsilon_{abcd}\widetilde\sigma^{cd}=-\widetilde\sigma_{ab}.
\label{eq:app-sigma-duality}
\end{equation}
The Clifford relations and hermiticity properties are
\begin{equation}
\sigma^a\widetilde\sigma^b+\sigma^b\widetilde\sigma^a=-2\delta^{ab}\mathbf 1_2,
\qquad
\widetilde\sigma^a\sigma^b+\widetilde\sigma^b\sigma^a=-2\delta^{ab}\mathbf 1_2,
\label{eq:app-clifford}
\end{equation}
\begin{equation}
(\sigma^a)^\dagger=-\widetilde\sigma^a,
\qquad
(\sigma^{ab})^\dagger=-\sigma^{ab},
\qquad
(\widetilde\sigma^{ab})^\dagger=-\widetilde\sigma^{ab}.
\label{eq:app-sigma-hermiticity}
\end{equation}
\begin{equation}
\sigma^a_{\alpha\dot\alpha}\widetilde\sigma_{a}^{\dot\beta\beta}
=-2\,\delta_\alpha{}^{\beta}\delta_{\dot\alpha}{}^{\dot\beta},
\label{eq:app-sigma-completeness}
\end{equation}
\begin{align}
\sigma^a\widetilde\sigma^b\sigma^c
&=-\delta^{ab}\sigma^c+\delta^{ac}\sigma^b-\delta^{bc}\sigma^a
+\varepsilon^{abcd}\sigma_d,
\nn\\
\widetilde\sigma^a\sigma^b\widetilde\sigma^c
&=-\delta^{ab}\widetilde\sigma^c+\delta^{ac}\widetilde\sigma^b-\delta^{bc}\widetilde\sigma^a
-\varepsilon^{abcd}\widetilde\sigma_d.
\label{eq:app-triple-sigma}
\end{align}
For commuting supersymmetry parameters one has the exchange rules
\begin{equation}
\zeta\chi=-\chi\zeta,
\qquad
\widetilde\zeta\widetilde\chi=-\widetilde\chi\widetilde\zeta,
\qquad
\zeta\sigma^a\widetilde\chi=\widetilde\chi\widetilde\sigma^a\zeta,
\qquad
\zeta\sigma^{ab}\chi=\chi\sigma^{ab}\zeta.
\label{eq:app-commuting-spinor-rules}
\end{equation}
If two dynamical fermions are exchanged, the additional minus sign from their Grassmann parity must also be included.  A useful Fierz identity, in these conventions, is
\begin{equation}
(\zeta\chi)(\widetilde\zeta\widetilde\chi)
=-\frac12(\zeta\sigma^a\widetilde\chi)(\chi\sigma_a\widetilde\zeta).
\label{eq:app-Fierz}
\end{equation}

The spinor covariant derivatives are
\begin{equation}
\nabla_\mu\zeta=\partial_\mu\zeta-\frac12\omega_{\mu ab}\sigma^{ab}\zeta,
\qquad
\nabla_\mu\widetilde\zeta=\partial_\mu\widetilde\zeta-\frac12\omega_{\mu ab}\widetilde\sigma^{ab}\widetilde\zeta.
\label{eq:spinor-cov-der}
\end{equation}
The spin connection is determined by the torsion-free Cartan equation
\begin{equation}
\dd e^a+\omega^a{}_{b}\wedge e^b=0,
\label{eq:app-Cartan}
\end{equation}
Equivalently, in components, with \(E_a{}^\mu\) the inverse frame,
\begin{equation}
\omega_{\mu ab}
=-e_{b\nu}\left(\partial_\mu E_a{}^\nu+\Gamma^\nu{}_{\mu\lambda}E_a{}^\lambda\right),
\label{eq:app-spin-connection-component}
\end{equation}
which is the component convention used in the symbolic checks.  The curvature convention is
\begin{equation}
R_{\mu\nu}{}^{ab}
=\partial_\mu\omega_\nu{}^{ab}-\partial_\nu\omega_\mu{}^{ab}
+\omega_\mu{}^{ac}\omega_{\nu c}{}^{b}
-\omega_\nu{}^{ac}\omega_{\mu c}{}^{b}.
\label{eq:app-curvature}
\end{equation}

\subsection{Shintani--Barnes regularization}\label{app:Barnes}

We define the $q$-Pochhammer symbol by
\begin{equation}
(z;q)_\infty
:=
\prod_{k=0}^{\infty}(1-z q^k).
\label{eq:q-pochhammer}
\end{equation}
For additive variables we use throughout the theta function
\begin{equation}
\Theta(u;\omega):=
(e^{2\pi\ii u};e^{2\pi\ii\omega})_\infty
(e^{2\pi\ii(\omega-u)};e^{2\pi\ii\omega})_\infty .
\label{eq:additive-Theta-def}
\end{equation}
and similarly for the double $q$-Pochhammer symbol and elliptic Gamma function
\begin{equation}
(z;q,p)_\infty := \prod_{k,j=0}^\infty (1-z\,q^k p^j),
\qquad
\Gamma(z;q,p):=\frac{(pq z^{-1};q,p)_\infty}{(z;q,p)_\infty}~.
\label{eq:poch2-Gamma}
\end{equation}
A useful property is
\begin{equation}
\frac{\Gamma(p z;q,p)}{\Gamma(z;q,p)}=\Theta(u;\omega)~,\qquad z=e^{2\pi\ii u},\quad q=e^{2\pi\ii\omega}~.
\end{equation}
Define the fugacities
\begin{equation}
z:=e^{2\pi \ii u}~,\quad q := e^{2\pi \ii \omega}~,\quad p:=e^{2\pi\ii \tau}~.
\label{eq:qdef}
\end{equation}
The product representations in \eqref{eq:q-pochhammer}--\eqref{eq:additive-Theta-def} converge for $|q|,|p|<1$ (equivalently $\Im(\omega),\Im(\tau)>0$), while the final expressions below define the analytically continued regularization for generic $\omega,\tau$.  The regularization is ordered: we first regularize the torus lattice and then the half-cylinder product.  This ordering fixes the Bernoulli phase, and a different ordering would shift the answer by a holomorphic non-vanishing contact factor, equivalently by a choice of local counterterm.  The symbols \(C_{4d}(\omega,\tau)\), \(C_{3d}(\omega)\) and \(C_{2d}\) below denote such scheme-dependent normalizations independent of the holonomy variable. We also use the additive elliptic-Gamma notation
\begin{equation}
\Gamma_e(u;\tau,\omega):=\Gamma(z;p,q)~.
\label{eq:additive-notation}
\end{equation}

The fundamental unregularized product we encounter in the computation of 1-loop determinants on $D^2\times T^2$ or $\cyl\times T^2$ is 
\begin{equation}
\prod_{m_\varphi\geq 0}\prod_{(m_x,m_y)\in\mathbb{Z}^2}
\Big(\omega m_\varphi + \tau m_x-m_y+u\Big)~.
\end{equation}
This product is to be understood in the Shintani--Barnes sense, i.e.\ by analytic continuation of the Hurwitz multiple zeta function and its associated multiple Shintani-Barnes Gamma functions~\cite{Barnes1904,Shintani1976,Friedman:2004}.  The elliptic-Gamma form and the accompanying Bernoulli phase are consistent with the modular and product formulae of Refs.~\cite{Felder:1999mq,Narukawa:2003}.  Namely,

\begin{equation}
\prod_{m_\varphi\geq 0}\prod_{(m_x,m_y)\in\mathbb{Z}^2}
\Big(\omega m_\varphi + \tau m_x-m_y+\mu\Big) \ \stackrel{\rm SB}{=}\ C_{4d}(\omega,\tau) \ \frac{e^{-\ii\frac{\pi}{3}B_{3,3}(\mu\mid 1,\tau,\omega)}}{\Gamma(e^{2\pi\ii \mu};e^{2\pi\ii \tau},e^{2\pi\ii \omega})}~,
\end{equation}
where $C_{4d}(\omega,\tau)$ depends only on $\omega,\tau$ (and on the chosen normalization of the Barnes double Gamma) and is independent of $\mu$. Here we defined the cubic Bernoulli polynomial
\begin{multline}
B_{3,3}(u;\omega_1,\omega_2,\omega_3) 
= \frac{u^3}{\omega _1 \omega _2 \omega _3}-\frac{3 u^2 \left(\omega _1+\omega _2+\omega _3\right)}{2 \omega _1 \omega _2 \omega _3}+\\
+\frac{u \left(\omega _1^2+\omega _2^2+\omega _3^2+3 \left(\omega _1 \omega _2+\omega _3 \omega _2+\omega _1 \omega _3\right)\right)}{2 \omega _1 \omega _2 \omega _3}+\\
 - \frac{\left(\omega _1+\omega _2+\omega _3\right) \left(\omega _1 \omega _2+\omega _3 \omega _2+\omega _1 \omega _3\right)}{4 \omega _1 \omega _2 \omega _3}
\end{multline}

\subsection{Jeffrey--Kirwan residue conventions}\label{app:JK-residue-conventions}

This appendix fixes the residue convention used throughout the paper.  The purpose is only to make the normalization of the contour measure explicit.  In particular, every occurrence of
\(
\operatorname{JK}_{U(r),\eta} Z
\)
in the main text is normalized with the measure
\begin{equation}
\prod_{a=1}^{r}\frac{\dd u_a}{2\pi \ii} .
\label{eq:JK-normalized-measure}
\end{equation}
Thus no extra factor of \((2\pi\ii)^r\) is missing from the residue formulae in the duality identities.

Let \(u=(u_1,\ldots,u_r)\) be local coordinates near a pole \(u_*\), and suppose that the polar hyperplanes passing through \(u_*\) are locally given by
\begin{equation}
Q_j(u-u_*)=0,\qquad j=1,\ldots,r,
\label{eq:JK-local-hyperplanes}
\end{equation}
where the charge covectors \(Q_j\in(\mathbb R^r)^*\) are linearly independent.  Near such a non-degenerate pole, a meromorphic \(r\)-form can be written as
\begin{equation}
\omega(u)
=
\frac{f(u)\,\dd u_1\wedge\cdots\wedge\dd u_r}
{Q_1(u-u_*)\cdots Q_r(u-u_*)}+\text{less singular terms},
\label{eq:JK-local-form}
\end{equation}
with \(f\) holomorphic at \(u_*\).  The Jeffrey--Kirwan residue in chamber \(\eta\in\mathbb R^r\) is
\begin{equation}
\operatorname*{JK\mbox{-}Res}_{u=u_*,\eta}\,\omega
=
\begin{cases}
\displaystyle
\frac{f(u_*)}{\left|\det(Q_1,\ldots,Q_r)\right|},
& \eta\in\operatorname{Cone}(Q_1,\,\ldots,Q_r),\\[3mm]
0,& \eta\notin\operatorname{Cone}(Q_1,\,\ldots,Q_r),
\end{cases}
\label{eq:JK-local-definition}
\end{equation}
up to the orientation convention chosen for the ordered set of charges.  In the convention used in the main text, the orientation is fixed so that a simple positive-charge pole of the one-dimensional factor \(z_{\rm RR}(u)\) contributes positively.

Equivalently, for an isolated pole selected by the chamber, the JK residue can be represented by the normalized contour integral
\begin{equation}
\operatorname*{JK\mbox{-}Res}_{u=u_*,\eta}\,\omega
=
\int_{\Gamma_{u_*,\eta}}
\prod_{a=1}^{r}\frac{\dd u_a}{2\pi\ii}\,
\frac{f(u)}{Q_1(u-u_*)\cdots Q_r(u-u_*)},
\label{eq:JK-contour-definition}
\end{equation}
where \(\Gamma_{u_*,\eta}\) is the small real \(r\)-torus encircling the selected hyperplanes with the orientation compatible with \(\eta\).  If one wrote the same contour integral using \(\prod_a\dd u_a\) instead, the answer would be multiplied by \((2\pi\ii)^r\).

For the cylinder blocks, the relevant simple poles come from the RR factor
\begin{equation}
z_{\rm RR}(u)
=
\frac{\exp\!\left[\frac{\ii\pi}{3}P_3(u)\right]}{\Theta(u;\omega)} .
\label{eq:JK-RR-block-reminder}
\end{equation}
Since
\begin{equation}
\Theta(u;\omega)=\Theta'_0\,u+O(u^2),
\qquad
\Theta'_0=\left.\frac{\partial}{\partial u}\Theta(u;\omega)\right|_{u=0},
\label{eq:JK-theta-local-expansion}
\end{equation}
one has
\begin{equation}
z_{\rm RR}(u)
=
\frac{\varphi_0}{\Theta'_0}\,\frac{1}{u}+O(1),
\qquad
\varphi_0:=
\exp\!\left[\frac{\ii\pi}{3}P_3(0)\right].
\label{eq:JK-simple-RR-pole}
\end{equation}
Therefore a selected simple pole of a positive-charge RR hyperplane contributes
\begin{equation}
\frac{\varphi_0}{\Theta'_0}
\label{eq:JK-RR-simple-factor}
\end{equation}
in the normalized JK convention.  This is the origin of the powers of \(\varphi_0/\Theta'_0\) in the residue formulae in the main text.

The theta-function confining identity in Subsection~\ref{subsubsec:theta-U0-identity} is written directly as a contour integral rather than as the abstract symbol \(\operatorname{JK}\).  For this reason its measure is displayed explicitly as
\begin{equation}
\prod_{a=1}^{N}\frac{\dd u_a}{2\pi\ii} .
\label{eq:JK-confining-measure-reminder}
\end{equation}
With this normalization the theta identity has the right-hand side shown in the main text.  With the unnormalized measure \(\prod_a\dd u_a\), the left-hand side would be multiplied by \((2\pi\ii)^N\).

\section{Anomaly Equations}\label{app:anomalies}

In this appendix the periods of the cubic Bernoulli polynomial are kept arbitrary.  We use the same symbol as in the main text, namely
\begin{equation}
B_3(u):=B_{3,3}(u;\omega_1,\omega_2,\omega_3)~,\quad P_3(x):=B_3(x)+B_3(\omega_3-x)~.
\end{equation}
A Robin block with argument \(u\) contributes \(P_3(u)\), while a Dirichlet block with argument \(u\) contributes the negative of the same expression evaluated at \(-u\).  Whenever a DD multiplet appears, its contribution is obtained by applying this reflection to its actual block argument.  This is the anomaly-polynomial version of the inverse relation \(z_{\rm DD}(u)z_{\rm RR}(-u)=1\).
For the $U(N_c)$ vector contribution we get the polynomial
\begin{equation}
\sum_{i,j=1}^{N_c}P_3(x_i-x_j)=2A
\left[
N_c\sum_i x_i^2-\left(\sum_i x_i\right)^2
\right]
+N_c^2 C
\end{equation}
where 
\begin{equation}
A:=  -\frac{3(\omega_1+\omega_2)}
{\omega_1\omega_2\omega_3}~,\quad C:= 
-\frac{(\omega_1+\omega_2)(\omega_1\omega_2+\omega_3^2)}
{2\omega_1\omega_2\omega_3}
~.
\end{equation}
For an adjoint DD/RR multiplet we get
\begin{align}
\pm\sum_{i,j=1}^{N_c}
P_3(u_i-u_j+t)
 =
\pm 2A\left[N_c\sum_i u_i^2
-\left(\sum_i u_i\right)^2\right]
\pm
N_c^2\left[At(t-\omega_3)+C\right]~.
\end{align}
For a DD/RR determinant multiplet we get
\begin{align}
P_3\Bigl(\sum_{i=1}^{N_c}u_i+M\Bigr)=\pm A
\left[
\left(\sum_i u_i\right)^2
+(2M-\omega_3)\sum_i u_i\right]
\pm \left[AM(M-\omega_3)
+C\right]~,
\end{align}
and similarly for a DD/RR anti-determinant multiplet 
\begin{align}
\pm P_3\Bigl(-\sum_{i=1}^{N_c}u_i+\bar M\Bigr)
=\pm A
\left[
\left(\sum_i u_i\right)^2
-(2\bar M-\omega_3)\sum_i u_i\right]
\pm \left[A\bar M(\bar M-\omega_3)
+C\right]~.
\end{align}
For a fundamental DD/RR multiplet we get the contribution
\begin{equation}
\pm\sum_i P_3(u_i+m)
=
\pm A\left[\sum_i u_i^2
+
(2m-\omega_3)\sum_i u_i\right]\pm N_c\left[Am(m-\omega_3)+C\right]~,
\end{equation}
and similarly for an anti-fundamental DD/RR multiplet
\begin{equation}
\pm\sum_i P_3(-u_i+\bar m)
=
\pm A\left[\sum_i u_i^2
-(2\bar m-\omega_3)\sum_i u_i\right]\pm N_c\left[A\bar m(\bar m-\omega_3)+C\right]~.
\end{equation}
Consider a vector, adjoints ($N_\text{ad}$), (anti-)determinants ($N_\text{d}$, $\bar N_\text{d}$) and (anti-)fundamentals ($N_\text{f}$, $\bar N_\text{f}$), contributing with signs (i.e. DD/RR conditions) $\epsilon_a, \eta_d,\bar\eta_{\bar d},\rho_f, \bar\rho_{\bar f}$ respectively. In order to cancel the quadratic anomaly we need the conditions
\begin{align}
& \ 2N_c (1+DD_\text{ad}-RR_\text{ad})+(DD_\text{f}-RR_\text{f})+(\bar DD_\text{f}-\bar RR_\text{f})=0~,\\
& \ 2(1+DD_\text{ad}-RR_\text{ad})=(DD_\text{d}-RR_\text{d})+(\bar DD_\text{d}-\bar RR_\text{d})~.
\end{align}
For the linear anomaly we need to fulfill the condition
\begin{align}
& 
\sum_d\eta_d M_d
-\sum_{\bar d}\bar\eta_{\bar d}\bar M_{\bar d}
+\sum_f\rho_f m_f
-\sum_{\bar f}\bar\rho_{\bar f}\bar m_{\bar f}=0~,
\end{align}
where we shifted all the masses by $\omega_3/2$ for convenience.  An option is 2 RR and 1 DD adjoints (e.g. $\mathcal{N}=4$). A related flavor-adjoint option requires one RR adjoint, not one DD adjoint, together with $N$ RR fundamentals and $N$ DD anti-fundamentals. Another possibility is the theory without adjoints, 1 DD determinant and 1 DD anti-determinant and fundamentals/anti-fundamentals satisfying
\begin{equation}
RR_\text{f}+\overline{RR}_\text{f}-DD_\text{f}-\overline{DD}_\text{f}=2N_c~.
\end{equation}
and the linear constraint 
\begin{equation}
M-\bar M+\sum_{f=1}^{DD_\text{f}} m^{DD}_f-\sum_{f=1}^{RR_\text{f}} m^{RR}_f-\sum_{f=1}^{\bar DD_\text{f}} \bar m^{DD}_f+\sum_{f=1}^{\bar RR_\text{f}} \bar m^{RR}_f=0~.
\end{equation}

The determinant-pair duality used in the main text is one concrete solution of these anomaly equations.  In that example the determinant masses are identified with the total fundamental and anti-fundamental flavor masses, and the remaining spectator fugacities obey the balancing condition displayed in Section~\ref{sec:duality}.  The full duality statement and its JK-residue proof are given in the main text, so we do not repeat them in this appendix.

\section{Lower Dimensional Limits}\label{app:dimensional-limits}

\subsection{Three-dimensional limit}

Given the form of the determinant, the degeneration from $\cyl \times T^2$ to $\cyl \times S^1$ is obtained by taking $\operatorname{Im}\tau\to+\infty$, equivalently $p=e^{2\pi\ii\tau}\to0$. The corresponding 3d 1-loops then are, for the Dirichlet polarization,
\begin{align}
Z^{\mathrm{CM}}_{\mathrm{1\mbox{-}loop}}\big|_{DD}  =
\prod_{(m_\varphi,m_y)\in\mathbb{Z}^2}
\Big(\omega m_\varphi + m_y-q_G\gamma_G^{(3d)}-(r-2)\gamma_R^{(3d)}\Big).
\label{eq:ZrawD3d}
\end{align}
and,  for the Robin polarization,   the raw product reads
\begin{align}
Z^{\mathrm{CM}}_{\mathrm{1\mbox{-}loop}}\big|_{RR} 
 =
\prod_{(m_\varphi,m_y)\in\mathbb{Z}^2}
\Big(\omega m_\varphi + m_y+q_G\gamma_G^{(3d)}+r\gamma_R^{(3d)}\Big)^{-1} ,
\label{eq:ZrawR3d}
\end{align}
with
\begin{equation}
    \gamma_R^{(3d)} = \frac{1}{2}( \alpha_\varphi \omega + \alpha_y )
\end{equation}
and
\begin{equation}
    \gamma_G^{(3d)} = \mathcal A_0 + \beta_y + \omega \beta_\varphi .
\end{equation}

The products in \eqref{eq:ZrawD3d}--\eqref{eq:ZrawR3d} are over the full lattice $\mathbb{Z}^2$ and are understood in the Shintani--Barnes sense, i.e.\ by analytic continuation of the Barnes multiple zeta function and its associated multiple Gamma functions.

Introduce now the shifts
\begin{equation}
\mu^{(3d)}_{D}:=q_G\gamma_G^{(3d)}+(r-2)\gamma_R^{(3d)},
\qquad
\mu^{(3d)}_{R}:=q_G\gamma_G^{(3d)}+r\,\gamma_R^{(3d)}.
\label{eq:mu3d}
\end{equation}
The Robin raw product is the reflected product with shift \(-\mu_R^{(3d)}\).  Its literal Shintani--Barnes evaluation is therefore obtained by applying the formula below at \(\mu=-\mu_R^{(3d)}\) and then inverting the result.  In the displayed RR block we instead use the same contact-term normalization as in the four-dimensional block convention and write the answer with argument \(\mu_R^{(3d)}\).  The two expressions differ by a holomorphic non-vanishing contact factor, so we denote the normalized equality by \(\doteq\) rather than by \(\stackrel{\rm SB}{=}\).  The canonical factorization of the $\mathbb{Z}^2$ product is
\begin{equation}
\prod_{(m_\varphi,m_y)\in\mathbb{Z}^2}\big(\omega m_\varphi+m_y-\mu\big)
\ \stackrel{\rm SB}{=}\
C_{3d}(\omega)\,
\exp\!\Big(-\frac{\pi \ii}{2}\,B_{2,2}(\mu\mid\omega,1)\Big)\,
\Theta(\mu;\omega),
\label{eq:Z2reg}
\end{equation}
where $C_{3d}(\omega)$ depends only on $\omega$ (and on the chosen normalization of the Barnes double Gamma) and is independent of $\mu$. Here we defined the quadratic Bernoulli polynomial
\begin{align}
B_{2,2}(u\mid\omega_1,\omega_2) & \ 
=
\frac{u^2}{\omega_1\omega_2}
-\frac{\omega_1+\omega_2}{\omega_1\omega_2}\,u
+\frac{\omega_1^2+3\omega_1\omega_2+\omega_2^2}{6\,\omega_1\omega_2}~.
\label{eq:B22explicit}
\end{align}

Applying \eqref{eq:Z2reg} to \eqref{eq:ZrawD3d}--\eqref{eq:ZrawR3d} gives
\begin{align}
Z^{\mathrm{CM}}_{\mathrm{1\mbox{-}loop}}\big|_{DD}^{(3d)}
&\stackrel{\rm SB}{=}\
C_{3d}(\omega)\,
\exp\!\Big(-\frac{\pi \ii}{2}\,B_{2,2}\big(\mu^{(3d)}_{D}\mid\omega,1\big)\Big)\,
\Theta( \mu^{(3d)}_{D};\omega)~,
\label{eq:Z3dDDreg}\\[2mm]
Z^{\mathrm{CM}}_{\mathrm{1\mbox{-}loop}}\big|_{RR}^{(3d)}
&\doteq
C_{3d}(\omega)^{-1}\,
\exp\!\Big(+\frac{\pi \ii}{2}\,B_{2,2}\big(\mu^{(3d)}_{R}\mid\omega,1\big)\Big)\,
\Theta(\mu^{(3d)}_{R};\omega)^{-1}~.
\label{eq:Z3dRRreg}
\end{align}
The prefactor $C_{3d}(\omega)$ cancels in determinant ratios and in polarization changes where only $\mu$ varies.  With the normalization displayed above, the three-dimensional DD and RR blocks are inverse at the same argument; the reflected product differs by a holomorphic non-vanishing contact factor.  The functional form of these 3d one-loop determinants agrees with the bulk $I\times T^2$ interval-index blocks of Ref.~\cite{Sugiyama:2020uqh}, once the contact-term normalization and boundary-matter sectors are matched.  In this precise sense, the bulk one-loop sector of the construction of Ref.~\cite{Sugiyama:2020uqh} is contained in the present framework as the $\operatorname{Im}\tau\to+\infty$ degeneration of the four-dimensional cylinder kernel.  It is also consistent with the lower-dimensional hyperbolic-block limits studied in Ref.~\cite{Pittelli:2018rpl}.

\subsection{Two-dimensional limit}

The further degeneration from $\cyl \times S^1$ to $\cyl$ is obtained by taking $\operatorname{Im}\omega\to+\infty$, equivalently $q=e^{2\pi\ii\omega}\to0$, or by setting $\alpha_y=\beta_y=0$ and keeping only the corresponding zero modes in the infinite products. For instance, if we take the $\operatorname{Im}\omega\to+\infty$ limit we obtain, for the Dirichlet polarization,
\begin{align}
Z^{\mathrm{CM}}_{\mathrm{1\mbox{-}loop}}\big|_{DD}  =
\prod_{m_y\in\mathbb{Z}}
\Big( m_y-q_G\gamma_G^{(2d)}-(r-2)\gamma_R^{(2d)}\Big).
\label{eq:ZrawD2d}
\end{align}
and,  for the Robin polarization,   the raw product reads
\begin{align}
Z^{\mathrm{CM}}_{\mathrm{1\mbox{-}loop}}\big|_{RR} 
 =
\prod_{m_y\in\mathbb{Z}}
\Big( m_y+q_G\gamma_G^{(2d)}+r\gamma_R^{(2d)}\Big)^{-1} ,
\label{eq:ZrawR2d}
\end{align}
with
\begin{equation}
    \gamma_R^{(2d)} = \frac{1}{2} \alpha_y 
\end{equation}
and
\begin{equation}
    \gamma_G^{(2d)} = \mathcal A_0 + \beta_y .
\end{equation}

Instead, if we just set all $y$-components to zero, we obtain for the Dirichlet polarization,  
\begin{align}
Z^{\mathrm{CM}}_{\mathrm{1\mbox{-}loop}}\big|_{DD}  =
\prod_{m_\varphi \in\mathbb{Z}}
\Big( m_\varphi -q_G\gamma_G^{(2d)}-(r-2)\gamma_R^{(2d)}\Big).
\label{eq:ZrawD2dbis}
\end{align}
and,  for the Robin polarization,   the raw product reads
\begin{align}
Z^{\mathrm{CM}}_{\mathrm{1\mbox{-}loop}}\big|_{RR} 
 =
\prod_{m_\varphi\in\mathbb{Z}}
\Big( m_\varphi +q_G\gamma_G^{(2d)}+r\gamma_R^{(2d)}\Big)^{-1} ,
\label{eq:ZrawR2dbis}
\end{align}
with
\begin{equation}
    \gamma_R^{(2d)} = \frac{1}{2} \alpha_\varphi 
\end{equation}
and
\begin{equation}
    \gamma_G^{(2d)} = \frac{\mathcal A_0}{\omega} + \beta_\varphi = \widehat{\mathcal A}_0 + \beta_\varphi .
\end{equation}

\subsection[The Z-product and the cylinder determinant]{The \texorpdfstring{$\mathbb{Z}$}{Z}-product and the cylinder determinant}

The $\mathbb{Z}$-products in \eqref{eq:ZrawD2d}--\eqref{eq:ZrawR2d} and in \eqref{eq:ZrawD2dbis}--\eqref{eq:ZrawR2dbis} are regularized by the Barnes one-Gamma function (equivalently, by the Hurwitz zeta prescription).

For $\mu\in\mathbb{C}$ one has the canonical Shintani--Barnes identity
\begin{equation}
\prod_{m\in\mathbb{Z}}(m-\mu)
\ \stackrel{\rm SB}{=}\
C_{2d}\,
\exp\!\big(-\pi \ii\,B_{1,1}(\mu\mid 1)\big)\,
\big(1-e^{2\pi \ii \mu}\big),
\label{eq:Z1reg}
\end{equation}
where $C_{2d}$ is a $\mu$-independent constant fixed by the normalization of $\Gamma_1$ (equivalently by the choice of zeta subtraction). In particular, \eqref{eq:Z1reg} is equivalent to the familiar form $2\sin(\pi\mu)$ up to a $\mu$-independent factor, and exhibits explicitly the phase defined by
the linear Bernoulli polynomial 
\begin{equation}
B_{1,1}(u\mid\omega)=\frac{u}{\omega}-\frac12~.
\label{eq:B11}
\end{equation}

Define now the shifts
\begin{equation}
\mu^{(2d)}_{D}:=q_G\gamma_G^{(2d)}+(r-2)\gamma_R^{(2d)},
\qquad
\mu^{(2d)}_{R}:=q_G\gamma_G^{(2d)}+r\,\gamma_R^{(2d)}.
\label{eq:mu2d}
\end{equation}
As in the three-dimensional limit, the Robin product is evaluated in the contact normalization compatible with \(z_{\rm DD}^{2d}(u)z_{\rm RR}^{2d}(-u)=1\).  Applying \eqref{eq:Z1reg} yields, for the $m_y\in\mathbb{Z}$ reduction \eqref{eq:ZrawD2d}--\eqref{eq:ZrawR2d},
\begin{align}
Z^{\mathrm{CM}}_{\mathrm{1\mbox{-}loop}}\big|_{DD}^{(2d)}
&\stackrel{\rm SB}{=}\
C_{2d}\,
\exp\!\big(-\pi \ii\,B_{1,1}(\mu^{(2d)}_{D}\mid 1)\big)\,
\big(1-e^{2\pi \ii \mu^{(2d)}_{D}}\big),
\label{eq:Z2dDDreg}\\[2mm]
Z^{\mathrm{CM}}_{\mathrm{1\mbox{-}loop}}\big|_{RR}^{(2d)}
&\stackrel{\rm SB}{=}\
C_{2d}^{-1}\,
\exp\!\big(+\pi \ii\,B_{1,1}(\mu^{(2d)}_{R}\mid 1)\big)\,
\big(1-e^{2\pi \ii \mu^{(2d)}_{R}}\big)^{-1}.
\label{eq:Z2dRRreg}
\end{align}
The same formulas apply verbatim to the alternative reduction \eqref{eq:ZrawD2dbis}--\eqref{eq:ZrawR2dbis}, with $\mu^{(2d)}_{D},\mu^{(2d)}_{R}$ computed using the corresponding $(\gamma_R^{(2d)},\gamma_G^{(2d)})$ in that branch. As in three dimensions, the constants $C_{2d}^{\pm 1}$ cancel in determinant ratios.


\end{document}